\documentclass[12pt,reqno]{article}
\usepackage[english]{babel}
\usepackage{amsmath,amsthm,commath,mathrsfs,amssymb,extarrows}
\usepackage{dsfont}
\usepackage{relsize}

\usepackage{pgfplots}

\usepackage{tikz}
\usetikzlibrary{arrows,intersections}

\usepackage[colorlinks = true,urlcolor = blue,citecolor = blue,breaklinks]{hyperref}

\usepackage{amsfonts}
\usepackage{graphicx}
\usepackage{enumerate}
\usepackage{subcaption}
\usepackage[right,pagewise,displaymath, mathlines]{lineno}
\usepackage{epstopdf}
\usepackage{color}
\usepackage{multirow}
\usepackage{authblk}
\usepackage[round]{natbib}
\usepackage{float}
\restylefloat{figure,table}

\usepackage{booktabs}

\usepackage{colortbl}

\usepackage{geometry}
\numberwithin{equation}{section}
\numberwithin{figure}{section}
\numberwithin{table}{section}

\newtheorem{theorem}{Theorem}[section]
\newtheorem{corollary}{Corollary}[section]
\newtheorem{lemma}{Lemma}[section]

\theoremstyle{definition}
\newtheorem{definition}{Definition}[section]

\newtheorem{note}{Note}[section]

\newtheorem{question}{Question}[section]

\numberwithin{equation}{section}

\definecolor{darkread}{rgb}{0.7, 0, 0}
\definecolor{darkbrown}{rgb}{0.55, 0.2, 0.15}
\definecolor{darkblue}{rgb}{0.1,0.1,0.6}
\definecolor{darkgreen}{rgb}{0.1,0.5,0.2}

\DeclareMathOperator*{\Var}{Var}
\DeclareMathOperator*{\DJI}{DJ}
\DeclareMathOperator*{\AO}{AO}

\newcommand{\dd}{\mathrm{d}}

\title{\Large\bf Detecting systematic anomalies affecting systems when inputs are stationary time series\thefootnote\relax\footnotetext{
We are indebted to Editor-in-Chief Fabrizio Ruggeri, an Associate Editor, and two anonymous referees for their most careful reading of our paper, constructive criticism, and numerous questions, comments and suggestions.
We are grateful to Yuri Davydov, Nadezhda Gribkova, Hong Li, Raghu Pasupathy, Jiandong Ren, Jianxi Su, Ruodu Wang, Aaron Nung Kwan Yip, and all the participants of the Risk Management and Actuarial Science Seminar (University of Waterloo and Tsinghua University), Actuarial Seminar (University of Wisconsin-Milwaukee), and the Industrial Mathematics and Statistics Seminar (Purdue University) for discussions and suggestions. Our research has been supported by the Natural Sciences and Engineering Research Council (NSERC) of Canada, and the national research organization Mathematics of Information Technology and Complex Systems (MITACS) of Canada.}}

\author[1]{Ning Sun}
\author[,1,2]{Chen Yang \thanks{Corresponding author; e-mail \href{mailto:cyang244@whu.edu.cn}{cyang244@whu.edu.cn} }} 
\author[1,3]{Ri\v cardas Zitikis}
\affil[1]{\normalsize School of Mathematical and Statistical Sciences, Western University, \break London, Ontario N6A 5B7, Canada}
\affil[2]{\normalsize Economics and Management School, Wuhan University, \break Wuhan, Hubei 430072, P.~R.~China}
\affil[3]{\normalsize Risk and Insurance Studies Centre, York University,  \break Toronto, Ontario M3J 1P3, Canada}

\date{}

\usepackage{setspace}
\usepackage[bottom]{footmisc}

\begin{document}
%\linenumbers

\maketitle
\vspace{-10mm}

\noindent
{\bf Abstract.}
We develop an anomaly-detection method when systematic anomalies, possibly statistically very similar to genuine inputs, are affecting control systems at the input and/or output stages. The method allows anomaly-free inputs (i.e., those before contamination) to originate from a wide class of random sequences, thus opening up possibilities for diverse applications. To illustrate how the method works on data, and how to interpret its results and make decisions, we analyze several actual time series, which are originally non-stationary but in the process of analysis are converted into stationary. As a further illustration, we provide a controlled experiment with anomaly-free inputs following an ARMA time series model under various contamination scenarios.
\smallskip


\newpage

%\tableofcontents

\section{Introduction}
\label{intro}

Control systems are often exposed to errors, intrusions, and other anomalies whose detection in a timely fashion is of paramount importance.
Computer systems monitor and control a myriad of physical processes, and their protection against random errors, deliberate intrusions \cite[e.g.,][]{d1987,ddw1999,calhhs2011,pak2013}, false data injections \cite[e.g.,][]{lzlwd2017}, and other disruptors is of much interest.

A vast number of methods have been proposed for the purpose. For example, we find methods based on deep learning \cite[e.g.,][]{hmw2017}, probabilistic arguments \cite[e.g.,][]{htclch2016,o2016}, artificial neural networks \cite[e.g.,][]{pds2017}, and Fourier techniques \cite[e.g.,][]{z2018}. \cite{CXS2018} discuss the effects of an early warning mechanism on system's reliability. For a recently developed LSTM-based intrusion detection system for in-vehicle can bus communications, we refer to \cite{hiofk2020}.

For complementary reviews on anomaly detection, we refer to \cite{cbk2009},   \cite{BBK2014}, and \cite{f2020}. For general information on various facets of risk and with them associated problems, we refer to, e.g., \cite{ABFZ2014} and \cite{Z2018}. For more specialized discussions on the topic, we refer to, e.g., \cite{CLY2017} and \cite{lzlwd2017}.

The emphasis in the present paper is on temporal aspects and dependence structures that arise in this area of research. There have been a number of studies tackling these issues from several perspectives. For example, \cite{BP1996} present a computational procedure capable of robust and sensitive statistical detection of deterministic and chaotic dynamics in  short and noisy time series. \cite{HXXZ2017} explore the role of dependence when assessing quantities such as system-compromise probabilities and the cost of attacks, which are then used to  develop optimization strategies. \cite{DL2018} tackle the problem of assessing whether temporal clusters in randomly occurring sequences of events are genuinely random.

Furthermore, \cite{fef2018} propose what is called the collective and point anomalies (CAPA) method  for detecting point anomalies (i.e., outliers in the statistical language) as well as anomalous segments, or collective anomalies. The method is suitable when collective anomalies are characterised by either a change in mean, variance, or both, ant it is capable of distinguishing collective anomalies from point anomalies. This and several other methods have been implemented in an R package by \cite{fgefb2020b}, where we also find
the multi-variate collective and point anomaly (MVCAPA) method of \cite{fef2019},
the proportion adaptive segment selection (PASS) method of \cite{jcl2013},
the Bayesian abnormal region detector (BARD) of \cite{bf2017},
and also sequential versions of CAPA and MVCAPA by \cite{fbe2020a}.
\cite{f2020} provides the state of the art on statistical anomaly detection, together with a guide for computational implementation.

The present paper is devoted to another anomaly-detection method that works irrespective of whether systems are being affected at the input or output stage, or at both stages simultaneously.
The important feature that distinguishes our method from the earlier ones is that it can detect persistent anomalies that may not change the regime (e.g., mean, variance, and/or autocorrelations) of data in an abrupt fashion during the period of observation. Hence, those statistical techniques that have been designed to detect outliers and other aberrations become ineffective in such situations.

As in many previous studies, we also consider dependent random inputs and, in turn, dependent outputs. This enables us to use the method in a myriad of applications. We have carefully proven the underlying theoretical results and illustrated the method using stationary time series under various contamination by anomalies scenarios. It is useful to recall at this point that historical data as well as  subject-matter knowledge are helpful in deciding how to reduce non-stationary random sequences to stationary ones, and transformations such as differencing and de-periodization can especially be helpful \citep[see, e.g.,][]{BJRL2015,BD1991}. We shall rely on such transformations in our real-world illustrative examples in Section~\ref{illustrate}.

The departure from the earlier explored by \cite{GZ2020} case of independent and identically distributed (iid) inputs to the herein tackled dependent random inputs and thus outputs requires considerable technical innovation and have given rise to notions such as $p$-reasonable order and temperate dependence, whose connections to classical notions such as phantom distributions have been illuminated. We note that the just mentioned parameter $p$ is related the $p$-th finite moment of inputs, and thus to the tail heaviness of the input distribution.

The rest of the paper is organized as follows.
In Section~\ref{notation} we introduce and discuss basic notation.
In Section~\ref{illustrate} we analyze two actual examples that illustrate the anomaly-detection method what we develop in subsequent sections.
In Section~\ref{avr} we introduce an experiment that further illustrates and guides our technical considerations.
In Section~\ref{prelim} we lay out a foundation for our anomaly-detection method.
In Section~\ref{expdesign} we illustrate the performance of the method graphically.
In Section~\ref{orderly systems} we explain how the method acts in anomaly-free orderly systems, whereas in Section~\ref{disorderly systems} we show how the method detects anomalies when they are present.
Section~\ref{conclude} concludes the paper with a brief summary of main results and several suggestions for future studies. Although some graphical illustrations are already given in the main body of the paper, Appendix~\ref{graphs} contains more extensive illustrations. Technical details such as lemmas and proofs are in Appendix~\ref{proofs}.

\section{Setting the stage: basic notation}
\label{notation}

Throughout the paper, we assume the existence of a function $h:\mathbb{R}^{1+d}\to \mathbb{R}$, called transfer function, that connects inputs $X_t\in \mathbb{R}$ and outputs $Y_t\in \mathbb{R}$ via the equation
\begin{equation}\label{model}
Y_t=h(X_t,\boldsymbol{\varepsilon}_t),
\end{equation}
where $ \boldsymbol{\varepsilon}_t\in \mathbb{R}^d$ are $d$-dimensional exogenous random variables, called anomalies. They can of course be equal to  $\mathbf{0}:=(0,\dots , 0)\in \mathbb{R}^{d}$, meaning that the system is free of anomalies. In this case the transfer function reduces to $h_0:\mathbb{R}\to \mathbb{R}$ defined by
\[
h_0(x)=h(x,\mathbf{0}),
\]
which we call the baseline function. When we wish to emphasize that outputs $Y_t$ arise from this anomaly-free case, we use the notation
\begin{equation}\label{error-free}
Y^0_t=h_0(X_t).
\end{equation}

To illustrate, let $d=2$, in which case we have  $\boldsymbol{\varepsilon}_t=(\varepsilon_{1,t},\varepsilon_{2,t})$. We may think of  $\varepsilon_{1,t}$ as anomalies affecting the inputs $X_t$ before they enter the control system, and $\varepsilon_{2,t}$ as anomalies affecting the (already affected) inputs when they exit the system. Thinking in this fashion, we arrive at the following transfer functions $h:\mathbb{R}^{3}\to \mathbb{R}$, which we use in our numerical experiment later in the paper:
\begin{enumerate}[TF1:]
  \item\label{tf1} $h(x,y,0):=h_0(x+y)$ when the system is affected by anomalies only at the input stage;
  \item\label{tf2}  $h(x,0,z):=h_0(x)+z$ when the system is affected by anomalies only at the output stage;
  \item\label{tf3}  $h(x,y,z):=h_0(x+y)+z$ when the system is affected by anomalies at the input and output stages.
\end{enumerate}

Besides the additive model, there are other models and thus other transfer functions that link inputs with exogenous variables \cite[e.g.,][]{Finkelshtainetal(1999),Frankeetal(2006), Frankeetal(2011),Guoetal(2018)}. Arguments in favour of using one model over another can be found in studies by, e.g.,  \cite{Peroteetal(2015)}, \cite{Su(2016)}, \cite{Semenikhineetal(2018)}, \cite{Guoetal(2018)}, and \cite{Guoetal(2019)}.

Model~\eqref{model} arises in many areas, including regression analysis, classification, and, generally, in machine learning \citep[e.g.,][]{htf2009}. It also relates our research to the so-called strategy-proof estimation in regression \cite[e.g.,][]{PeroteandPerote-Pena(2004),Peroteetal(2015)}.

\begin{note}
Visually, model~\eqref{model} may give the impression that the outputs depend only on the current value of inputs, but $X_t$, at least in the case of causal time series, is a linear combination of the contemporary and historical values of the underlying white noise. That is, $X_t$ is the inner product $\langle \boldsymbol{\beta}, \mathbf{Z}_t \rangle $ of a (finite or infinite) sequence $\boldsymbol{\beta}=(\beta_i)_{i\ge 0}$ of parameters and  a (finite or infinite) sequence $\mathbf{Z}_t=(Z_{t-i})_{i\ge 0}$ of uncorrelated random variables. We shall elaborate more on this topic in Section~\ref{conclude},  in the context of potential future work.
\end{note}

Although the function $h$ might be known to, e.g., the control system's manufacturer, its precise formula may not be known to those working in the area of anomaly detection (e.g., company's IT personnel). The transfer function might even deviate from its original specifications due to, e.g., wear and tear. Furthermore, in the context of, say, economic variables, which we shall encounter in the next section, their relationships might be postulated by academics but in actuality, the true relationships (i.e., the transfer mechanisms from one to another) usually deviate from any model. In addition, the relationships might be, and usually are, affected by exogenous economic and other variables. Hence, to accommodate various scenarios associated with model uncertainty, we shall aim at deriving results for very large classes of transfer functions, that is, under very mild assumptions.

\section{Two actual illustrations}
\label{illustrate}

This section is devoted to two real-world illustrations, which make up the contents of the following two subsections. The illustrations are based on pairs $(X_t,Y_t)$ of economic variables, which are  observed only for $1\le t \le n$ for some sample sizes $n$. The inputs and outputs are dependent. With some luck, one of these variables can be assessed from the values of another variable, although not precisely because the transfer mechanism (i.e., the transfer function $h$) is not known, except of course in academic models. This, however, is not of concern to us because our primary interest is in finding out whether exogenous economic or other variables are  systematically affecting the relationship between $X_t$ and $Y_t$. That is, we want to answer the following question:
\begin{question}\label{question}
Are there $ \boldsymbol{\varepsilon}_t$'s in model~\eqref{model}?
\end{question}

At this point, we may instinctively start to debate as to the extent of smoothing of the scatterplot $(X_1,Y_1), \dots , (X_n,Y_n)$, assuming that we want to do it: extreme undersmoothing would result in a wiggly function $h$ with no errors $ \boldsymbol{\varepsilon}_t$, whereas too much smoothing would result in a nice function $h$ but with large errors $ \boldsymbol{\varepsilon}_t$.
Hence, the researcher's subjectively chosen level of smoothing determines whether or not there are errors in the model, and how large they are. We therefore do not do any smoothing. Our task is to find out if the hypothetical transfer function from one economic variable to another is affected by exogenous systematic variables, whatever they might be.

\subsection{Dow Jones and Australian All Ordinaries Indices}

The data \citep[Example~8.1.1]{BD2016} consist of the closing values of the Dow Jones Index (DJ) and the Australian
All Ordinaries Index of Share Prices (AO) recorded at the termination of trading on 251
successive trading days up to August 26, 1994. From the original data, we calculate the percentage  relative price changes, known as percentage returns, and plot them in Figure~\ref{RDJandRAO}.
\begin{figure}[h!]
    \centering
    \begin{subfigure}[b]{0.39\textwidth}
        \includegraphics[width=\textwidth]{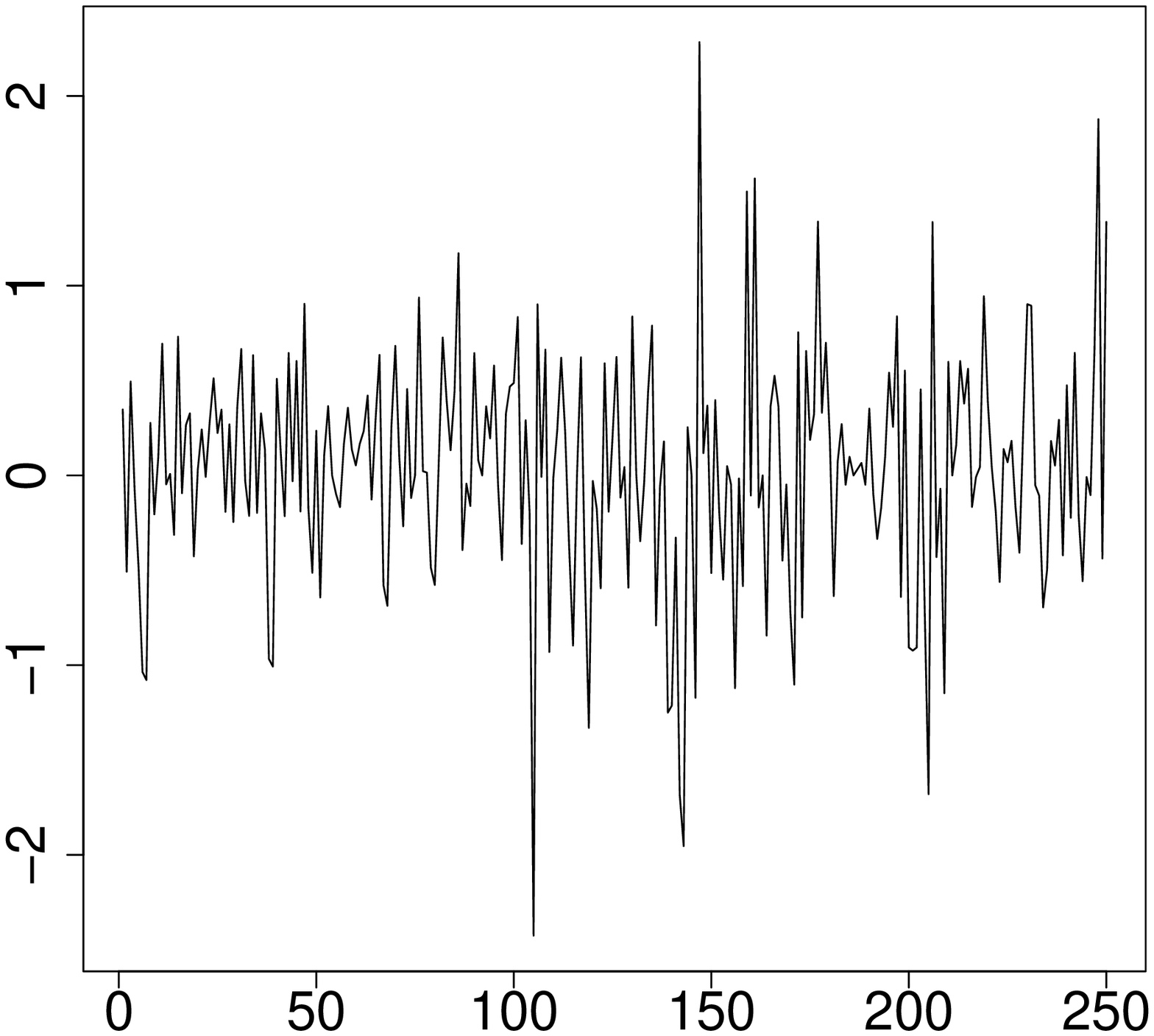}
        \caption{DJ.}\label{RDJ}
    \end{subfigure}
\hspace{10mm}
    \begin{subfigure}[b]{0.39\textwidth}
        \includegraphics[width=\textwidth]{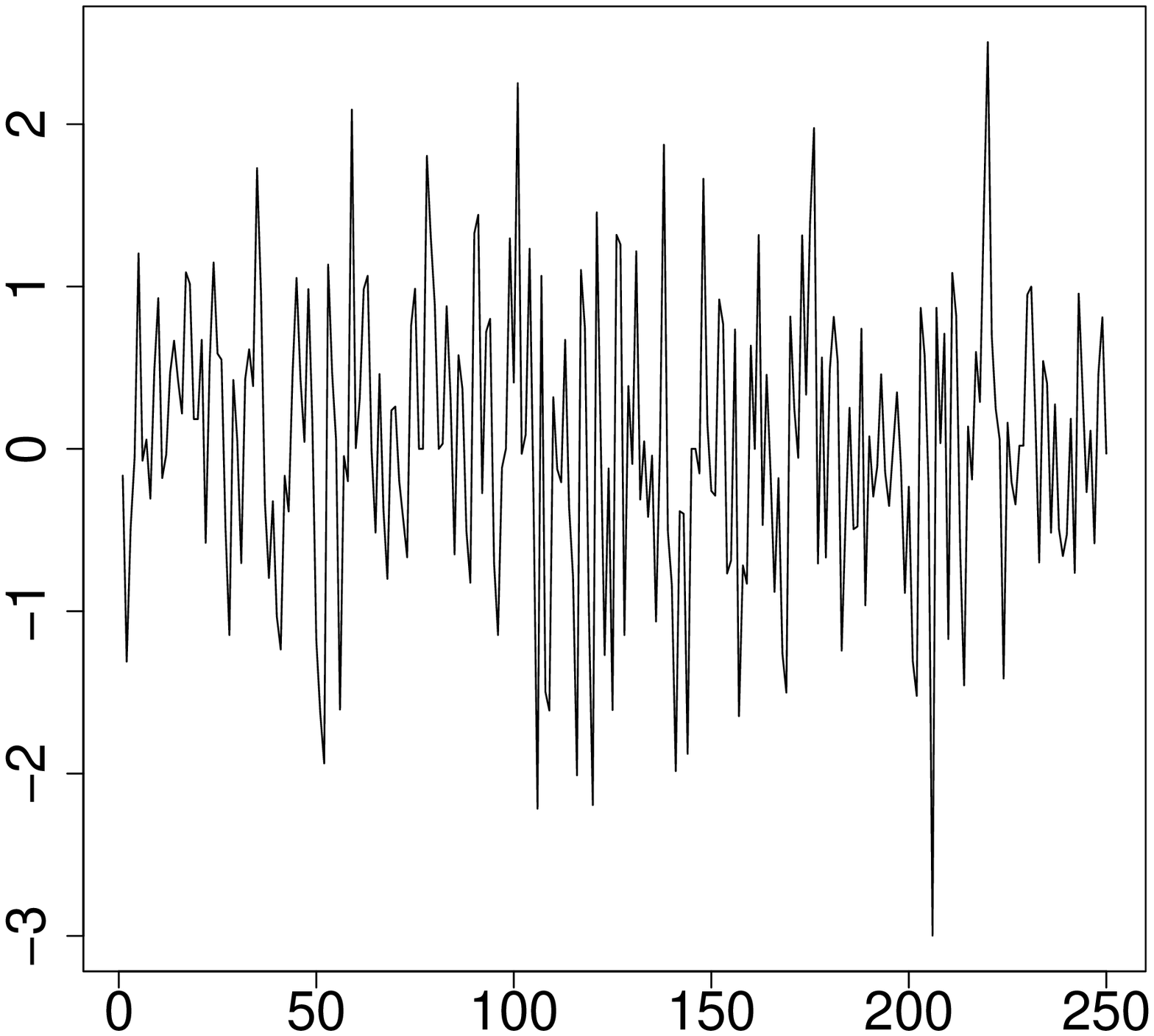}
        \caption{AO.}\label{RAO}
    \end{subfigure}
    \caption{The percentage returns of DJ and AO.}\label{RDJandRAO}
\end{figure}

To answer Question~\ref{question}, we employ an index $I_n$, whose mathematical definition will be introduced in Section~\ref{prelim}. At the moment, what really matters and interests us are the conclusions that we can reach, for which we use the following decision rules:
\begin{description}
\item[Decision 1:]
If the transition from inputs to outputs is accomplished without systematic interference, then, when the sample size $n$ grows, the index $I_n$ stays away from $1/2$.
\item[Decision 2:]
If, however, the transition is exposed to systematic interference, then, when the sample size $n$ grows, the index $I_n$ tends to $1/2$.
\end{description}

Due to limited sample sizes $n$ or some other reasons, it may not always be clear whether or not the index $I_n$ tends to $1/2$. In such cases we additionally calculate another index, denoted by $B_{n,2}$, whose mathematical definition will be given in Section~\ref{prelim}.
The meaning of the index $B_{n,2}$ relies on its growth to infinity, and it supplements Decisions~1 and~2 in the following way:
\begin{description}
\item[Supplement 1:]
If the transition from inputs to outputs is accomplished without systematic interference, then, when the sample size $n$ grows, the index $B_{n,2}$ stays asymptotically bounded, that is,  $B_{n,2}=O_{\mathbb{P}}(1)$ in mathematical terms.
\item[Supplement 2:]
If, however, the transition is exposed to systematic interference, then, when the sample size $n$ grows, the index $B_{n,2}$ tends to infinity.
\end{description}

Equipped with these indices $I_n$ and $B_{n,2}$, we can now look at the closing values of DJ and AO. The first question that arises is which of the two variables should be the ``input.'' The answer is naturally related to causality, but to avoid any prejudicial statement and thus controversy, we do our analysis both ways:  first we take DJ as the input and thus AO as the output, and then interchange their roles. The two cases with their respective indices $I_n$ and $B_{n,2}$ are visualized in Figure~\ref{DJandAOindices}.
\begin{figure}[h!]
    \centering
      %(or a blank line to force the subfigure onto a new line)
    \begin{subfigure}[b]{0.39\textwidth}
        \includegraphics[width=\textwidth]{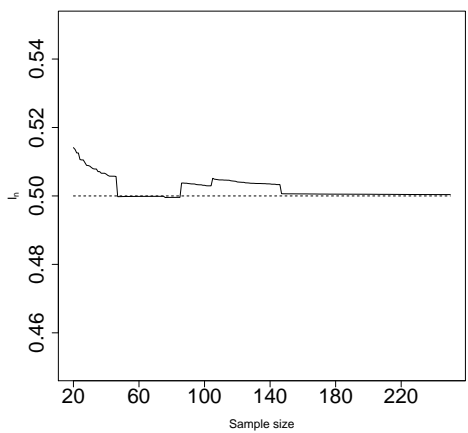}
        \caption{$I_n$ when $(X,Y)=(\DJI,\AO)$.}\label{YAOI}
    \end{subfigure}
\hspace{10mm}
    %(or a blank line to force the subfigure onto a new line)
    \begin{subfigure}[b]{0.39\textwidth}
        \includegraphics[width=\textwidth]{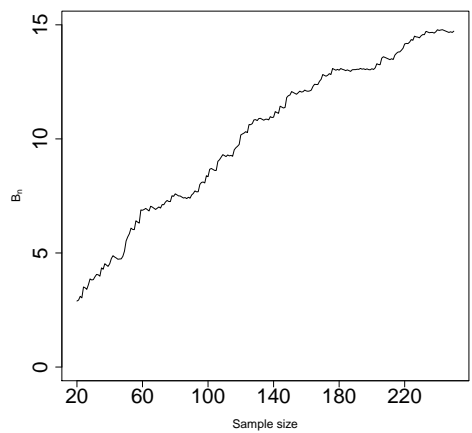}
        \caption{$B_{n,2}$ when $(X,Y)=(\DJI,\AO)$.}\label{YAOB}
    \end{subfigure}
\hspace{10mm}
      %(or a blank line to force the subfigure onto a new line)
    \begin{subfigure}[b]{0.39\textwidth}
        \includegraphics[width=\textwidth]{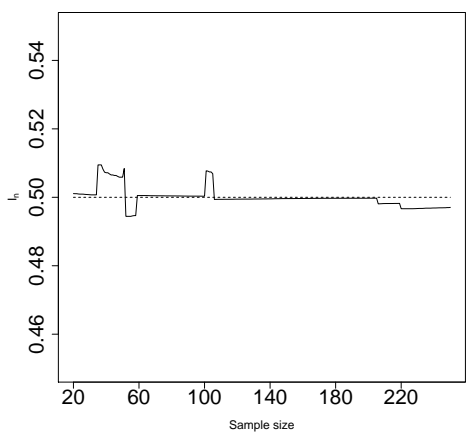}
        \caption{$I_n$ when $(X,Y)=(\AO,\DJI)$.}\label{YDJI}
    \end{subfigure}
\hspace{10mm}
    %(or a blank line to force the subfigure onto a new line)
    \begin{subfigure}[b]{0.39\textwidth}
        \includegraphics[width=\textwidth]{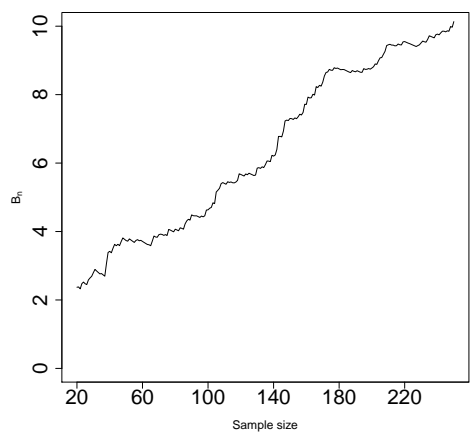}
        \caption{$B_{n,2}$ when $(X,Y)=(\AO,\DJI)$.}\label{YDJB}
    \end{subfigure}
    \caption{The indices $I_n$ and $B_{n,2}$ corresponding to DJ and AO with respect to the sample sizes $n=20,\ldots,149$.}\label{DJandAOindices}
\end{figure}

The graphs suggest that there is exogenous interference when transferring DJ to AO, and also the other way around, although there is a little dip below $1/2$ on the right-hand side of Figure~\ref{YDJI}, which may not be of importance given its small value. The index $B_{n,2}$ sends the same message as $I_n$. Hence, we comfortably conclude the existence of interference, although more data might overturn the conclusion.

Note that in Figure~\ref{DJandAOindices} we always start graphing the panels at $n=20$. This is so because for small values of $n$, the index $I_n$ fluctuates wildly between $0$ and $1$, as it should, which will be clearly seen from the mathematical definition of $I_n$. Hence, by starting at $n=20$, we are able to better depict the behaviour of $I_n$ near $1/2$, which is what really matters for our anomaly-detection method.

\subsection{Sales with a leading indicator}

The data \citep[Example~8.1.2]{BD2016} consist of 150-day sales with a leading indicator, plotted in Figure~\ref{LS origin}.
\begin{figure}[h!]
    \centering
    \begin{subfigure}[b]{0.39\textwidth}
        \includegraphics[width=\textwidth]{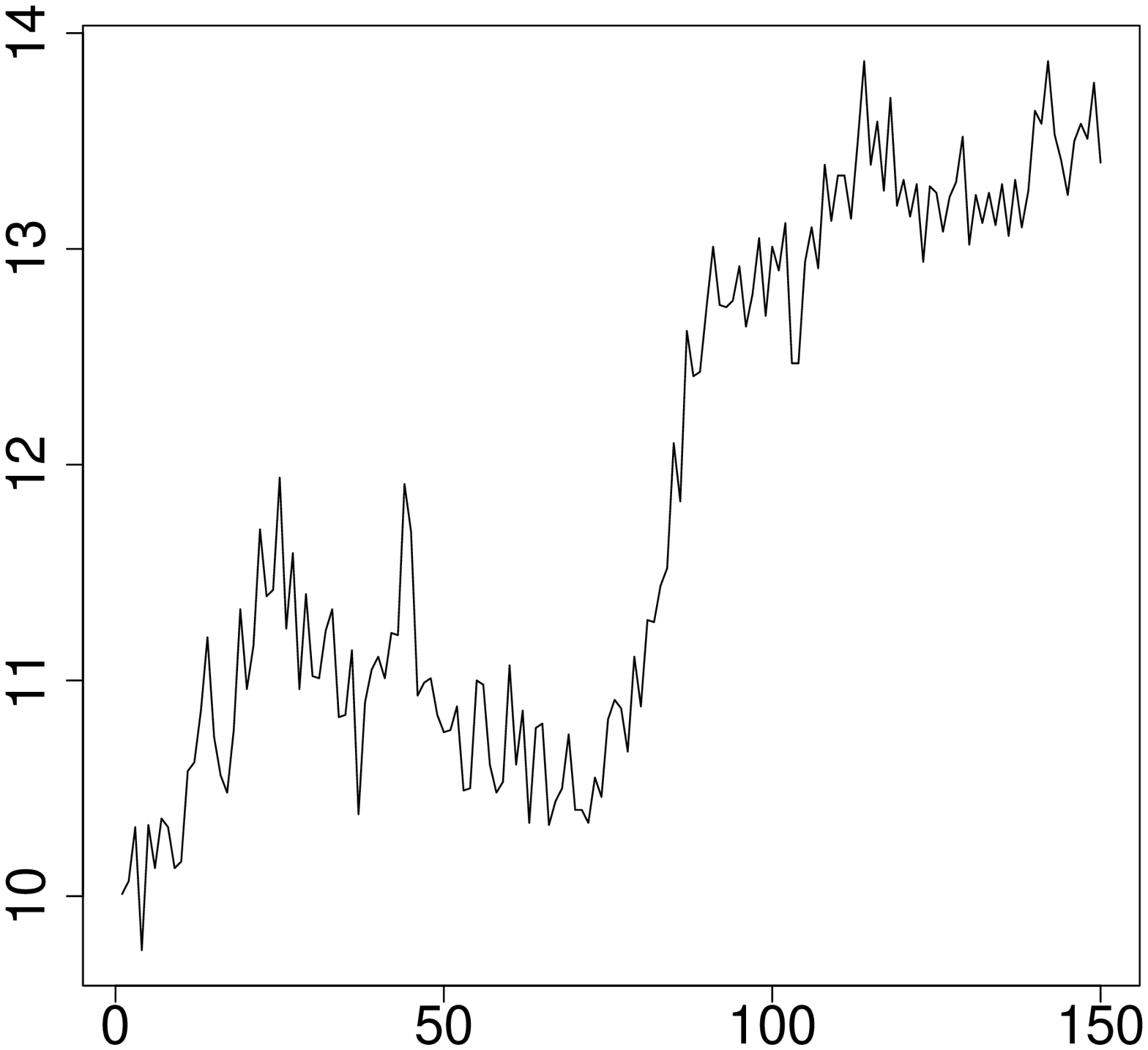}
        \caption{The leading indicator}
    \end{subfigure}
\hspace{10mm}
    \begin{subfigure}[b]{0.39\textwidth}
        \includegraphics[width=\textwidth]{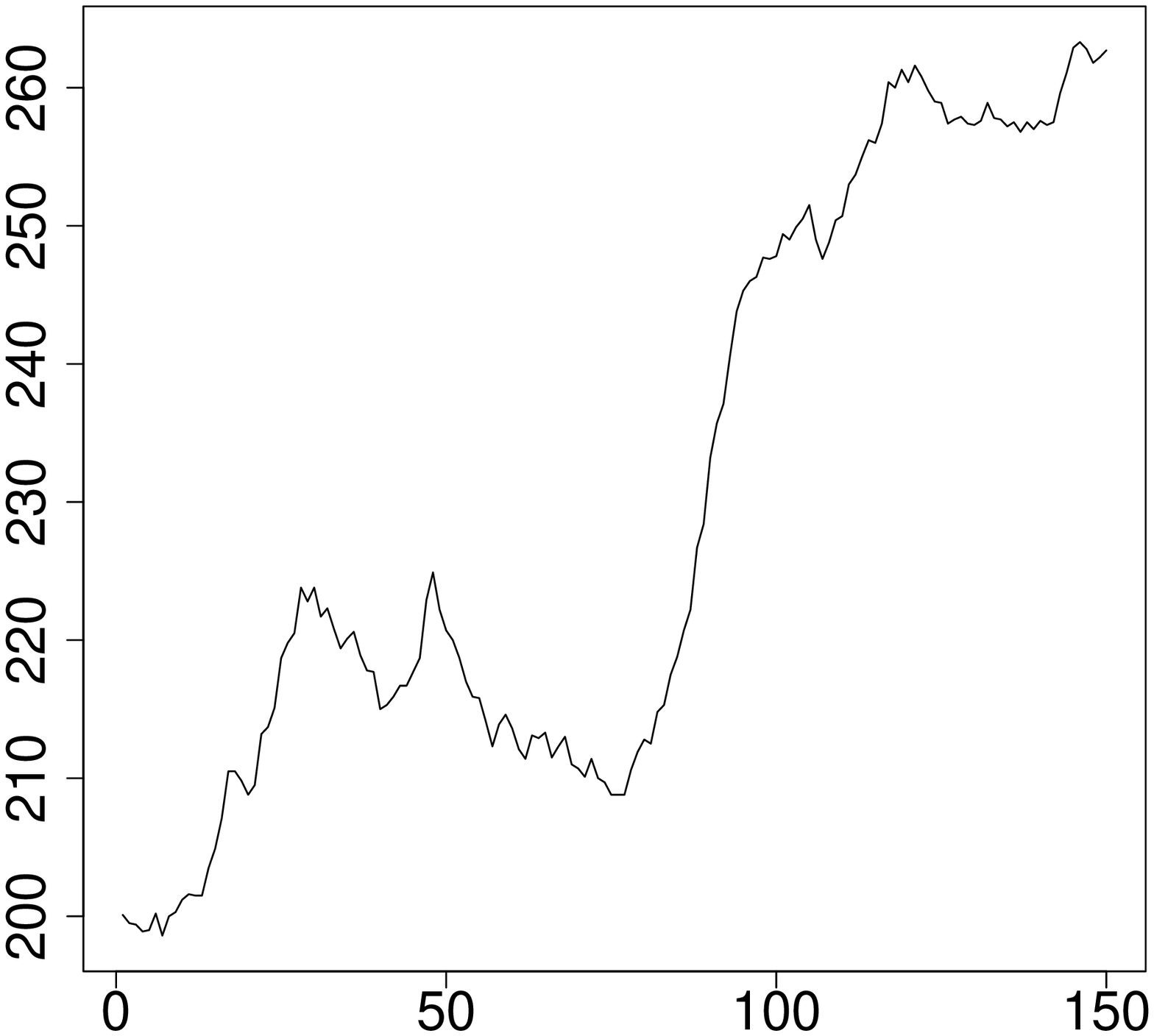}
        \caption{The sales}
    \end{subfigure}
    \caption{The original 150-day data of the leading indicator and the sales.}\label{LS origin}
\end{figure}
The data are non-stationarity, and so we difference it at lag~$1$. The transformed data are plotted in Figure~\ref{LS diff}.
\begin{figure}[h!]
    \centering
    \begin{subfigure}[b]{0.39\textwidth}
        \includegraphics[width=\textwidth]{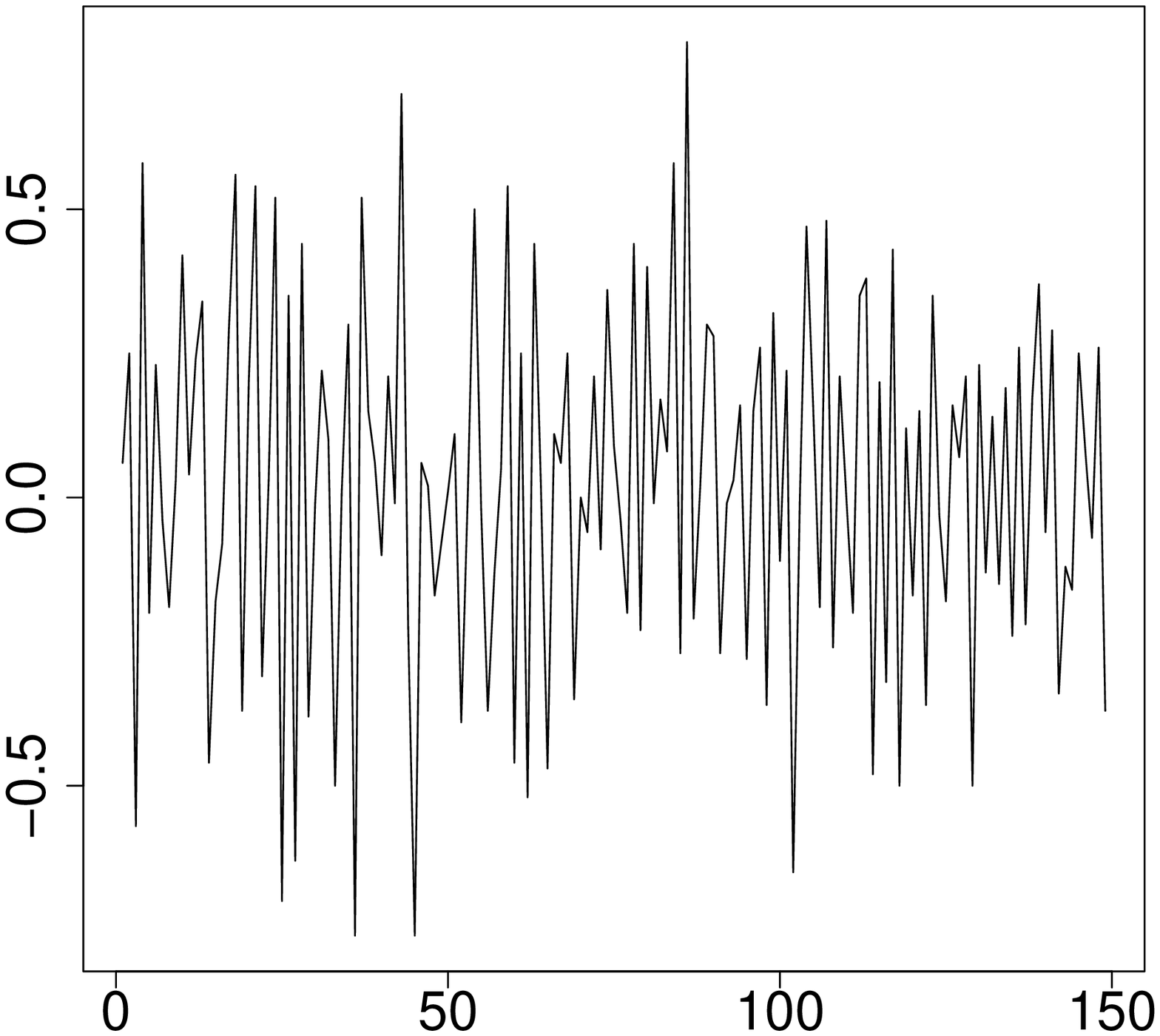}
        \caption{The leading indicator}
    \end{subfigure}
\hspace{10mm}
    \begin{subfigure}[b]{0.39\textwidth}
        \includegraphics[width=\textwidth]{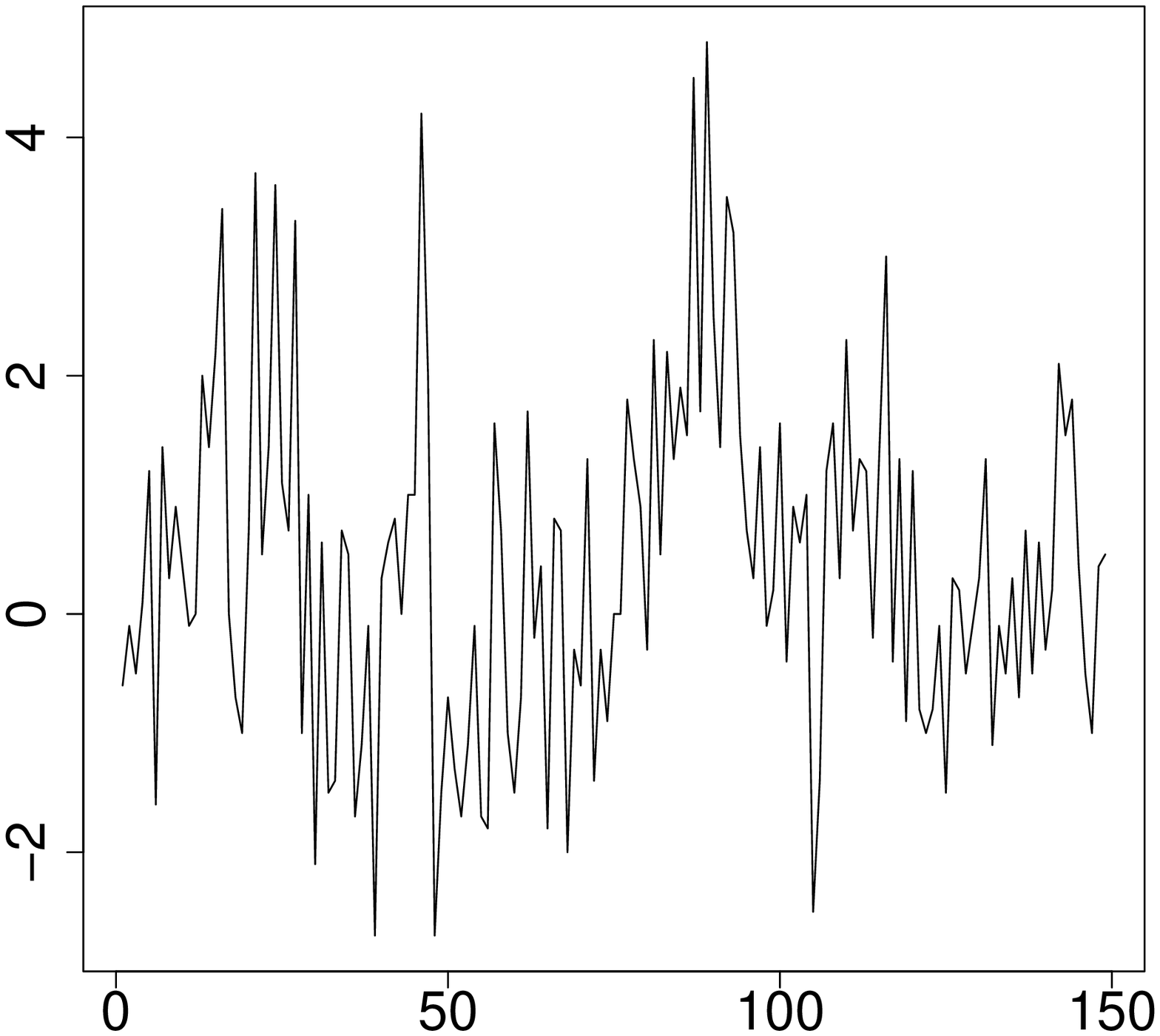}
        \caption{The sales}
    \end{subfigure}
    \caption{The 1-lag differences of the leading indicator and the sales.}\label{LS diff}
\end{figure}
We set the differenced leading indicator as the input and the differenced sales as the output. There are two reasons for this choice: first, it makes economic sense, and second, the differenced leading indicator exhibits stationarity whereas differenced sales seem to hint at some periodicity. Having thus made these choices, we next calculate the indices $I_n$ and $B_{n,2}$, which are depicted in  Figure~\ref{LS I B}.
\begin{figure}[h!]
    \centering
    \begin{subfigure}[b]{0.39\textwidth}
        \includegraphics[width=\textwidth]{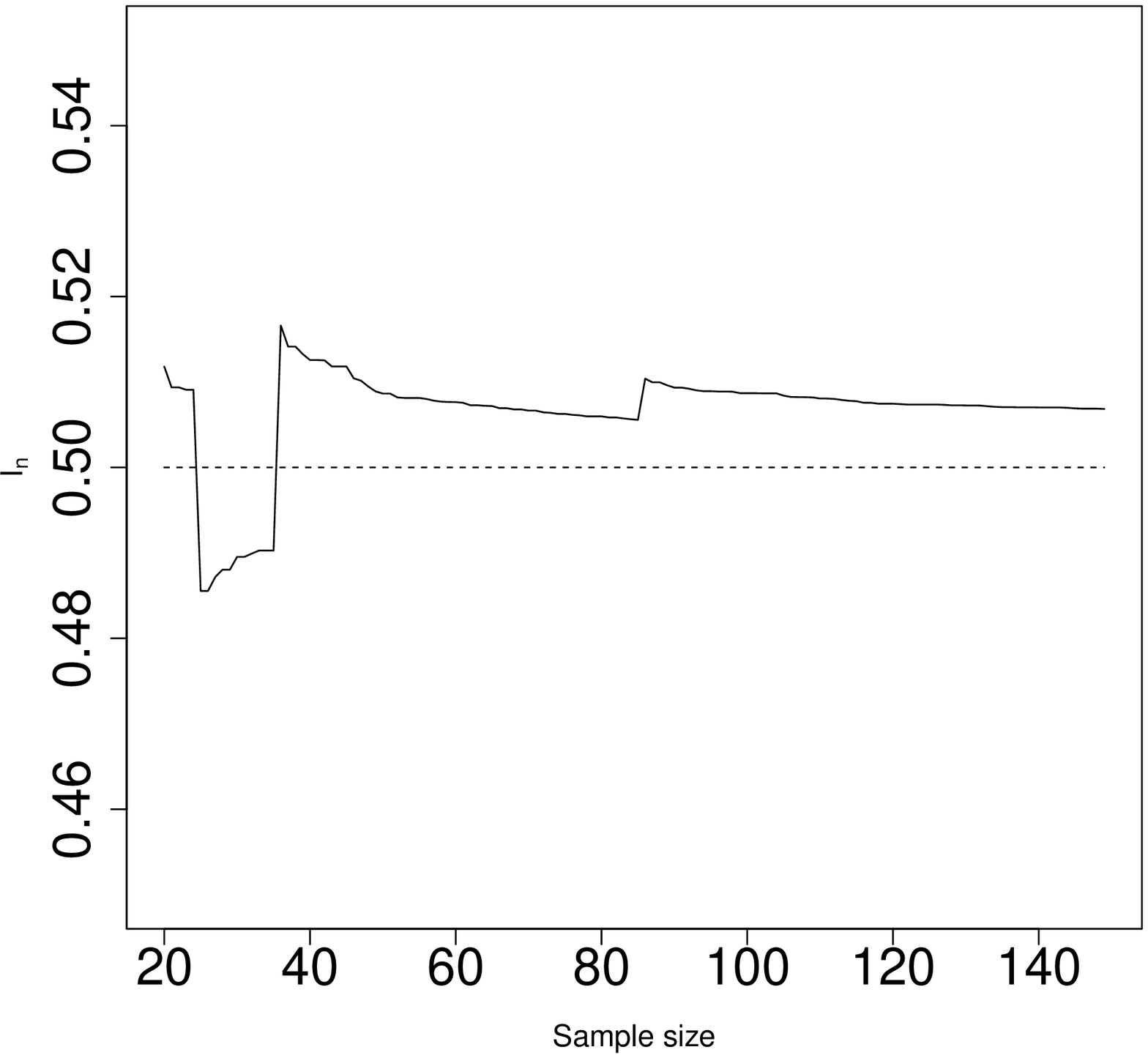}
        \caption{$I_n$.}
    \end{subfigure}
\hspace{10mm}
    \begin{subfigure}[b]{0.39\textwidth}
        \includegraphics[width=\textwidth]{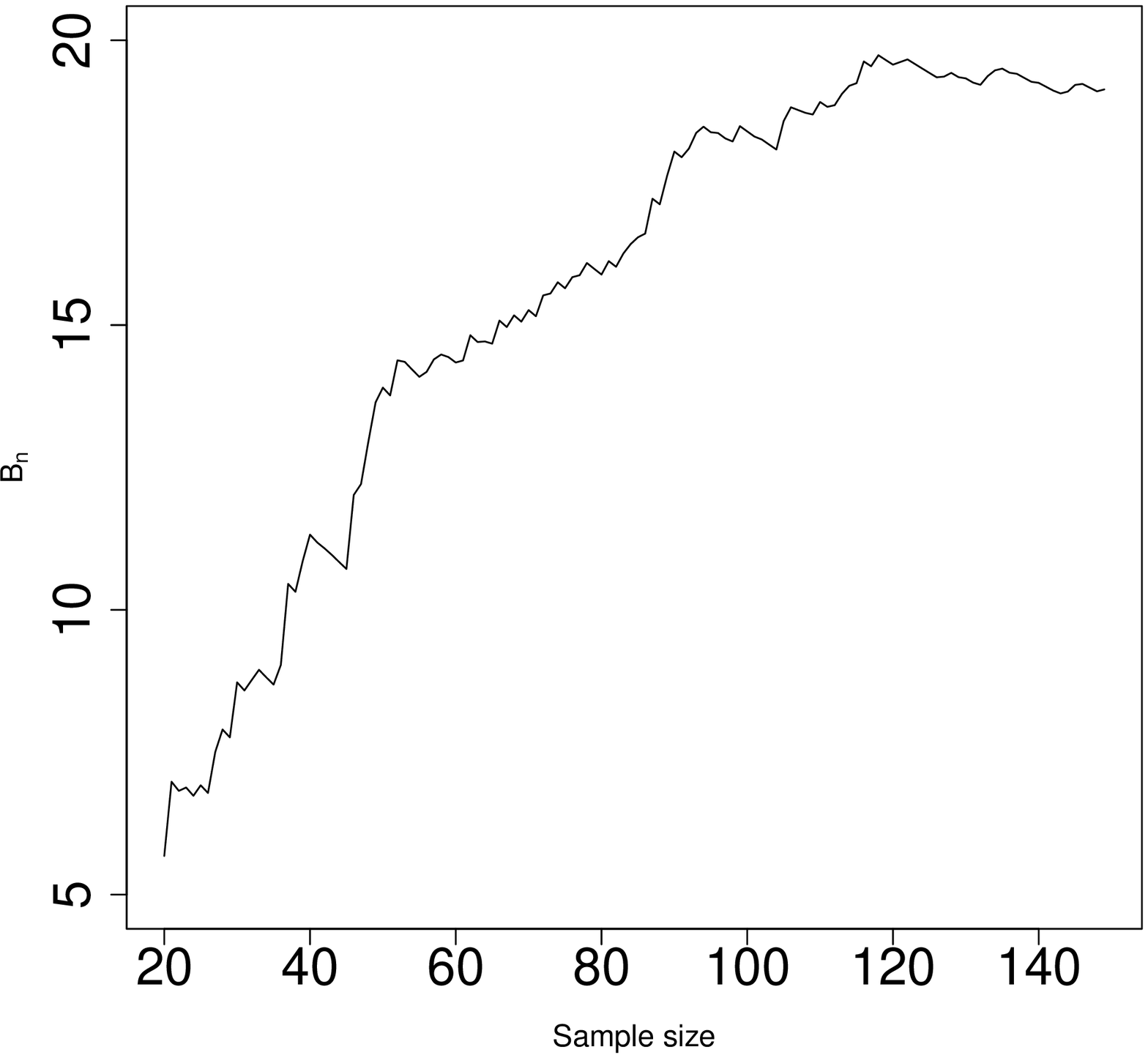}
        \caption{$B_{n,2}$.}
    \end{subfigure}
    \caption{The indices $I_n$ and $B_{n,2}$ of the 150-day sales data with respect to $n=20,\ldots,149$.}\label{LS I B}
\end{figure}
The index $I_n$ does not tend to $1/2$ and the index $B_{n,2}$ stops rising at about $n=120$. These observations suggest the lack of exogenous interference when transiting from the inputs to the outputs, that is, from the leading variable to the sales.

\section{Introducing a controlled experiment}
\label{avr}

To explore how the anomaly-detection method works, we have designed an experiment based on a simple (from the statistical modeling perspective) control system, which is the automatic voltage regulator (AVR) that has been an active research area with a considerable number of innovative designs and algorithms proposed in the literature. For details, we refer to the recent contributions by, e.g., \cite{CD2018}, \cite{G2020}, and extensive references therein.

In its simplest form, the AVR intakes voltages $X_{t}$ and outputs more stable voltages $Y_{t}$ within a pre-specified service range $[a,b]$.  When it is known that the system is free of anomalies, the outputs are
\[
Y_t^0 =(X_t\wedge b -a)_{+}+a
=
\left\{
  \begin{array}{ll}
    a   & \hbox{when } X_t<a,  \\
    X_t & \hbox{when } a\le X_t \le b,  \\
    b   & \hbox{when } X_t>b.
  \end{array}
\right.
\]
Hence, using the ``clamped'' baseline function (see Figure~\ref{h-example})
\begin{equation}\label{clamped-h}
h_c(x)=(x\wedge b -a)_{+}+a,
\end{equation}
the outputs are $Y_t^0=h_c(X_t)$.
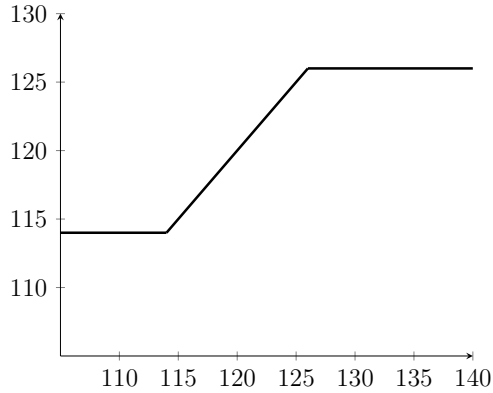
\begin{figure}[h!] \bigskip
\centering
\begin{tikzpicture}[thick, scale=0.8]
\begin{axis}[axis lines=middle,xmin=105,xmax=140,ymin=105,ymax=130]
\addplot+[no marks,domain=100:114,samples=200, very thick, black] {114};
\addplot+[no marks,domain=114:126,samples=200, very thick, black] {x};
\addplot+[no marks,domain=126:140,samples=200, very thick, black] {126};
\end{axis}
\end{tikzpicture}
    \caption{The transfer function $h_c$ with $a=114$ and $b=126$ corresponding to the automatic voltage regulator with the transfer window $120\pm 6 $ volts (i.e., $\pm 5\% $).}
    \label{h-example}
\end{figure}
Note that the clamped function $h_c$ is Lipschitz continuous but not continuously differentiable, and we shall keep this example in mind when deriving results in Sections~\ref{orderly systems} and \ref{disorderly systems} so that to avoid making assumptions that would exclude functions such as $h_c$.

\begin{note}
Transfer functions similar to the one in Figure~\ref{h-example} appear naturally in various reinsurance treaties, where direct insurers and reinsurers calculate their risk transfers using formulas resembling~\eqref{clamped-h} with pre-determined deductibles and policy limits as parameters, very much like $a$ and $b$ in equation~\eqref{clamped-h}. Determining whether or not anomalies (e.g., processing errors) are affecting such transfers is of interest to all parties involved.
\end{note}

In the numerical experiment in Section~\ref{expdesign}, we shall use the clamped function $h_c$ as the baseline function $h_0$, and then use the transfer functions~TF\ref{tf1}--TF\ref{tf3} (Section~\ref{notation}) as $h$. The anomaly-free inputs $X_t$ in the experiment are assumed to follow the $\text{ARMA}(1,1)$ time series, with the input anomalies $\delta_t$ and the output anomalies $\epsilon_t$ being independent within and among them, and coming from a certain parametric distribution. Note that such anomalies can be interpreted as genuinely unintentional; they may arise from, e.g., systematic measurement-errors due to faulty equipment. In Section~\ref{orderly systems} we shall develop results for the anomaly-free outputs $Y^0_t=h_0(X_t)$. In Section~\ref{disorderly systems} we shall do the same for the anomaly affected case, that is, when $\boldsymbol{\varepsilon}_t=(\delta_t,\epsilon_t)\in \mathbb{R}^2$ and thus $Y_t=h(X_t,\delta_t,\epsilon_t)$.

\begin{note}
When the inputs $X_t$  are iid random variables, which is a very special case of the present paper,  anomaly detection in systems with $\delta_t=0$ has been studied by \citet{GZ2020}, with $\epsilon_t=0$ by \citet{GZ2018}, and with arbitrary anomalies $(\delta_t,\epsilon_t)$ by \citet{GZ2019metron}. In the present paper we extend those iid-based results to scenarios when inputs are governed by stationary time-series models, which is a highly important feature from the practical point of view. To achieve these goals, a considerable technical work has to be done, which we  present in Appendix~\ref{proofs}.
\end{note}

We are now ready to familiarize with the anomaly-detection method, and in particular with mathematical definitions of (dis)orderly systems and of the indices $I_n$ and $B_{n,2}$, as well as of the more general index $B_{n,p}$.

\section{Anomaly detection: a foundation}
\label{prelim}

Let $X_1,\dots , X_n$ be the observable part of a stationary sequence of inputs $X_t$, whose marginal cumulative distribution functions (cdf's) are the same; we denote them by $F$. With these observable inputs, there are associated outputs $Y_1,\dots , Y_n$, and so we are dealing with the random input-output pairs $(X_1,Y_1), \dots , (X_n,Y_n)$. Based on them, we wish to determine whether the system transferring the inputs into the outputs is functioning as intended or is systematically affected by anomalies. To successfully tackle this problem, we first need to rigorously define (dis)orderly systems.

Let the cdf $F$ be continuous, which allows us without loss of generality to state that all the inputs $X_1,\dots , X_n$ are different. Hence, their order statistics
\[
X_{1:n}<X_{2:n}< \cdots < X_{n:n}
\]
are strictly increasing. This facilitates unambiguous definition of the concomitants of the outputs $Y_1,\dots , Y_n$, which are denoted by $Y_{1,n}, \dots ,  Y_{n,n}$ and defined by the equation
\[
Y_{t,n}=\sum_{s=1}^n Y_{s}\mathds{1}\{X_{s}=X_{t:n}\},
\]
where $\mathds{1}$ is the indicator: it is equal to $1$ when the condition $X_{s}=X_{t:n}$ is satisfied and $0$ otherwise. We are now in the position to define (dis)orderly systems.

\begin{definition}\label{def-00}
We say that the outputs and thus the system are \textit{in $p$-reasonable order} with respect to the inputs for some $p>0$ if
\[
B_{n,p}:={1\over n^{1/p}}\sum_{t=2}^n |Y_{t,n}-Y_{t-1,n}|=O_{\mathbb{P}}(1)
\]
when $n\to \infty $. If, however, $B_{n,p}\to_{\mathbb{P}} \infty $,
then we say that the outputs and thus the system are \textit{out of $p$-reasonable order}  with respect to the inputs.
\end{definition}

Although this definition is a technicality that is necessary for our anomaly-detection method, it is also natural from the practical point of view. Indeed, detection of anomalies in disorderly systems can hardly be a task worth pursuing. As to the parameter $p$, its role in Definition~\ref{def-00}  is to control tail heaviness of the outputs, and we shall later see that this is achieved by controlling tail heaviness of the inputs. Roughly speaking, we can view $p$ as the order of finite moments. Note that when the outputs are in $p$-reasonable order, the outputs are in $r$-reasonable order for all $r\le p$. On the other hand, if the outputs are out of $p$-reasonable order, the outputs are out of $r$-reasonable order for all $r\ge p$. Hence, we can say that for any given system, there is a threshold $p$ delineating the sets of in-order and out-of-order outputs.

To successfully detect anomalies affecting a system, we of course need to know that the brand new system was in orderly state. For a rigorous definition of the latter notion, we slightly adjust Definition~\ref{def-00} as follows.

\begin{definition}\label{def-00a}
The anomaly-free outputs and thus the anomaly-free system are \textit{in $p$-reasonable order} with respect to the inputs for some $p>0$ if
\[
B^0_{n,p}:={1\over n^{1/p}}\sum_{t=2}^n |Y^0_{t,n}-Y^0_{t-1,n}|=O_{\mathbb{P}}(1)
\]
when $n\to \infty $. (For obvious reasons, we do not consider systems that are out of order when they are free of anomalies.)
\end{definition}

To illustrate the anomaly-free case, that is, when all $\boldsymbol{\varepsilon}_t$'s are equal to $\mathbf{0}$, if all the inputs happen to be equal to the same constant, say $c$, then the system is in $p$-reasonable order for every $p>0$, because $B^0_{n,p}=0$. More generally, next Theorem~\ref{th-1aa} will show that if the transfer function $h$ is sufficiently smooth (e.g., Lipschitz continuous), then the system is in $p$-reasonable order for some $p>0$ even when the inputs are random, although not too heavy tailed. We need to introduce additional notation before we can formulate the theorem.

Let $a_{X}$ and $b_{X}$ be the endpoints of the support of the cdf $F$, that is,
\begin{align*}
a_X&=\sup \{ x\in \mathbb{R} : F(x)=0\},
\\
b_X&=\inf \{ x\in \mathbb{R} : F(x)=1\}.
\end{align*}
These endpoints can of course be infinite, but they never coincide because the cdf $F$ is assumed to be continuous. Therefore, the open interval $(a_{X},b_{X})$ is never empty.

Next, we recall that $h_0$ is called absolutely continuous if there is a function $h^*_0$, called the Radon-Nikodym derivative of $h_0$, that satisfies the equation
\[
h_0(v)-h_0(u) =\int_{u}^{v}h^*_0(x)\dd x
\]
for all $u\le v$. Now we are ready to formulate the theorem that describes the circumstances under which the anomaly-free system is orderly.

\begin{theorem}\label{th-1aa}
The anomaly-free outputs are in $p$-reasonable order with respect to the inputs for some $p\ge 1$ if there is $\alpha \in [1,p]$ such that $\mathbb{E}(|X_1|^{p/\alpha})<\infty $ and one of the following conditions holds:
\begin{enumerate}[\rm (i)]
\item \label{part-1}
If $\alpha=1$, then the baseline function $h_0$ is Lipschitz continuous, that is, there is a constant $K\ge 0$ such that, for all $x,y\in [a_{X},b_{X}]$,
\[
 |h_0(x) - h_0(y)| \le K|x - y| .
\]
\item \label{part-2}
If $\alpha>1$, then the baseline function $h_0$ is absolutely continuous on the interval $[a_{X},b_{X}]$ and its Radon-Nikodym derivative $h^*_0$
satisfies
\[
\int_{-\infty }^{\infty }|h^*_0(x)|^{\alpha/(\alpha-1)}\dd x <\infty .
\]
\end{enumerate}
\end{theorem}

Next are two facts (to be proven later) upon which we base our anomaly-detection method:

\begin{description}
\item[Fact 1:]
If the anomaly-free outputs are in $p$-reasonable order with respect to the inputs, then, under some fairly weak assumptions on the inputs and the transfer function $h$ (details in Section~\ref{orderly systems}), the index
\begin{equation}\label{index-00}
I^0_n:={\sum_{i=2}^n (Y^0_{i,n}-Y^0_{i-1,n})_{+}
\over \sum_{i=2}^n |Y^0_{i,n}-Y^0_{i-1,n}|}
\end{equation}
converges, when $n\to \infty $, to a limit other than $1/2$.
\item[Fact 2:]
If, due to anomalies, the outputs are out of $p$-reasonable order with respect to the inputs, then (details in Section~\ref{disorderly systems}) the index
\begin{equation}\label{index-01}
I_n:={\sum_{i=2}^n (Y_{i,n}-Y_{i-1,n})_{+}
\over \sum_{i=2}^n |Y_{i,n}-Y_{i-1,n}|}
\end{equation}
converges to $1/2$ when $n\to \infty $.
\end{description}

Establishing these two facts rigorously is a complex and lengthy exercise, which we do in Sections~\ref{orderly systems} and \ref{disorderly systems}, as well as in Appendix~\ref{proofs}.  To show that the task is worth the effort, in the next section we show how the anomaly-detection method actually works in the case of the AVR-based experiment that we introduced in Section~\ref{avr}.

\section{The experiment: parameter choices and results}
\label{expdesign}

To illustrate the anomaly-detection method, and in particular Facts 1 and 2 formulated in the previous section, we use the AVR-based experiment with the following parameter choices.

First, the anomaly-free inputs $X_t$ (see Figure~\ref{arma-1})
\begin{figure}[h!]
\includegraphics[width=\textwidth,height=0.5\textwidth]{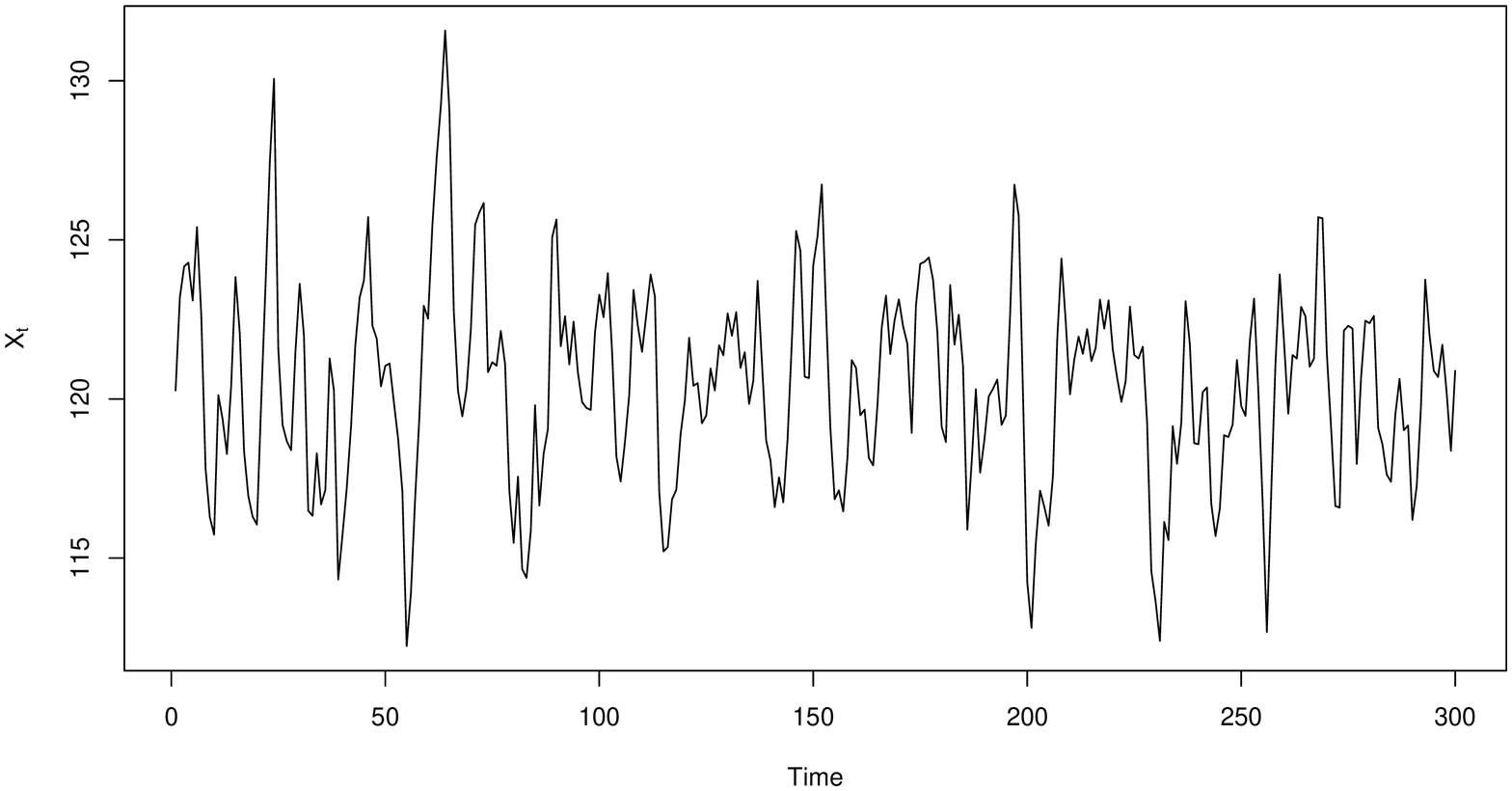}
    \caption{$\text{ARMA}(1,1)$ inputs $(X_t)_{t\in \mathbb{Z}}$ as specified by  model~\eqref{arma-model}.}
    \label{arma-1}
\end{figure}
follow the $\text{ARMA}(1,1)$ time series model
\begin{equation}\label{arma-model}
(X_t - 120) = 0.6(X_{t-1} - 120) + \eta_t + 0.4\eta_{t-1} ,
\end{equation}
where the white noise sequence $\eta_t$ consists of iid, mean zero, normal $\mathcal{N}(0,\sigma^2_{\eta})$ random variables with the variances
$$
\sigma^2_{\eta} = \dfrac{3^2(1 - 0.6^2)}{1 + 2(0.6)(0.4) + 0.4^2} = \dfrac{5.76}{1.64} \approx 3.512195.
$$
Under these specifications \citep[e.g.,][Eq.~(3.4.7), p.~79]{BJRL2015} the input time series $X_t$  has the marginal normal distribution with mean $120$ and variance $9$, that is,
\[
X_t\sim\mathcal{N}(120,9)
\]
for every $t\in\mathbb{Z}$.

Next, the input anomalies $\delta_t$ are iid $\text{Lomax}(\alpha,1)$ with shape parameter $\alpha>0$ \citep[e.g.,][Section~2.3.11, pp.~23--24]{LX2006}, and the output anomalies $\epsilon_t$ are also iid $\text{Lomax}(\alpha,1)$. Both the input and output anomalies are  independent of each other, and they are also independent of the inputs $X_t$. (Such anomalies can be interpreted as genuinely unintentional.)
Hence, the anomalies are independent, non-negative, random variables with the means
\[
\mathbb{E}(\delta_t)=\mathbb{E}(\epsilon_t)={1\over \alpha-1}
\]
and the variances
\[
\Var(\delta_t)=\Var(\epsilon_t)=
\left\{
  \begin{array}{ll}
   {\alpha\over (\alpha-1)^2(\alpha-2)} & \hbox{ when } \alpha>2,  \\
    \infty & \hbox{ when } 1<\alpha \le 2.
  \end{array}
\right.
\]
We set the following values for the shape parameter $\alpha $:
\begin{itemize}
  \item $\alpha=11$, which gives $\mathbb{E}(\delta_t)=\mathbb{E}(\epsilon_t)=0.1$ and $\Var(\delta_t)=\Var(\epsilon_t)=0.0122$, thus making, in average, the anomalies look small  if compared to the nominal voltage $120$;
  \item $\alpha=1.2$, which gives $\mathbb{E}(\delta_t)=\mathbb{E}(\epsilon_t)=5$ and infinite variances, thus making, in average, the anomalies look moderate in size if compared to the nominal voltage $120$.
\end{itemize}

We shall see that in both cases the method detects the anomalies with remarkable easiness, although the required sample size when $\alpha=11$ needs to be, naturally, larger than when $\alpha=1.2$ in order to reach the same conclusion. A few clarifying notes follow.

\begin{note}
The terms ``small'' and ``moderate'' that we used to describe anomalies with average values $0.1$ and $5$, respectively, are our terms and may not coincide with what the reader might think about such anomalies. Nevertheless, it seems to us that the terms ``small'' and ``moderate'' correlate well with the accepted notions of ``strict'' and ``satisfactory'' AVR service ranges, which are $120\pm 3$ and $120\pm 6$, respectively.
\end{note}

\begin{note}
Among the two choices of $\alpha $ made above, one leads to a finite variance and another to infinite. These two distinct scenarios are of practical interest. Indeed, based on empirical evidence, there has been a considerable discussion in the literature as to what distribution tails (and related dependence structures) could be suitable for modelling, e.g., data traffic and cyber risks. For details and further references on the topic, we refer to, e.g., \cite{hrs1998}, \cite{MS2010}, and \cite{EHF2016}.
\end{note}

To proceed with the set-up of our AVR-based experiment, we next introduce three service ranges, among which ``satisfactory'' and ``strict'' are commonly used terms in practice, and ``precise'' is an  artefact.
\begin{description}
\item[Satisfactory:]
  \[
  [a,b]=[114,126],
   \]
  which is $120\pm 6$ (i.e., $\pm 5\%$) and is considered a standard supply range in, e.g., Canada (recall Figure~\ref{h-example}), with $120$ being the nominal voltage.
  \item[Strict:]
  \[
  [a,b]=[117,123],
  \]
  which is $120\pm 3$ (i.e., $\pm 2.5\%$).
  \item[Precise:]
  \[
  [a,b]=[120,120]=\{120\}.
  \]
\end{description}

The anomaly-detection method in the case of the clamped transfer function $h_c$ corresponding to these three service ranges is illustrated in the next three subsections, with a more extensive set of illustrative graphs provided in Appendix~\ref{graphs}. For space considerations, we only consider the case $p=2$. We start with the precise service range.

\subsection{Precise service range}

When the AVR service range is precise, which is an artefact created only for illustrative purposes, the clamped function is constant, that is,
$$
h_c(v) =120.
$$
Since $h_0=h_c$ in this experiment, all the three transfer functions~TF\ref{tf1}--TF\ref{tf3} in the anomaly-free case yield
\[
B_{n,2}^0= 0,
\]
thus implying that the system is orderly. The index $I_n^0$ is undefined, as it is the ratio $0/0$. These notes also apply to the anomaly-affected case $h(x,y,0)$, as it is equal to $h_c(x+y)$, which is $120$, a constant. In the remaining two anomaly-affected cases $h(x,0,z)$ and $h(x,y,z)$, which are identical and given by the equations
\[
h(x,0,z) =120+z= h(x,y,z),
\]
the output anomalies affect the system. Figure~\ref{precise-1}
\begin{figure}[h!]
    \centering
    \begin{subfigure}[b]{0.39\textwidth}
        \includegraphics[width=\textwidth]{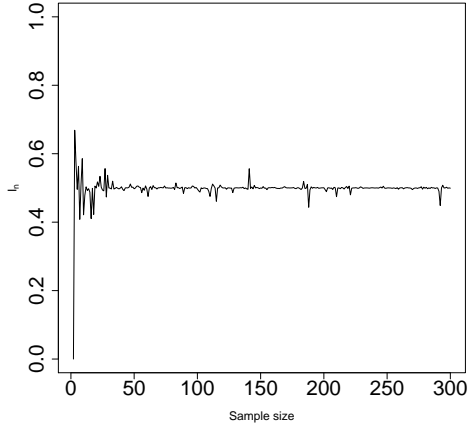}
        \caption{$I_n$ with $\text{Lomax}(1.2,1)$ anomalies}
    \end{subfigure}
\hspace{10mm}
      %(or a blank line to force the subfigure onto a new line)
    \begin{subfigure}[b]{0.39\textwidth}
        \includegraphics[width=\textwidth]{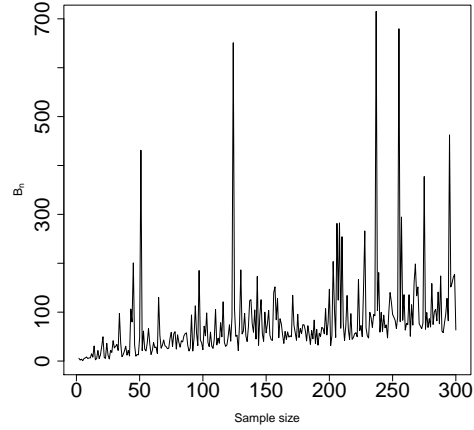}
        \caption{$B_{n,2}$ with $\text{Lomax}(1.2,1)$ anomalies}
    \end{subfigure}
\\
    \begin{subfigure}[b]{0.39\textwidth}
        \includegraphics[width=\textwidth]{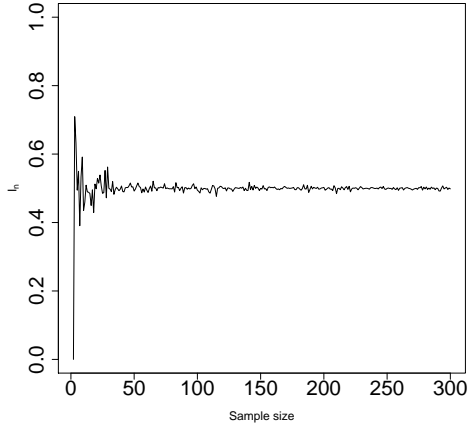}
        \caption{$I_n$ with $\text{Lomax}(11,1)$ anomalies}
    \end{subfigure}
\hspace{10mm}
      %(or a blank line to force the subfigure onto a new line)
    \begin{subfigure}[b]{0.39\textwidth}
        \includegraphics[width=\textwidth]{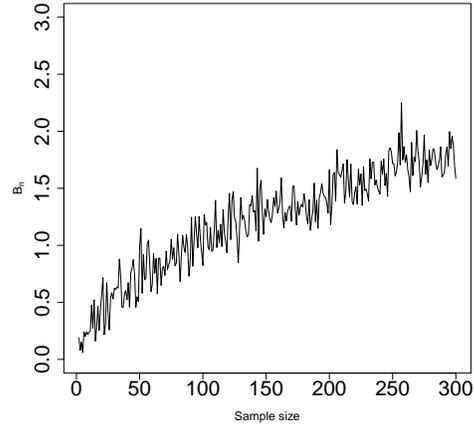}
        \caption{$B_{n,2}$ with $\text{Lomax}(11,1)$ anomalies}
    \end{subfigure}
    \caption{The anomaly-affected indices $I_n$ and $B_{n,2}$ for the precise service range with respect to the sample sizes $2\le n\le 300$ for $\text{ARMA}(1,1)$ inputs.}
    \label{precise-1}
\end{figure}
depicts $I_n$ and $B_{n,2}$ for various sample sizes $n$. Note that the index $I_n$ initially fluctuates but quickly starts to tend to $1/2$, whereas $B_{n,2}$ is increasing with respect to the sample size. These two observations suggest that anomalies are affecting the system, which is indeed the case given the experimental design.

\subsection{Strict service range}

When the AVR service range is strict, the clamped function is
$$
h_c(v) = (\min\{123, v\} - 117)_+ + 117
=
\left\{
  \begin{array}{ll}
    117   & \hbox{when } v<117,  \\
    v & \hbox{when } 117\le v \le 123,  \\
    123   & \hbox{when } v>123.
  \end{array}
\right.
$$
Figure~\ref{precise-2}
\begin{figure}[h!]
    \centering
    \begin{subfigure}[b]{0.39\textwidth}
        \includegraphics[width=\textwidth]{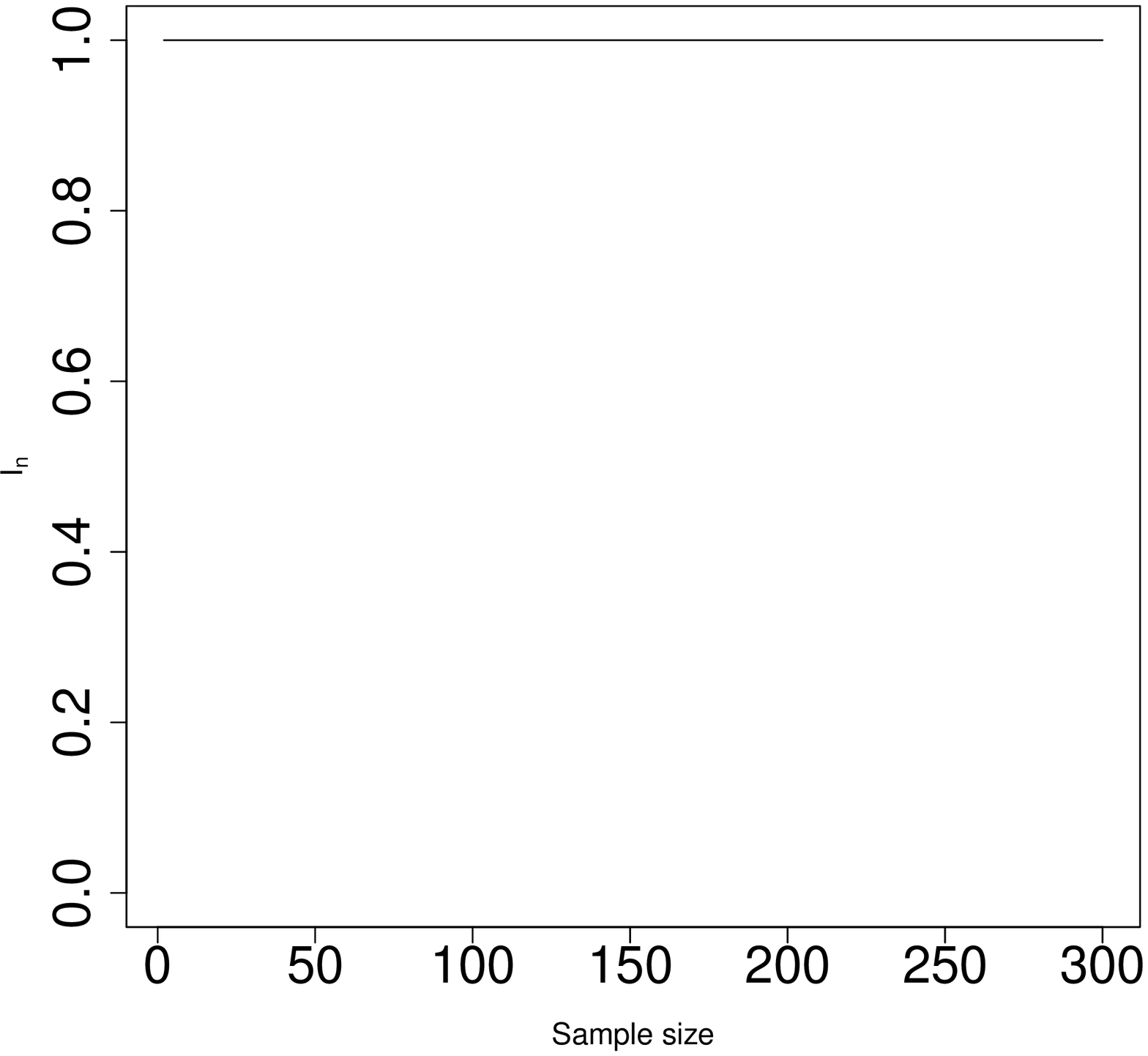}
        \caption{$I_n^0$}
    \end{subfigure}
\hspace{10mm}
      %(or a blank line to force the subfigure onto a new line)
    \begin{subfigure}[b]{0.39\textwidth}
        \includegraphics[width=\textwidth]{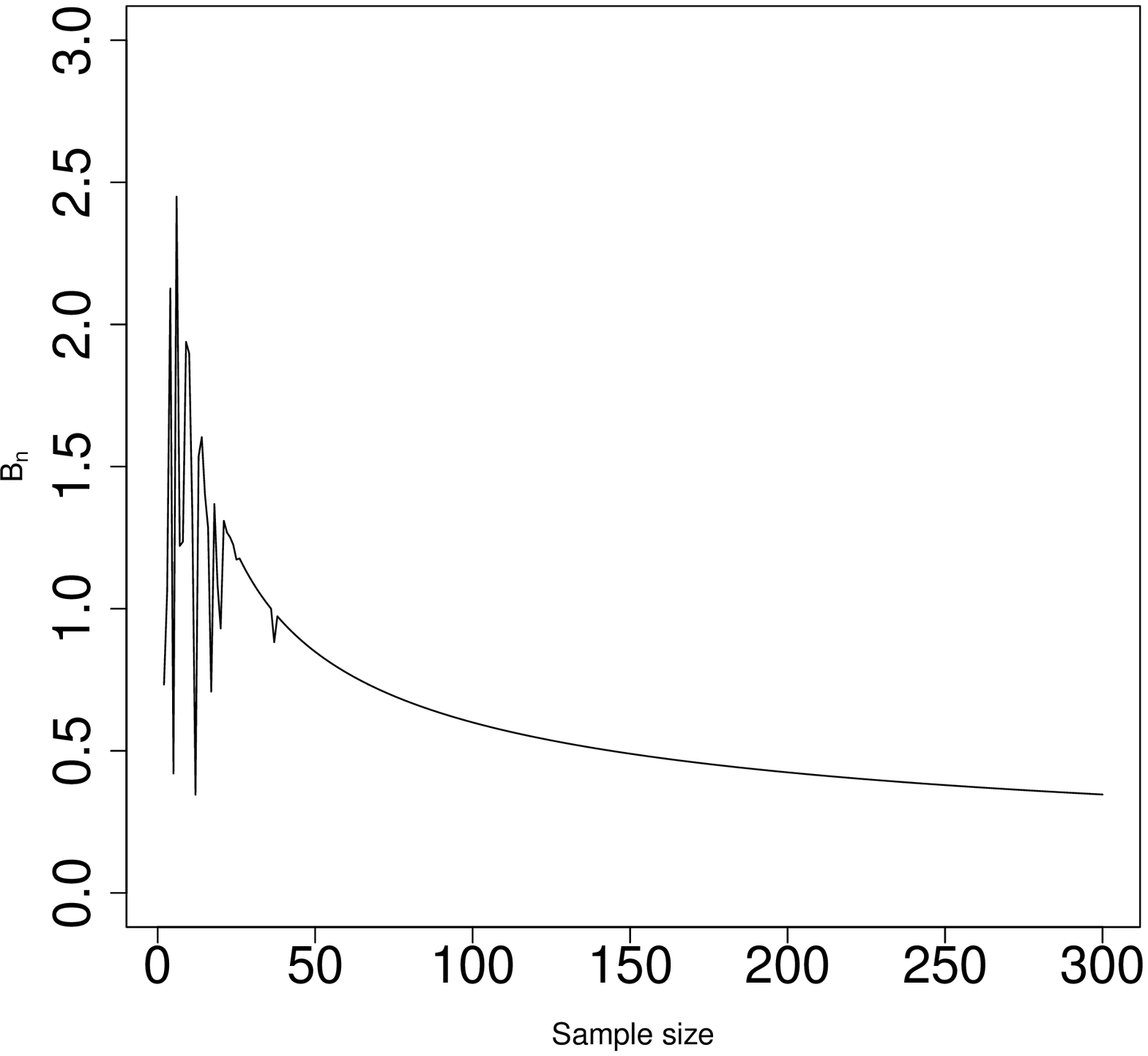}
        \caption{$B_{n,2}^0$}
    \end{subfigure}
    \caption{The anomaly-free indices $I_n^0$ and $B_{n,2}^0$ for the strict service range with respect to the sample sizes $2\le n\le 300$ for $\text{ARMA}(1,1)$ inputs.}
    \label{precise-2}
\end{figure}
depicts the anomaly-free indices $I_n^0$ and $B_{n,2}^0$ for various sample sizes $n$. Looking at the graphs, we safely infer that $B^0_{n,p}=O_{\mathbb{P}}(1)$, which implies that the anomaly-free system is orderly with respect to the inputs, and we also see that $I_n^0$ does not converge to $1/2$, which confirms that the system is free of anomalies.

When, however, the system is affected by anomalies at the input and/or output stages, the behaviour of $I_n$ and $B_{n,2}$ changes drastically. We see from Figures~\ref{stringent-2} (when $\alpha=1.2$) and \ref{stringent-2-11} (when $\alpha=11$) that the index $I_n$ tends to $1/2$. Naturally, it tends to $1/2$ faster when $\alpha=1.2$ than when $\alpha=11$, simply because the anomalies  in the latter case are less noticeable. For both $\alpha$ values, the index $B_{n,2}$ has the tendency to grow. These observations suggest that the system is being affected by anomalies, which is indeed the case.

\subsection{Satisfactory service range}

When the AVR service range is satisfactory, the clamped function is
$$
h_c(v) = (\min\{126, v\} - 114)_+ + 114
=
\left\{
  \begin{array}{ll}
    114   & \hbox{when } v<114,  \\
    v & \hbox{when } 114\le v \le 126,  \\
    126   & \hbox{when } v>126.
  \end{array}
\right.
$$
Figure~\ref{precise-3}
\begin{figure}[h!]
    \centering
    \begin{subfigure}[b]{0.39\textwidth}
        \includegraphics[width=\textwidth]{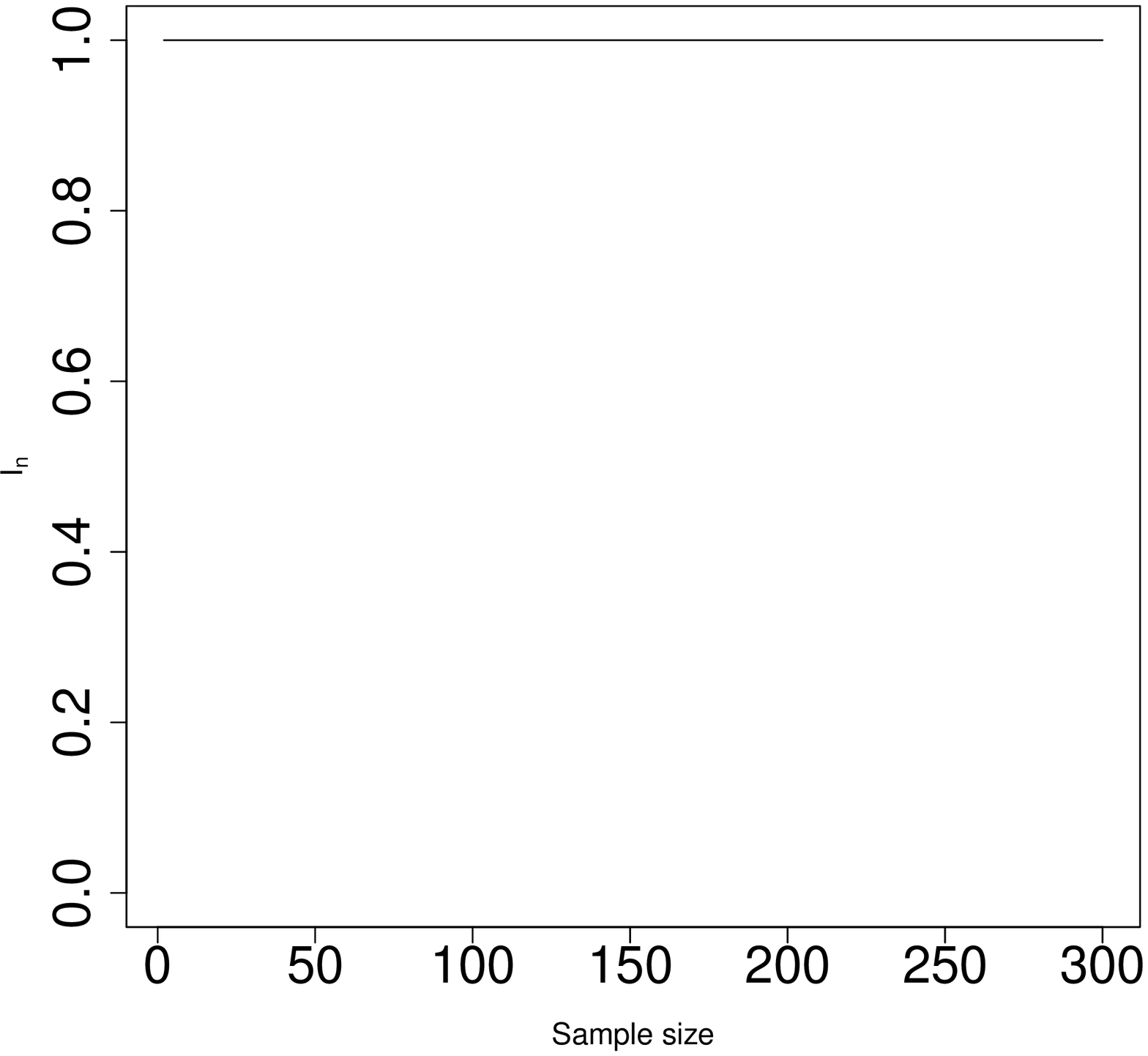}
        \caption{$I_n^0$}
    \end{subfigure}
\hspace{10mm}
      %(or a blank line to force the subfigure onto a new line)
    \begin{subfigure}[b]{0.39\textwidth}
        \includegraphics[width=\textwidth]{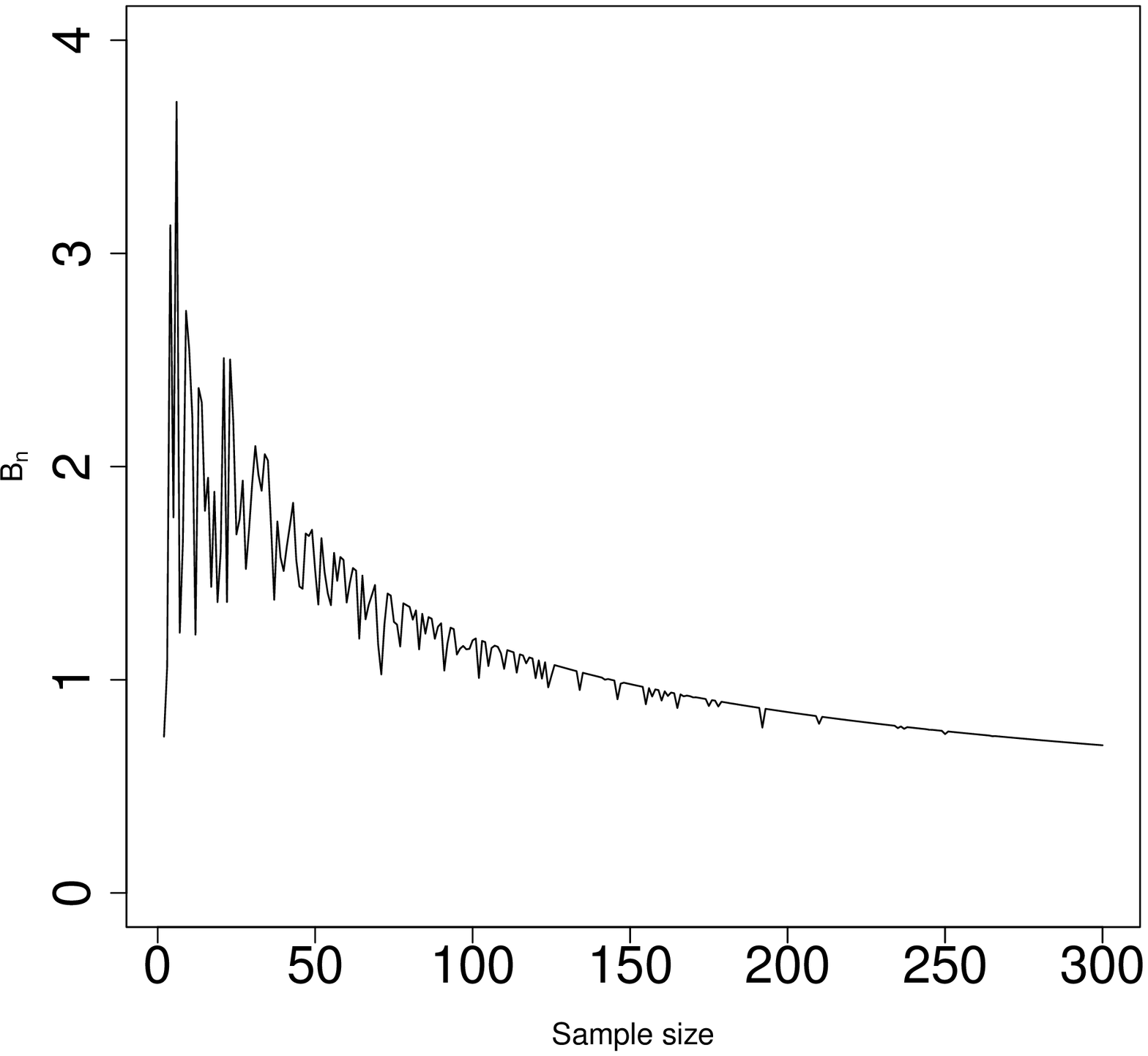}
        \caption{$B_{n,2}^0$}
    \end{subfigure}
    \caption{The anomaly-free indices $I_n^0$ and $B_{n,2}^0$ for the satisfactory service range with respect to the sample sizes $2\le n\le 300$ for $\text{ARMA}(1,1)$ inputs.}
    \label{precise-3}
\end{figure}
resembles Figure~\ref{precise-2}, and Figures~\ref{standard-3}--\ref{standard-3-11} convey essentially the same information as Figures~\ref{stringent-2}--\ref{stringent-2-11}.

\section{Anomaly-free systems}
\label{orderly systems}

In this section we specify conditions under which the index $I^0_n$ given by equation~\eqref{index-00} converges to a limit other than $1/2$. We note at the outset that since the positive part $z_{+}$ of any real number $z\in \mathbb{R}$ can be expressed as $(|z|+z)/2$, the index $I^0_n$ can be re-written as
\begin{equation}\label{index-0}
I^0_n={1\over 2}\bigg( 1+
{Y^0_{n,n}-Y^0_{1,n}\over \sum_{i=2}^n |Y^0_{i,n}-Y^0_{i-1,n}|} \bigg ) .
\end{equation}
Hence, our goal becomes to understand when and where the numerator and the denominator of the ratio
\begin{equation}\label{ratio-0}
\Lambda_n:={Y^0_{n,n}-Y^0_{1,n}\over \sum_{i=2}^n |Y^0_{i,n}-Y^0_{i-1,n}|}
\end{equation}
converge. These are the topics of the following two subsections.

\subsection{Asymptotics of the $\Lambda_n$ numerator}

We start with a definition.

\begin{definition}
We say that the inputs $X_t$ having the same continuous marginal cdf $F$ are \textit{temperately dependent} if, for every $x\in (a_{X},b_{X})$ and when $n\to \infty $,
\begin{equation}
\label{md-1}
\mathbb{P}(X_{1:n}\ge x)\to 0
\quad \textrm{and} \quad
\mathbb{P}(X_{n:n}\le x)\to 0.
\end{equation}
\end{definition}

A few clarifying notes follow. First, the open interval $(a_{X},b_{X})$ is not empty  because the cdf $F$ is continuous. Second, our use of the term \textit{temperately dependent} is natural because property~\eqref{md-1} is simultaneously related to $1$-minimally and $n$-maximally dependent random variables in the terminology used by \citet{GC1992}.  To clarify, consider two extreme cases:
\begin{itemize}
\item
If $X_t$ is a sequence of iid random variables, then
$\mathbb{P}(X_{1:n}\ge x)=(1-F(x-))^n$
and
$\mathbb{P}(X_{n:n}\le x)=F(x)^n$, where $x\mapsto F(x-)$ is the left-continuous version of $F$. Since for all $x\in (a_{X},b_{X})$, both $F(x)$ and  $F(x-)$ are in the interval $ (0,1)$, property~\eqref{md-1} holds.
\item
If $X_t$ is a sequence of super-dependent random variables,
that is, if there is a random variable $X$ such that $X_t=X$ for all $t\in \mathbb{Z}$, then $X_{1:n}=X=X_{n:n}$ and so neither of the two probabilities in property~\eqref{md-1} converges to $0$.
\end{itemize}

Finally, we note that the concept of temperate dependence is closely related to the existence of phantom distributions, which originate from the work of \cite{O1987}. For further details, examples, and extensive references on this topic, we refer to \cite{J1991}, \cite{B2007}, and \cite{DJL2015}. Phantom distributions for non-stationary random sequences have been tackled by \cite{J1993}.

\begin{theorem}\label{th-md}
Let the inputs $X_t$ be strictly stationary and temperately dependent. If the baseline function $h_0$ is absolutely continuous on the interval $[a_{X},b_{X}]$ and its Radon-Nikodym derivative $h_0^*$ is integrable on $[a_{X},b_{X}]$, then
\begin{equation}
\label{md-2}
Y^0_{n,n}-Y^0_{1,n}\to_{\mathbb{P}} h_0(b_{X})- h_0(a_{X})
\end{equation}
when $n\to \infty $.
\end{theorem}

Recall that a strictly stationary time series $X_t$ is $\alpha$-mixing (i.e., strongly mixing) if
\[
\alpha_X(t) := \sup
\big |\mathbb{P}(A\cap B) - \mathbb{P}(A)\mathbb{P}(B)\big |
\to 0
\]
when $t\to \infty $, with the supremum taken over all $A\in\mathscr{F}^0_{-\infty}$ and $B\in\mathscr{F}^{\infty}_t$, where the two $\sigma$-algebras are defined as follows:
\[
\mathscr{F}^0_{-\infty}=\sigma(X_u, u \le 0)
\quad \textrm{and} \quad
\mathscr{F}^{\infty}_k=\sigma(X_v, v\ge t ).
\]
We refer to \citet{LL1996},   \cite{B2007}, and \cite{R2017} for details and references on various notions of mixing. In the context of the present paper, a particularly important random sequence is the strictly stationary ARMA($p,q$) time series $X_t$. We refer to \citet{Mokkadem1988} who has shown, among other things, that a strictly stationary ARMA time series is $\beta$-mixing (i.e., completely regular) and thus $\alpha$-mixing (i.e., strongly mixing).

\begin{theorem}\label{lm-1}
If the inputs $X_t$ are strictly stationary and $\alpha$-mixing, then they are temperately dependent.
\end{theorem}

The proofs of Theorems~\ref{th-md} and \ref{lm-1} are in Appendix~\ref{proofs}.

\subsection{Asymptotics of the $\Lambda_n$ denominator}

We start with a definition, which is a weak (i.e., in probability) form of the classical Glivenko-Cantelli theorem.

\begin{definition}
We say that the inputs $X_t$ having the same marginal cdf $F$ satisfy the Glivenko-Cantelli property if
\begin{equation}
\label{gc=0}
\Vert F_n-F\Vert:=\sup_{x\in \mathbb{R}}\big | F_n(x)-F(x) \big | \to_{\mathbb{P}} 0,
\end{equation}
where $F_n$ is the empirical cdf based on the random variables $X_{1},  \dots ,  X_{n}$.
\end{definition}

The classical Glivenko-Cantelli theorem says that statement~\eqref{gc=0} (with convergence in probability replaced by almost surely) holds for iid random sequences. Establishing the Glivenko-Cantelli property for dependent sequences has been a challenging but fruitful task. In particular, results by \citet[Corollary~2.1, p.~49]{CR1992}
and \citet[Proposition~7.1, p.~114]{R2017} tell us that if a strictly stationary sequence $X_t$ is $\alpha$-mixing and there exists a constant $\nu>0$ such that
\begin{equation}
\label{sm-0}
\alpha_X(t)=O( t^{-\nu})
\end{equation}
when $t\to \infty $, then the Glivenko-Cantelli property holds.

Consider now the stationary  $\text{ARMA}(p,q)$ time series $X_t$ that follows the dynamical model
\[
\sum^p_{i=0}\phi_iX_{t-i} = \sum^q_{j=0}\theta_j \eta_{t-j} ,
\quad t\in \mathbb{Z},
\]
with $\phi_0=1$ and some parameters $\phi_i, \theta_j \in \mathbb{R}$ such that the absolute values of all the roots of the characteristic polynomial $z\mapsto \sum^p_{i=0}\phi_iz^i$ are (strictly) greater than $1$. Hence, the time series is causal.

\begin{note}
There is a clash of notation between the $p$ in $\text{ARMA}(p,q)$ and the $p$ in the earlier introduced definition of $p$-reasonable order. The two $p$'s are unrelated, and we do not expect them to cause any confusion. We have simply run out of different notation, especially given the deeply rooted traditions in the literature, such as those in time series analysis.
\end{note}

\citet[Theorem~1]{Mokkadem1988} has proved that if the white noise sequence $\eta_t$ is iid with absolutely continuous (with respect to the Lebesgue measure) marginal distributions, then the time series $X_t$ is geometrically completely regular. That is, there exists $\rho\in (0,1)$ such that
\begin{equation}
\label{arma-0}
\beta_X(t)=O(\rho^t)
\end{equation}
when $t\to \infty $, where $\beta_X(t)$ is the complete regularity coefficient \citep{D1973} defined by
\[
\beta_X(t) =\mathbb{E}\Big(\sup\big|\mathbb{P}(B \mid \mathscr{F}^0_{-\infty}) - \mathbb{P}(B)\big|\Big) ,
\]
where the supremum is taken over all $B\in\mathscr{F}^{\infty}_t $.
As noted by \cite{Mokkadem1988}, the bound
\[
\alpha_X(t)\le \beta_X(t)
\]
holds, and thus statement~\eqref{arma-0} implies \eqref{sm-0} for any $\nu >0$, which in turn establishes the Glivenko-Cantelli property for the sequence $X_t$.

We are now in the position to formulate the main result of this subsection concerning the denominator on the right-hand side of equation~\eqref{index-0}.

\begin{theorem}\label{th-md-2}
Let the inputs $X_t$ be strictly stationary, temperately dependent,  satisfy the Glivenko-Cantelli property, and have finite $p$-th moments  $\mathbb{E}(|X_t|^p)<\infty $ for some $p>2$. Let the cdf $F$ and its quantile function $F^{-1}$ be continuous. Finally, assume that the baseline function $h_0$ is absolutely continuous on $[a_{X},b_{X}]$ and such that its Radon-Nikodym derivative $h^*_0$ is continuous on a finite interval $[a,b]\subseteq [a_{X},b_{X}]$ and vanishes outside $[a,b]$. Then
\begin{equation}
\label{md-4}
\sum_{t=2}^n |Y^0_{t,n}-Y^0_{t-1,n}|\to_{\mathbb{P}} \int_{a}^{b}|h^*_0(x)|\dd x
\end{equation}
when $n\to \infty $.
\end{theorem}

The restriction of the support of the Radon-Nikodym derivative $h^*_0$ to only a finite interval $[a,b]$ is due to two reasons:
\begin{enumerate}[1)]
\item
those real-life examples (automatic voltage regulators, insurance layers, etc.)  that initiated our current research are based on finite transfer windows;
\item
dealing with finite intervals $[a,b]$ considerably simplifies mathematical technicalities, which is an appealing feature, especially because we do not have a solid practical justification that would  warrant further technical complexities.
\end{enumerate}

\subsection{Back to the index $I^0_n$}

The following corollary to Theorems~\ref{th-md} and \ref{th-md-2} is the main result of entire   Section~\ref{orderly systems}. Since the conditions of Theorem~\ref{th-md} make up only a subset of the conditions of Theorem~\ref{th-md-2}, we thus impose the latter set of conditions when formulating the corollary.

\begin{corollary}\label{cor-md-2}
Under the conditions of Theorem~\ref{th-md-2}, we have
\begin{equation}\label{index-0cor}
I^0_n\to_{\mathbb{P}}
I^0_{\infty} := {\int_{a}^{b}(h^*_0(x))_{+}\dd x
\over \int_{a}^{b}|h^*_0(x)|\dd x}
\end{equation}
when $n\to \infty $.
\end{corollary}

Note the representation (recall equation~\eqref{index-0})
\begin{equation}\label{index-0cor00}
I^0_{\infty} =
{1\over 2} \big( 1+ \Lambda(h_0)\big),
\end{equation}
where
\[
 \Lambda(h_0)={h_0(b)-h_0(a) \over \int_{a}^{b}|h^*_0(x)|\dd x}.
\]
We shall see in the next section that for anomaly-affected systems, the empirical index $I_n$ converges to $0.5$. To distinguish this case from the limit $I^0_{\infty}$ in the currently discussed anomaly-free case, we need to assume $\Lambda(h_0)\neq 0$, which is tantamount to assuming $h_0(b)\neq h_0(a)$, because the numerator of $\Lambda(h_0)$ is positive. This is natural from the practical point of view as it excludes those transfer functions (which are usually non-decreasing) whose values at the end-points of the transfer window $[a,b]$ coincide. Moreover, given model uncertainty, we need to ensure that $h_0$ belongs to a class of functions for which $\Lambda(h_0)$ is sufficiently distant from $0$ so that in the presence of statistical uncertainty we could still -- with high confidence -- be able to see whether the empirical index $I_n$ converges to $0.5$ or some other number.

\begin{note}
The ratio on the right-hand side of statement~\eqref{index-0cor} arises as a normalized distance in a functional  space~\citep{DZ2017}, which after a discretization gives rise to the index $I^0_n$ and thus, in turn, to the index $I_n$~\citep{CDGZ2018}. For a generalization of these indices to multi-argument functions with further applications, we refer to~\cite{DMZ2019}. For related mathematical considerations, we refer to \cite{P1964}.
\end{note}

\section{Anomaly-affected orderly systems}
\label{disorderly systems}

If the system is out of $p$-reasonable order, then the index $I_n$ tends to $1/2$, as shown in the next theorem.

\begin{theorem}\label{th-1}
Let the outputs $Y_t$ be identically distributed random variables with finite $p$-th moments $\mathbb{E}(|Y_t|^p)<\infty $ for some $p\ge 1$. If the outputs are out of $p$-reasonable order  with respect to the inputs, then
\[
I_n\to_{\mathbb{P}} {1\over 2}
\]
when $n\to \infty $.
\end{theorem}

To apply Theorem~\ref{th-1} for detecting non-degenerate anomalies $\boldsymbol{\varepsilon}_t$, we need to assume that when all $\boldsymbol{\varepsilon}_t$'s are equal to $\mathbf{0}$, then the system is in $p$-reasonable order. Hence, our task in this section is this:
Assuming that the system with anomaly-free outputs $Y^0_{i}=h_0(X_{i})=h(X_{i},\mathbf{0})$ is in $p$-reasonable order for some $p>0$, we need to show that the system becomes out of $p$-reasonable order when the outputs $Y_t$ are equal to $h(X_t,\boldsymbol{\varepsilon}_t)$ with non-degenerate at $\mathbf{0}$ random anomalies $\boldsymbol{\varepsilon}_t$. In other words, assuming $B^0_{n,p}=O_{\mathbb{P}}(1)$
when $n\to \infty $, we need to show that
$B_{n,p}\to _{\mathbb{P}} \infty $ when $\boldsymbol{\varepsilon}_t$'s are non-degenerate at $\mathbf{0}$.

To avoid overloading arguments with mathematical complexities, from now on we set $d=2$ and work with the three transfer functions~TF\ref{tf1}--TF\ref{tf3} (Section~\ref{notation}). Consequently, in the anomaly-free case we have
\begin{equation}\label{bnp0-a}
B^0_{n,p}={1\over n^{1/p}}\sum_{t=2}^n |h_0(X_{t:n})-h_0(X_{t-1:n})|.
\end{equation}
We refer to Theorem~\ref{th-1aa} for a description of those inputs $X_t$ and the baseline function $h_0$ for which the anomaly-free system is in $p$-reasonable order. Our next theorem deals with the case when the input anomalies $\delta_t$ are absent.

\begin{theorem}\label{tf-1}
Let $\delta_t=0$ for all $t\in \mathbb{\mathbb{Z}}$, and let the anomaly-free outputs $Y_t^0$ be in $p$-reasonable order with respect to the inputs for some $p>0$. The outputs $Y_t$ are out of $p$-reasonable order with respect to the inputs if and only if the output anomalies $\epsilon_t$ are out of $p$-reasonable order with respect to the inputs.
\end{theorem}

To illustrate Theorem~\ref{tf-1}, consider the case when the output anomalies $\epsilon_t$ and the inputs $X_t$ are independent. Assume also that the output anomalies $\epsilon_t$ are iid and have finite first moments. (Such anomalies can be interpreted as genuinely unintentional.) In this case, the joint distribution of the concomitants $(\epsilon_{1,n},\dots, \epsilon_{n,n})$ is the same as the joint distribution of the anomalies $(\epsilon_1,\dots , \epsilon_n)$ themselves. Consequently, the output anomalies $\epsilon_t$ are out of $p$-reasonable order with respect to the inputs if and only if
\begin{equation}\label{eps-0}
{1\over n^{1/p}}\sum_{t=2}^n |\epsilon_{t}-\epsilon_{t-1}|\to _{\mathbb{P}} \infty
\end{equation}
when $n\to \infty $. Statement~\eqref{eps-0} holds (see Lemma~\ref{positive-moment}) whenever the distribution of $\epsilon_1$ is non-degenerate (at any one point). Hence, we have the following corollary to Theorem~\ref{tf-1}.

\begin{corollary}\label{cor-1}
Let $\delta_t=0$ for all $t\in \mathbb{\mathbb{Z}}$, and let the anomaly-free outputs $Y_t^0$ be in $p$-reasonable order with respect to the inputs.  The outputs $Y_t$ are out of $p$-reasonable order with respect to the inputs for every $p>1$ whenever the output anomalies $\epsilon_t$ are iid, non-degenerate, and independent of the inputs.
\end{corollary}

For a special but important case of Corollary~\ref{cor-1}, recall that by Theorem~\ref{th-1aa}, the anomaly-free outputs $Y_t^0$ are in $p$-reasonable order with respect to the inputs when the baseline function $h_0$ is Lipshitz continuous. We shall encounter the latter assumption in the following two theorems.

First, we tackle the case when the output anomalies $\epsilon_t$ are not present.

\begin{theorem}\label{tf-2}
Let $\epsilon_t=0$ for all $t\in \mathbb{\mathbb{Z}}$, and let the baseline function $h_0$ be Lipshitz continuous. Furthermore, let the inputs $X_t$ be strictly stationary,  $\alpha $-mixing, and have finite $p$-th moments  $\mathbb{E}(|X_t|^p)<\infty $ for some $p> 1$. Then the outputs $Y_t$ are out of $p$-reasonable order with respect to the inputs whenever the following conditions hold:
\begin{enumerate}[\rm (i)]
\item  \label{part-11}
the input anomalies $\delta_t$ are iid and independent of the inputs $X_t$;
\item  \label{part-12}
$\mathbb{E}\big(|h_0(X_{1}+\delta_{2})-h_0(X_{1}+\delta_{1})|\big) >0 $.
\end{enumerate}
\end{theorem}

A sufficient condition for assumption~\eqref{part-12} can be obtained via Lemma~\ref{positive-moment} and the elementary bound
\[
\mathbb{E}\big(|h_0(X_{1}+\delta_{2})-h_0(X_{1}+\delta_{1})|\big) \ge \mathbb{E}\big(|g_0(\delta_{2})-g_0(\delta_{1})|\big),
\]
where
\begin{equation}\label{g0-0}
g_0(y)=\mathbb{E}(h_0(X_{1}+y)).
\end{equation}
That is, assumption~\eqref{part-12} is satisfied whenever the distribution of $g_0(\delta_{1})$ is non-degenerate.

Finally, we tackle the case when the two anomalies $\delta_t$ and $\epsilon_t$ are non-degenerate.

\begin{theorem}\label{tf-3}
Let the baseline function $h_0$ be Lipshitz continuous. Furthermore, let the inputs $X_t$ be strictly stationary, $\alpha $-mixing, and have finite $p$-th moments $\mathbb{E}(|X_t|^p)<\infty $ for some $p> 1$.  Then the outputs $Y_t$ are out of $p$-reasonable order with respect to the inputs whenever the following conditions hold:
\begin{enumerate}[\rm (i)]
\item \label{part-121}
the input anomalies $\delta_t$ are iid, the output anomalies $\epsilon_t$ are also iid, they are independent of each other, and are also independent of the inputs $X_t$;
\item \label{part-122}
$\mathbb{E}\big(|h_0(X_{1}+\delta_{2})+\epsilon_{2}
-h_0(X_{1}+\delta_{1})-\epsilon_{1}|\big)>0$.
\end{enumerate}
\end{theorem}

Assumption~\eqref{part-122} is satisfied  when the random variable $g_0(\delta_{1})+\epsilon_{1}$ is non-degenerate, where $g_0$ is the same function as in equation~\eqref{g0-0}.

\section{A summary and potential extensions}
\label{conclude}

In this paper we have explored a method for detecting systematic anomalies affecting systems when  genuine anomaly-free inputs belong to a large class of stationary time series, or can be reduced to such. The anomalies may mimic (from the distributional point of view) the genuine inputs so closely that the contaminated system may not exhibit any visual aberrations, yet they can be detected using the herein proposed method. Supporting probabilistic and statistical results have been rigorously derived, and conditions under which they hold carefully specified. This rigour facilitates confidence when interpreting results and thus when making decisions.

To illustrate how the method works in practice, we have illustrated it using actual time series and also included a numerical experiment under various model and anomaly specifications. The results have shown that the method is robust and is able to detect even tiny systematic anomalies, although, naturally, under longer periods of observation. The method covers light- and heavy-tailed inputs, thus showing its versatility in applications, including those that are associated with data traffic and cyber risks.

Several interesting topics for future study naturally arise from the present paper, throughout which we have so far concentrated on the model
\begin{equation}\label{mod-0}
Y_t=h(X_t,\boldsymbol{\varepsilon}_t)
\end{equation}
with one-dimensional inputs $X_t\in \mathbb{R}$ and outputs $Y_t\in \mathbb{R}$, and $d$-dimensional anomalies $\boldsymbol{\varepsilon}_t \in \mathbb{R}^{d}$. With the time series structure of inputs, contemporary and historical observations enter into the model via the equation $X_t=\langle \boldsymbol{\beta}, \mathbf{Z}_t \rangle $. This point of view together with naturally occurring multidimensional predictors in regression, classification, and, generally, in machine learning lead us to the model
\begin{equation}\label{mod-1}
Y_t=h(\mathbf{X}_t,\boldsymbol{\varepsilon}_t)
\end{equation}
with $k$-dimensional predictors $\mathbf{X}_t\in \mathbb{R}^k$ for some $k\in \mathbb{N} \cup \{+\infty \}$. More generally, problems associated with anomaly detection in parallel computer systems, electrical grid,  and wireless communication architectures such as SISO, SIMO, etc. \citep[e.g.,][]{tv2005,k2017} lead us to the model
\begin{equation}\label{mod-2}
\mathbf{Y}_t=h(\mathbf{X}_t,\boldsymbol{\varepsilon}_t)
\end{equation}
with $q$-valued ($q\in \mathbb{N}$) transfer function $h:\mathbb{R}^{k+d}\to \mathbb{R}^q$ and thus $q$-dimensional outputs $\mathbf{Y}_t$.

The transition from the univariate inputs $X_t$ to the multivariate ones $\mathbf{X}_t$ gives rise to serious mathematical challenges, particularly because of the lack of total ordering in multi-dimensional Euclidean spaces. We feel that the coordinate-wise ordering might lead to a useful anomaly-detection method, but at this moment it looks ad hoc, lacking geometric interpretation and thus intuitive appeal. The optimization problems tackled by \cite {DZ2017}, and \cite{DMZ2019} might give a clue as to what path to take. Alternatively, studies  by \cite{k1997} and \cite{m2002} on ordering multi-dimensional elements could give rise to an effective solution.

The multivariate nature of outputs $\mathbf{Y}_t$ also creates serious statistical-testing and decision-making problems, but we feel that with some effort, such problems can be tackled with the help of e-values studied by \citet{VW2021} and the multiple testing procedures developed by \citet{WR2022}. The e-values are expectation-based versions of the classical p-values. They are simpler to use, thus facilitating multiple hypothesis testing and, in turn, decision making. Hence, the e-values can give rise to impressively powerful and convenient statistical tools in the context of systematic-anomaly detection in, e.g., parallel computer systems, electrical grid, wireless communication architectures, and so on.

Finally, we conclude with the note that anomaly detection problems involve adversarial aspects (e.g., adversarial signal processing, adversarial hypothesis testing) which involve adversaries (intruders) who change their strategies over time. For a glimpse of such research areas, we refer to
\cite{BP2013}, \cite{BR2018}, \cite{BT2018}, \cite{AMB2021}. Naturally, machine learning techniques play a pivotal role in these areas.

\appendix

\section{Graphical illustrations}
\label{graphs}

In this appendix we illustrate the behaviour of $I_n$ and $B_{n,2}$ when  genuine, anomaly-free inputs follow the $\text{ARMA}(1,1)$ time series and the system is affected by iid $\text{Lomax}(\alpha,1)$ anomalies at the input and/or output stages. In the figures that follow, the system is always affected by anomalies. Hence, the index $I_n$ always tends to $1/2$ whereas $B_{n,2}$ grows together with the sample size $n$. Note also that convergence of $I_n$ to $1/2$ is slower when anomaly averages are smaller, meaning that anomalies are less noticeable.  This suggests, naturally, that larger sample sizes are needed to reach desired confidence when making decisions.

\begin{figure}[h!]
    \centering
    \begin{subfigure}[b]{0.36\textwidth}
        \includegraphics[width=\textwidth]{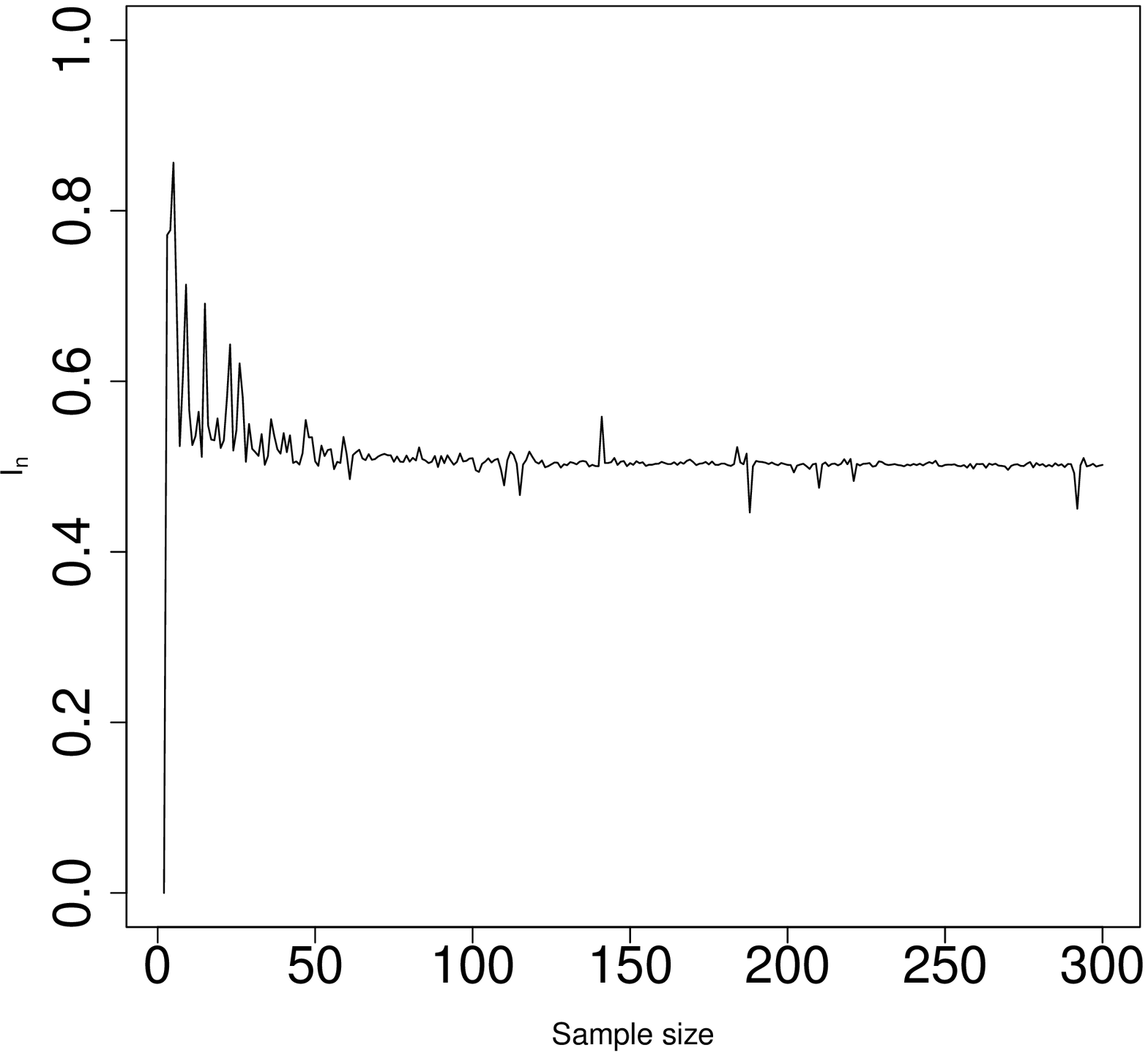}
        \caption{$I_n$ for $h(X_t,0, \epsilon_t)$.}
    \end{subfigure}
\hspace{10mm}
    \begin{subfigure}[b]{0.36\textwidth}
        \includegraphics[width=\textwidth]{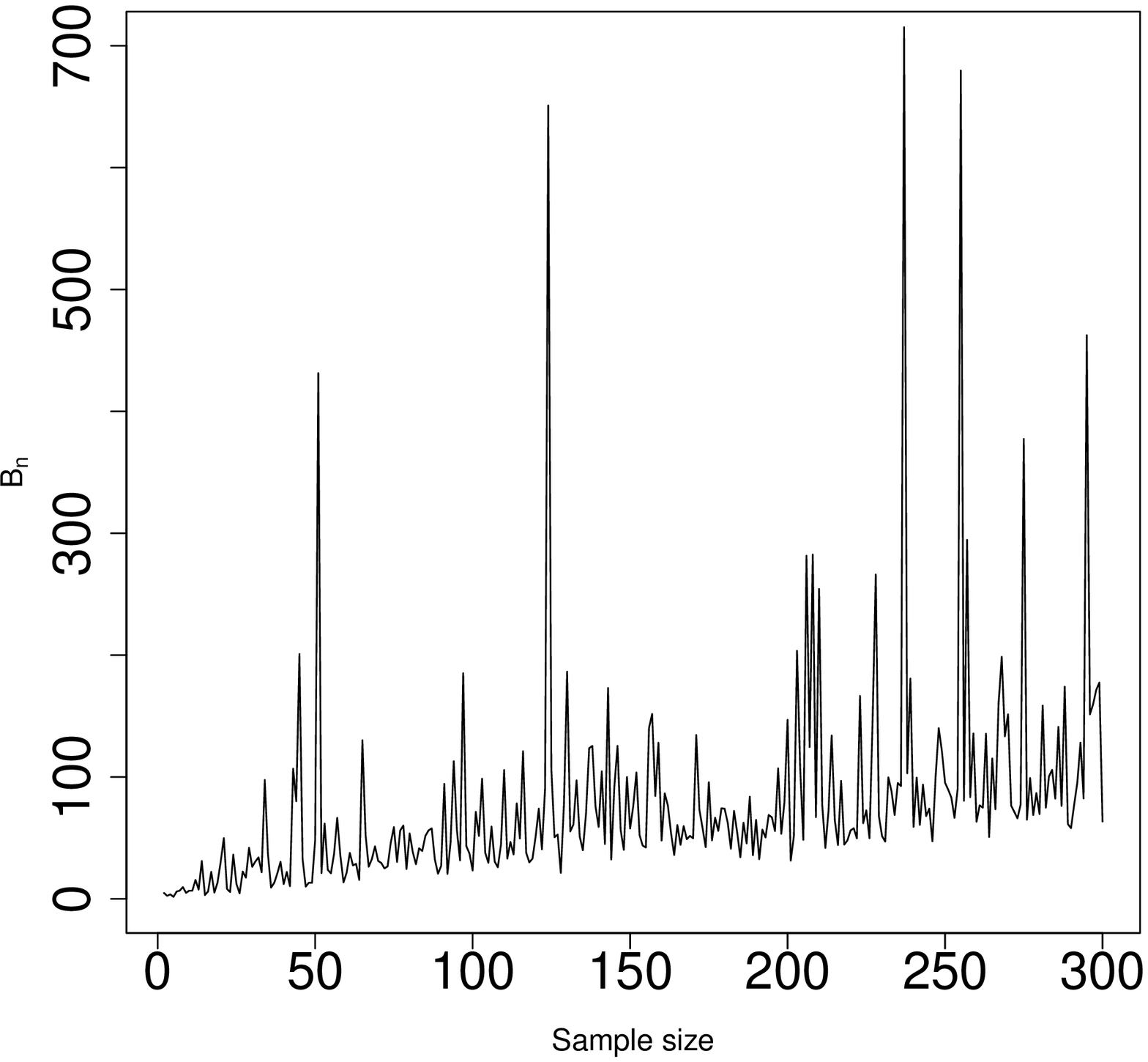}
        \caption{$B_{n,2}$ for $h(X_t,0, \epsilon_t)$.}
    \end{subfigure}
\hspace{10mm}
      %(or a blank line to force the subfigure onto a new line)
    \begin{subfigure}[b]{0.36\textwidth}
        \includegraphics[width=\textwidth]{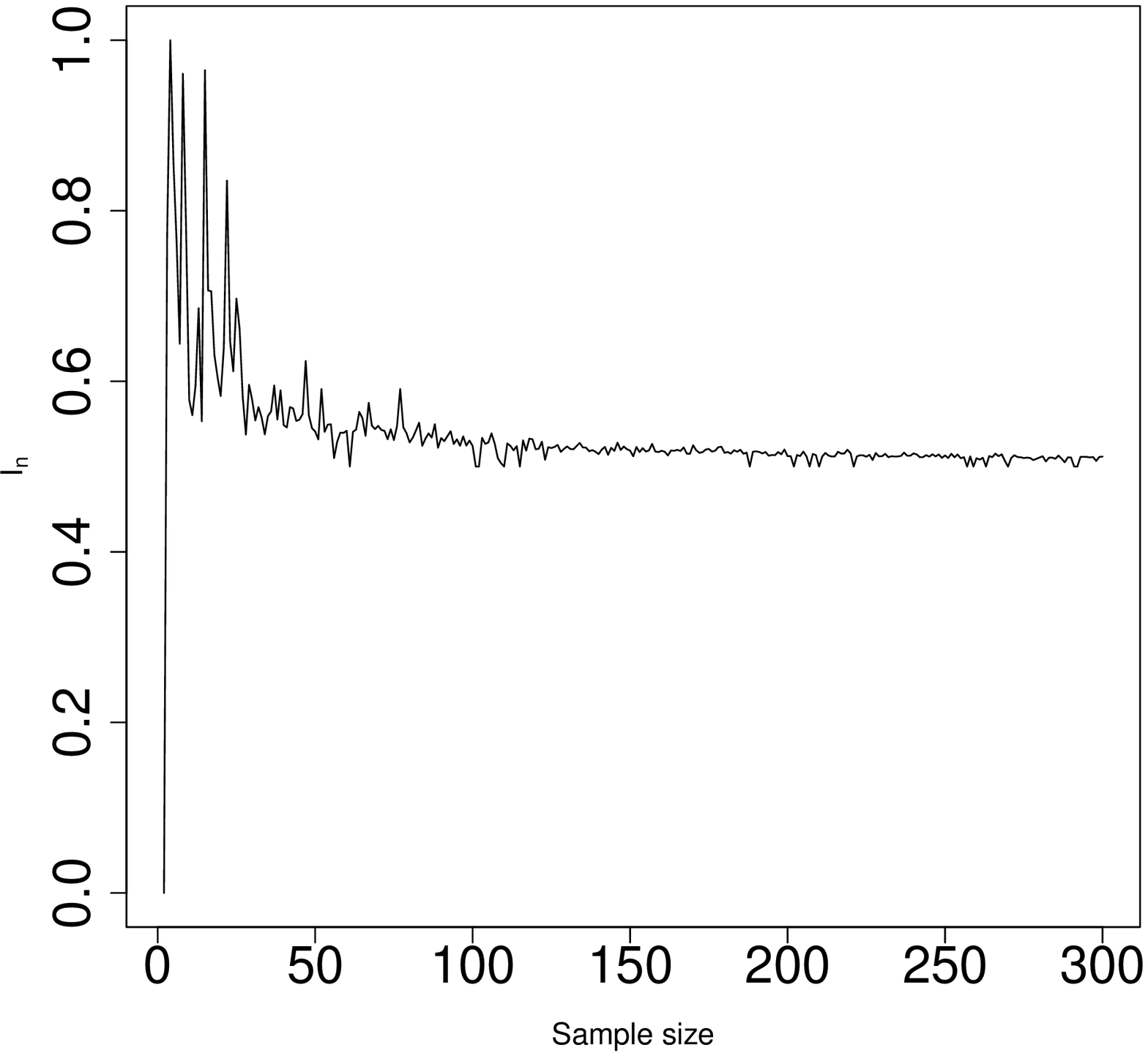}
        \caption{$I_n$ for $h(X_t, \delta_t,0)$.}
    \end{subfigure}
\hspace{10mm}
    %(or a blank line to force the subfigure onto a new line)
    \begin{subfigure}[b]{0.36\textwidth}
        \includegraphics[width=\textwidth]{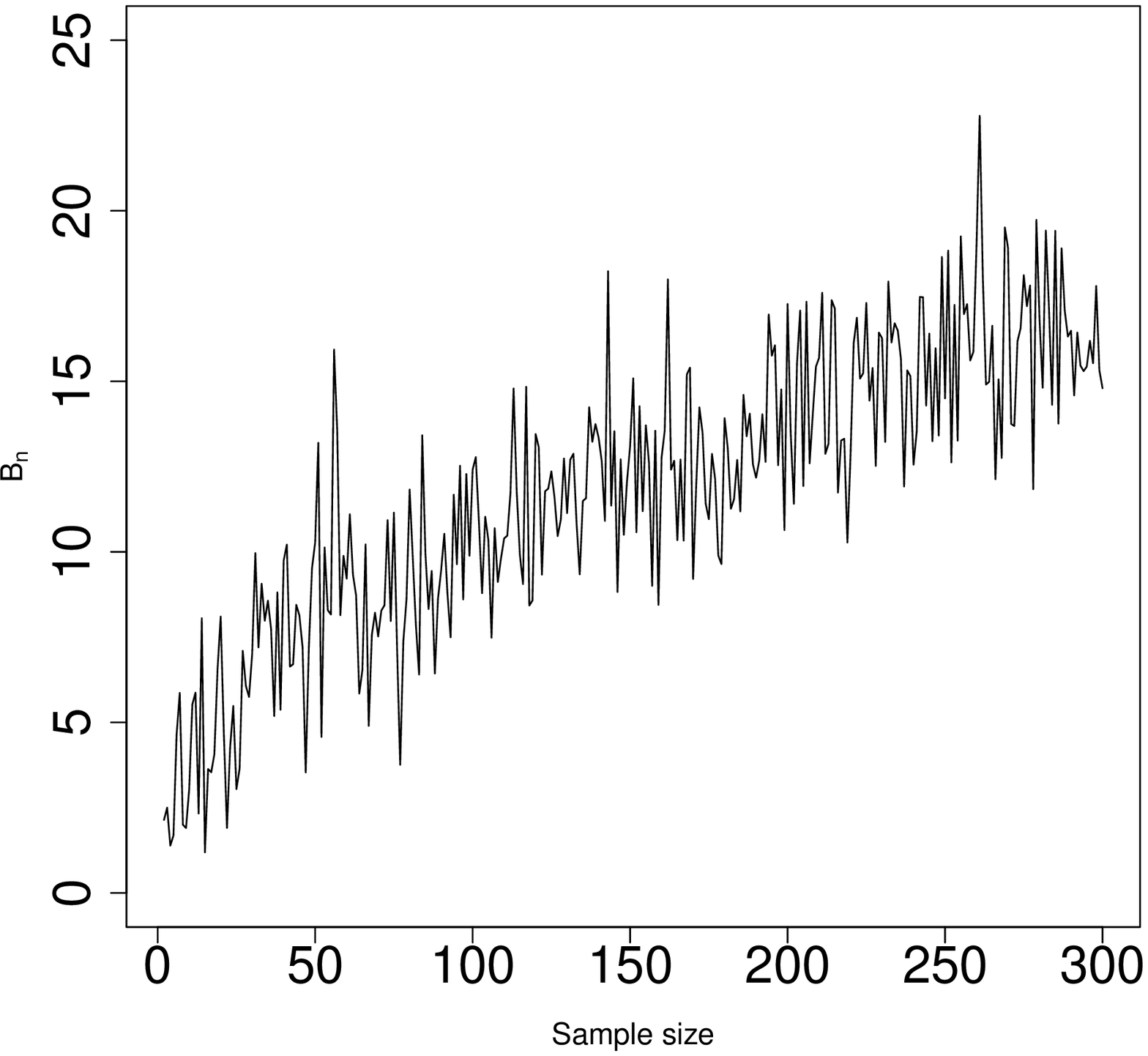}
        \caption{$B_{n,2}$ for $h(X_t, \delta_t,0)$.}
    \end{subfigure}
\hspace{10mm}
      %(or a blank line to force the subfigure onto a new line)
    \begin{subfigure}[b]{0.36\textwidth}
        \includegraphics[width=\textwidth]{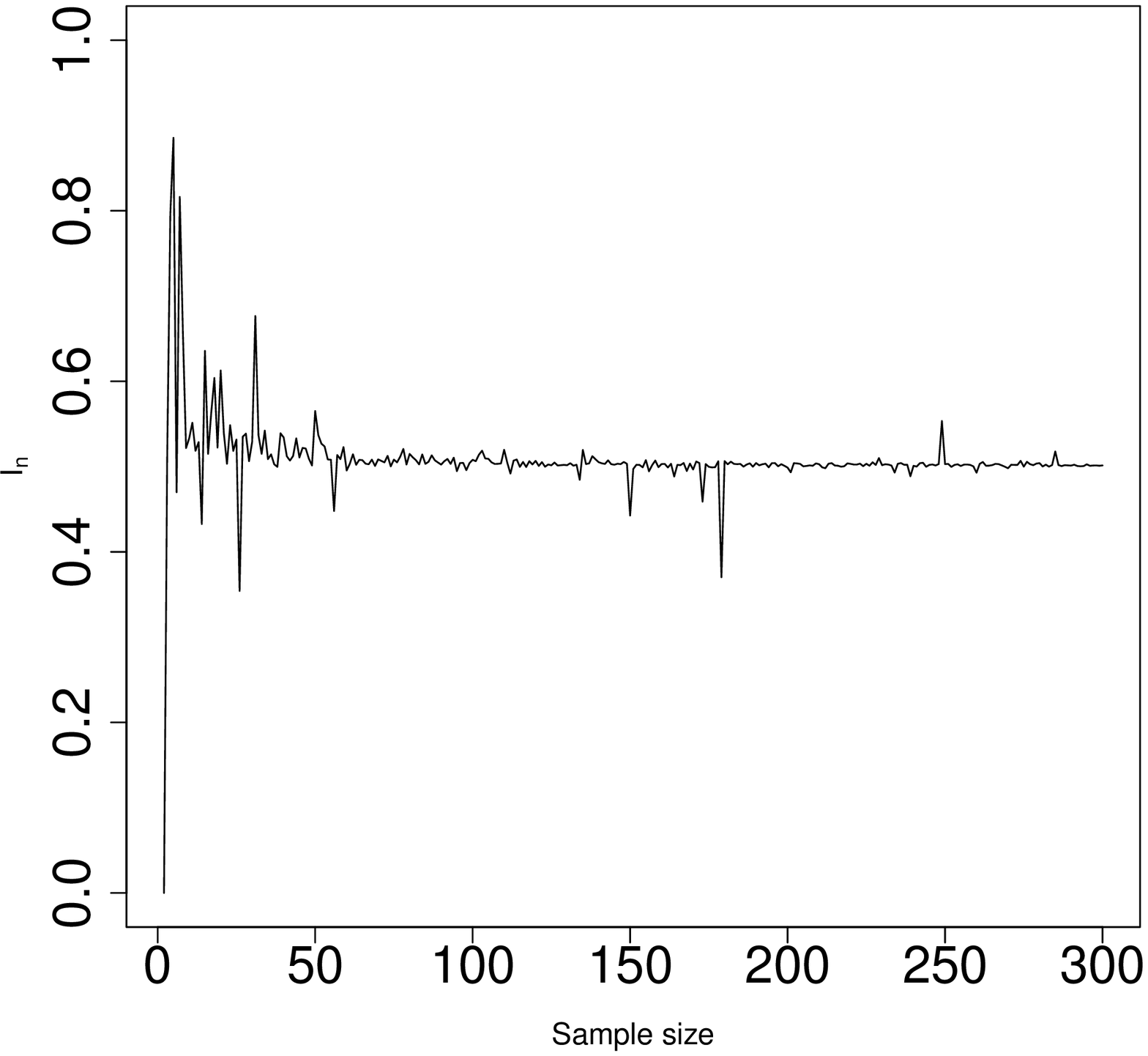}
        \caption{$I_n$ for $h(X_t, \delta_t, \epsilon_t)$.}
    \end{subfigure}
\hspace{10mm}
    %(or a blank line to force the subfigure onto a new line)
    \begin{subfigure}[b]{0.36\textwidth}
        \includegraphics[width=\textwidth]{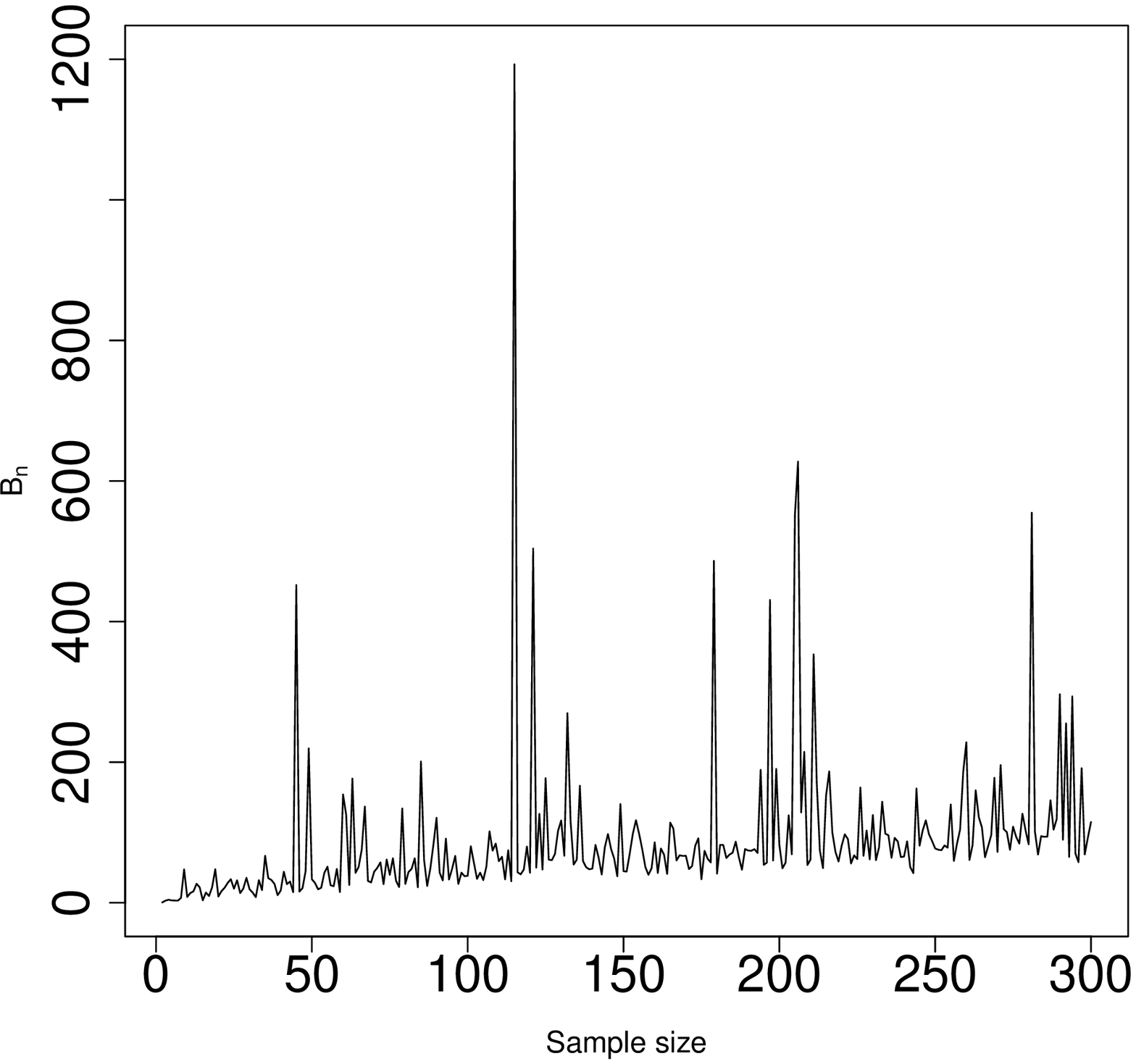}
        \caption{$B_{n,2}$ for $h(X_t, \delta_t, \epsilon_t)$.}
    \end{subfigure}
    \caption{The anomaly-affected indices $I_n$ and $B_{n,2}$ for the strict service range with respect to  $2\le n\le 300$ for $\text{ARMA}(1,1)$ inputs and iid $\text{Lomax}(1.2,1)$ anomalies.}
    \label{stringent-2}
\end{figure}

\begin{figure}[h!]
    \centering
    \begin{subfigure}[b]{0.36\textwidth}
        \includegraphics[width=\textwidth]{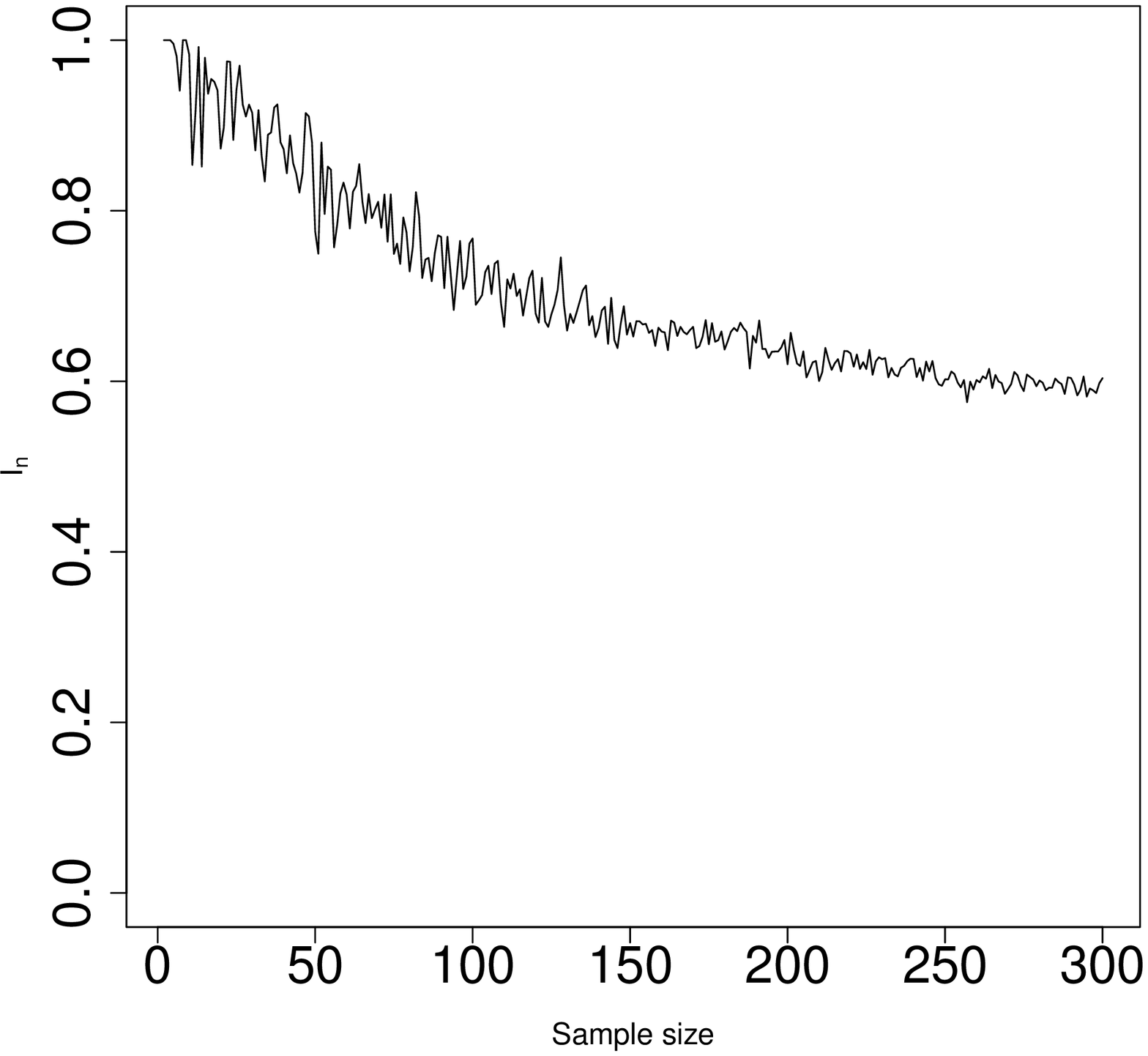}
        \caption{$I_n$ for $h(X_t,0, \epsilon_t)$.}
    \end{subfigure}
\hspace{10mm}
    \begin{subfigure}[b]{0.36\textwidth}
        \includegraphics[width=\textwidth]{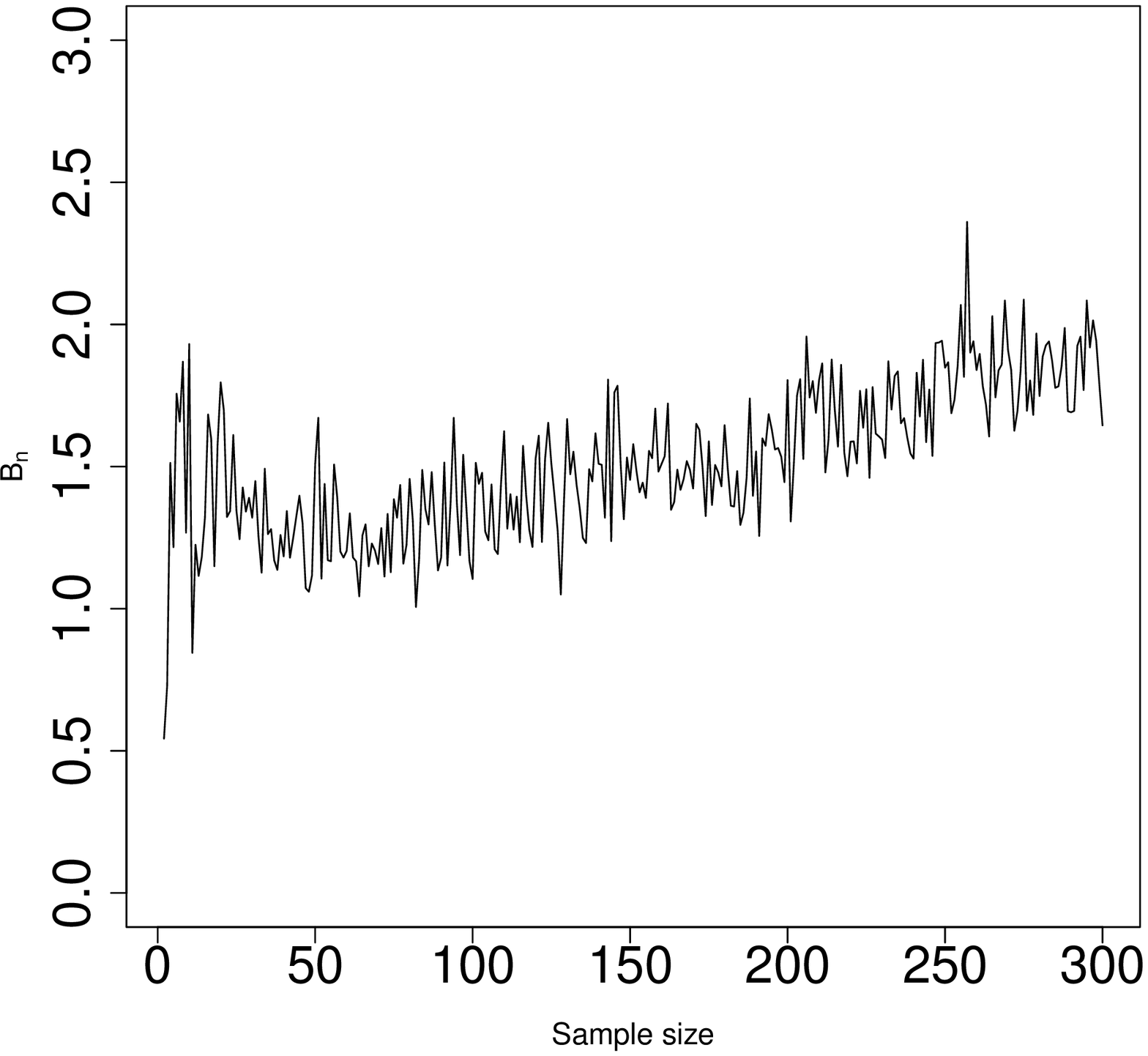}
        \caption{$B_{n,2}$ for $h(X_t,0, \epsilon_t)$.}
    \end{subfigure}
\hspace{10mm}
      %(or a blank line to force the subfigure onto a new line)
    \begin{subfigure}[b]{0.36\textwidth}
        \includegraphics[width=\textwidth]{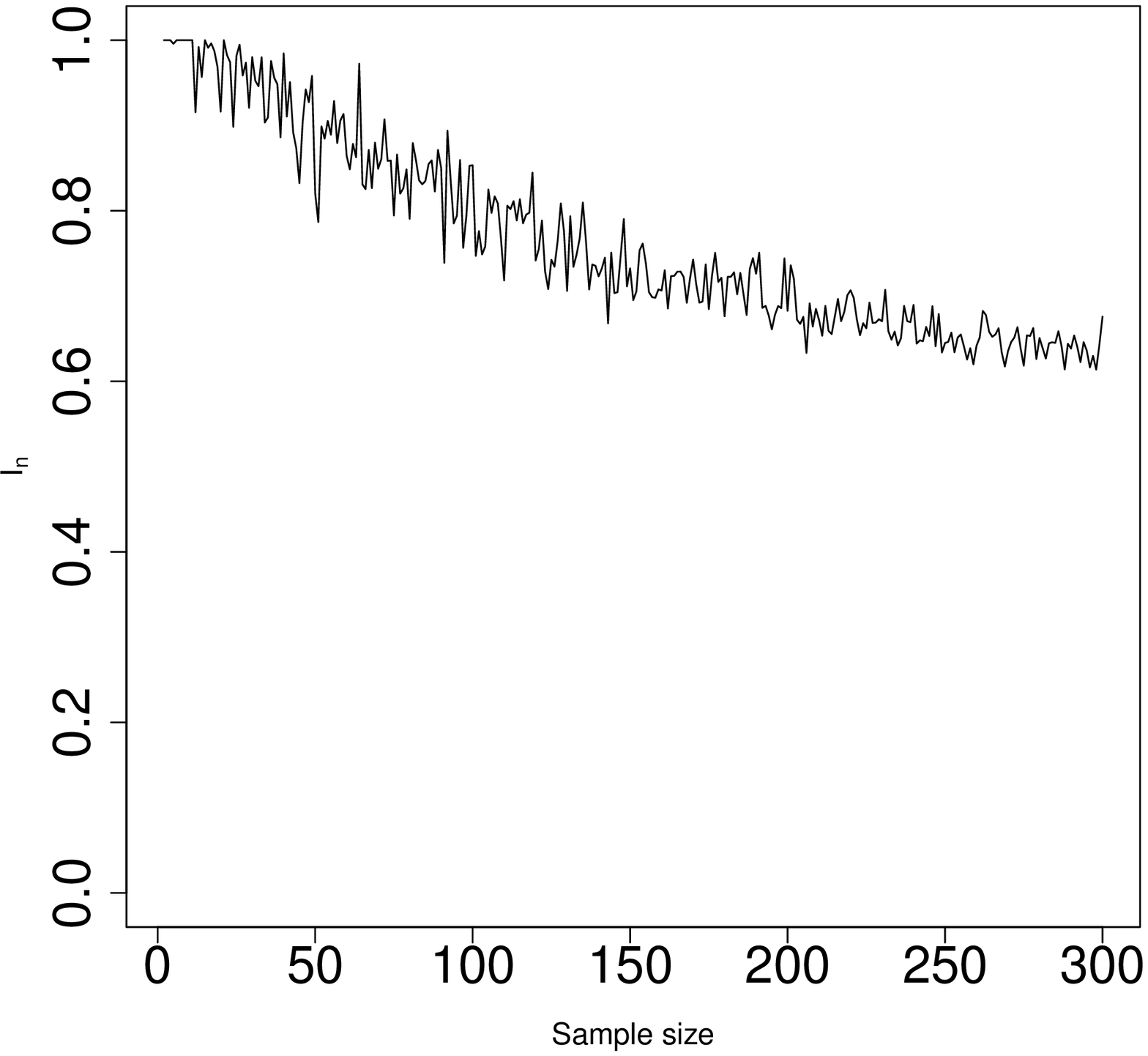}
        \caption{$I_n$ for $h(X_t, \delta_t,0)$.}
    \end{subfigure}
\hspace{10mm}
    %(or a blank line to force the subfigure onto a new line)
    \begin{subfigure}[b]{0.36\textwidth}
        \includegraphics[width=\textwidth]{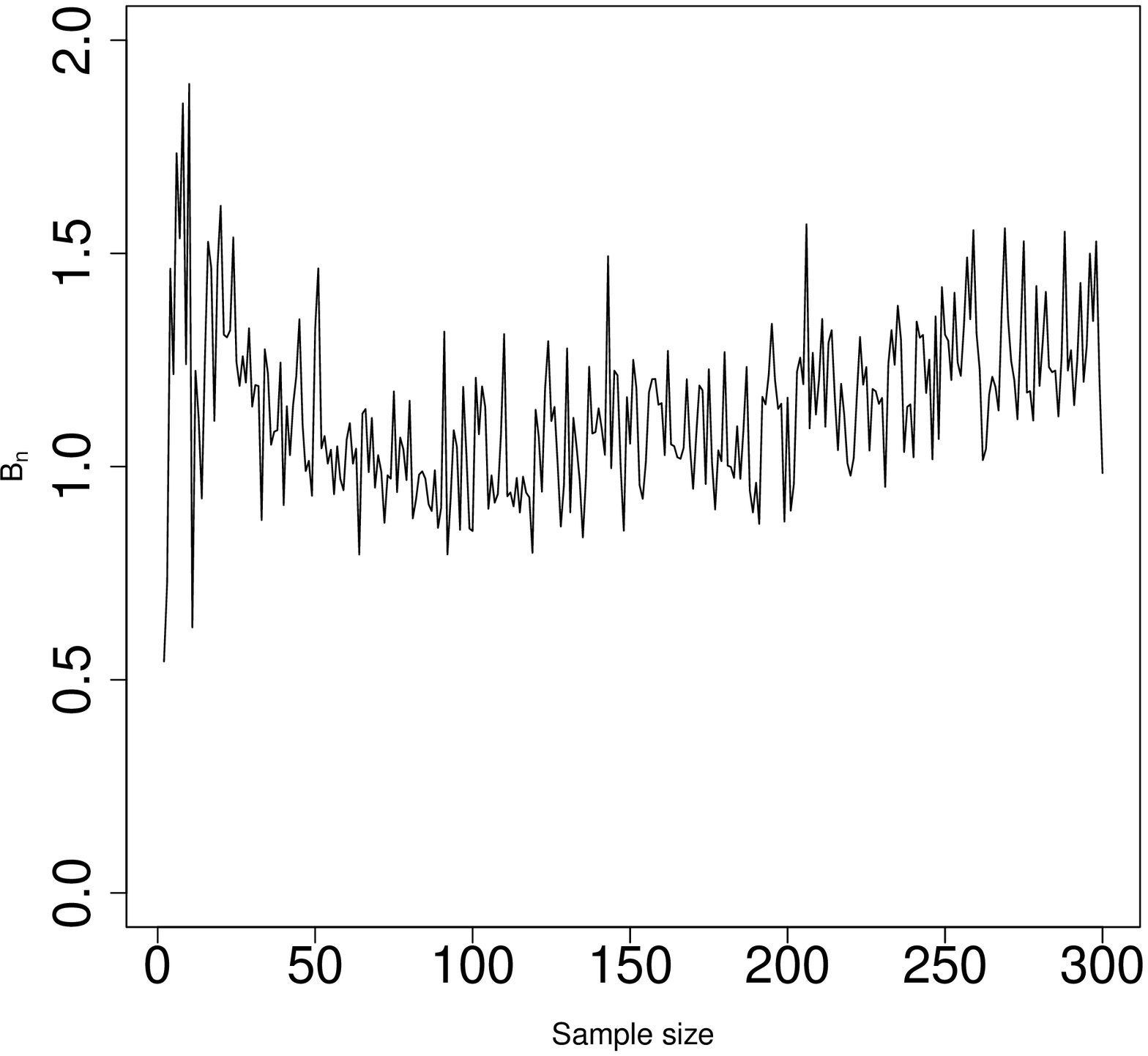}
        \caption{$B_{n,2}$ for $h(X_t, \delta_t,0)$.}
    \end{subfigure}
\hspace{10mm}
      %(or a blank line to force the subfigure onto a new line)
    \begin{subfigure}[b]{0.36\textwidth}
        \includegraphics[width=\textwidth]{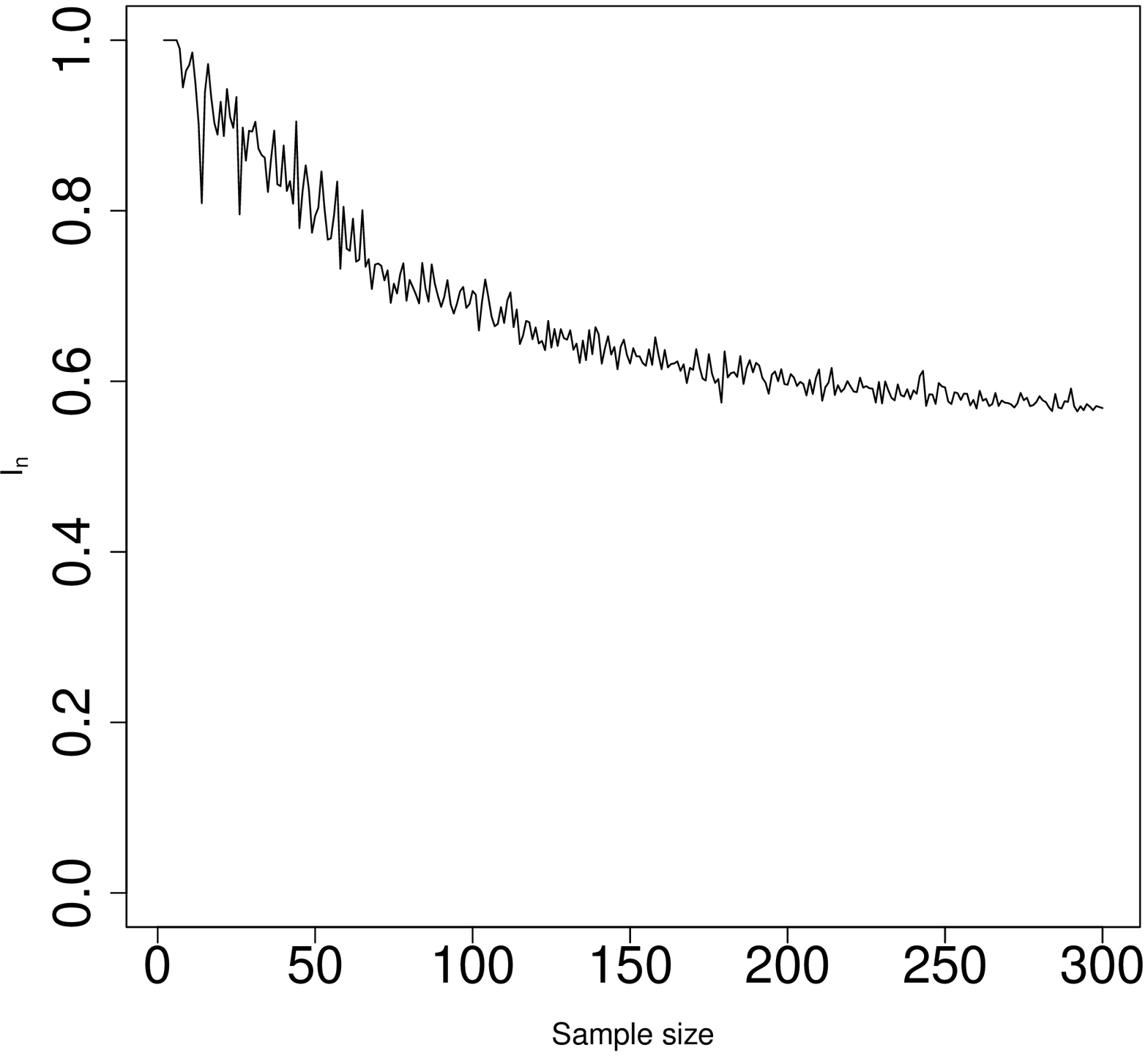}
        \caption{$I_n$ for $h(X_t, \delta_t, \epsilon_t)$.}
    \end{subfigure}
\hspace{10mm}
    %(or a blank line to force the subfigure onto a new line)
    \begin{subfigure}[b]{0.36\textwidth}
        \includegraphics[width=\textwidth]{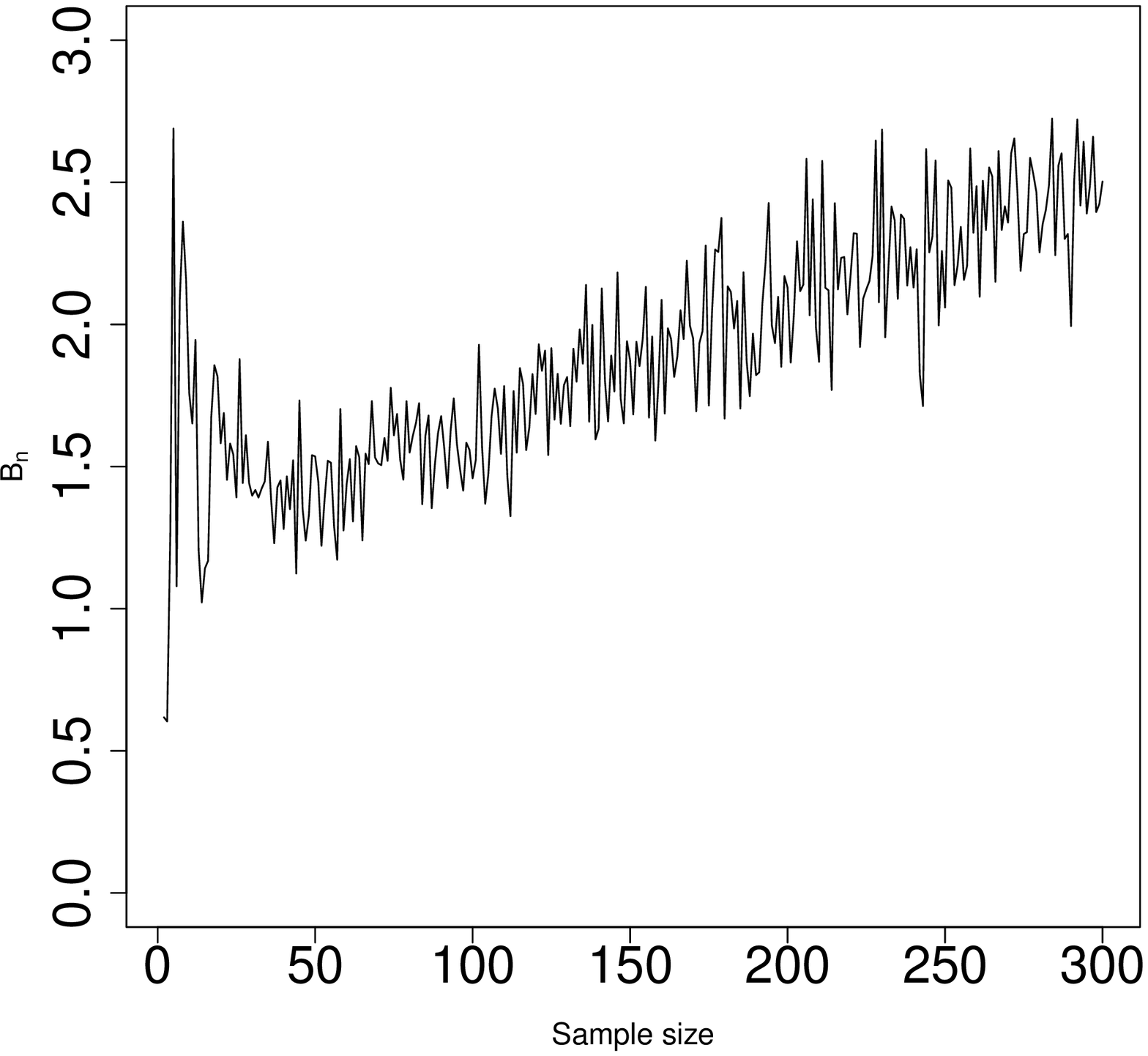}
        \caption{$B_{n,2}$ for $h(X_t, \delta_t, \epsilon_t)$.}
    \end{subfigure}
    \caption{The anomaly-affected indices $I_n$ and $B_{n,2}$ for the strict service range with respect to  $2\le n\le 300$  for $\text{ARMA}(1,1)$ inputs and iid  $\text{Lomax}(11,1)$ anomalies.}
    \label{stringent-2-11}
\end{figure}

\begin{figure}[h!]
    \centering
    \begin{subfigure}[b]{0.36\textwidth}
        \includegraphics[width=\textwidth]{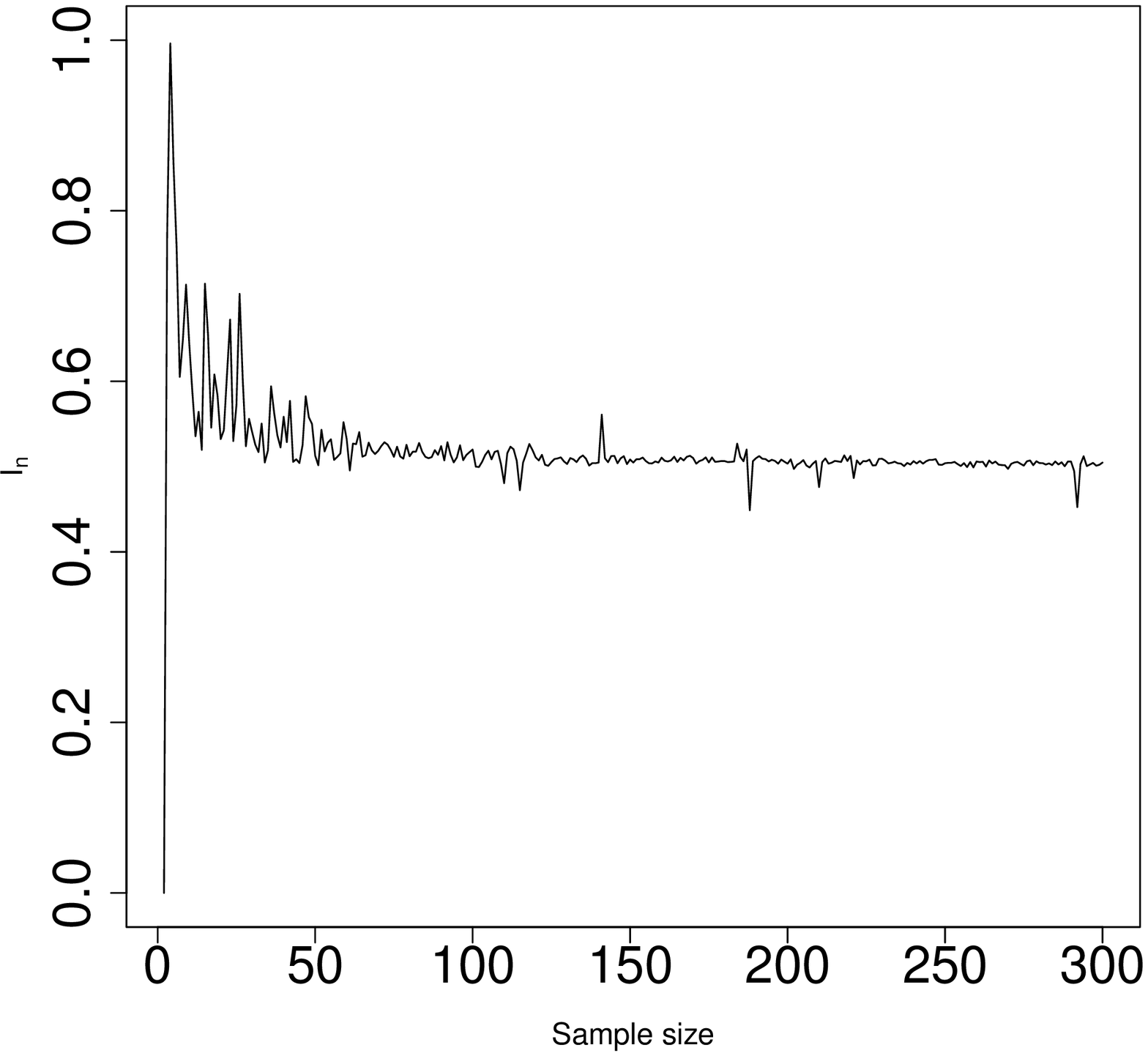}
        \caption{$I_n$ for $h(X_t,0, \epsilon_t)$.}
    \end{subfigure}
\hspace{10mm}
    \begin{subfigure}[b]{0.36\textwidth}
        \includegraphics[width=\textwidth]{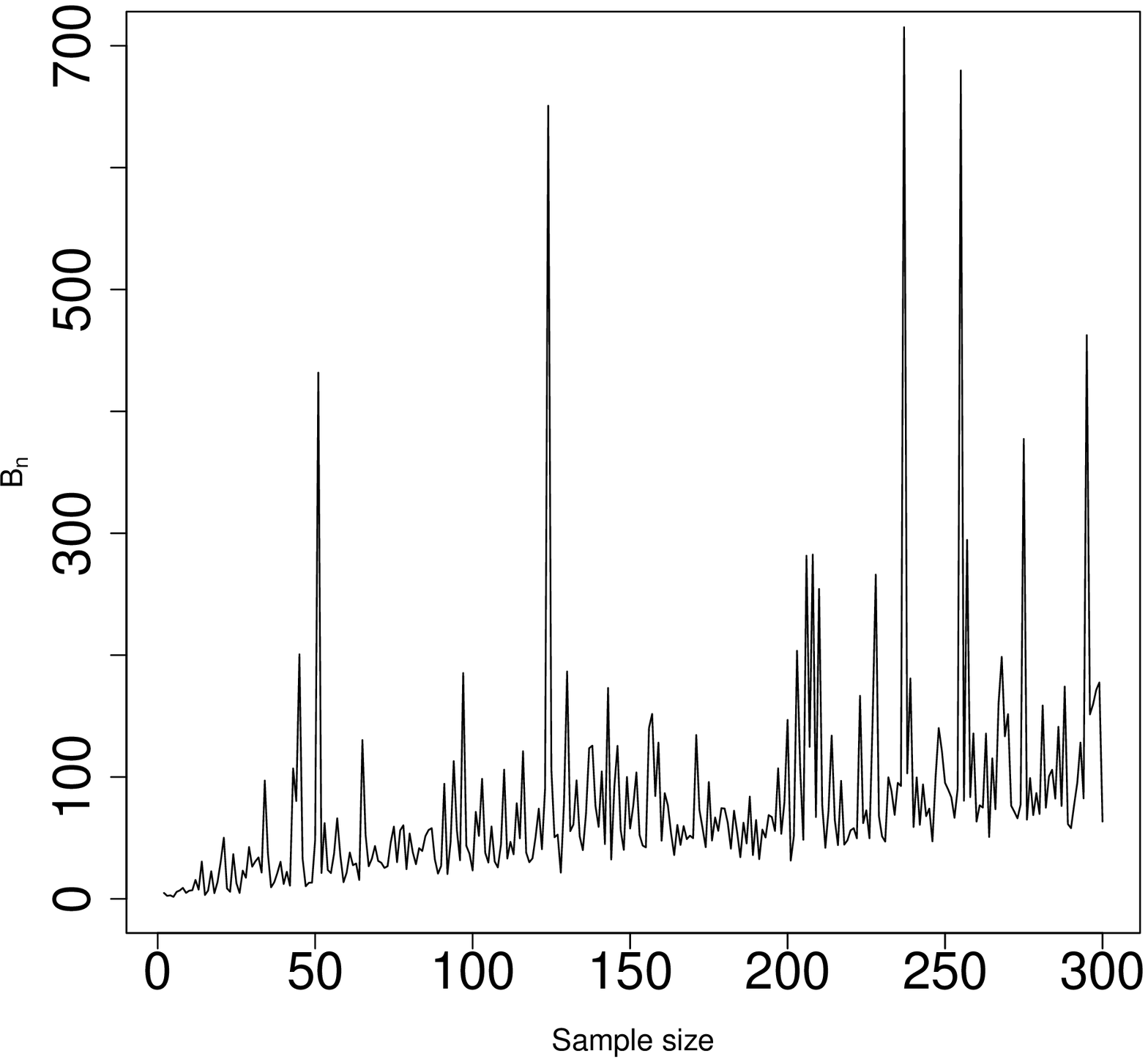}
        \caption{$B_{n,2}$ for $h(X_t,0, \epsilon_t)$.}
    \end{subfigure}
\hspace{10mm}
      %(or a blank line to force the subfigure onto a new line)
    \begin{subfigure}[b]{0.36\textwidth}
        \includegraphics[width=\textwidth]{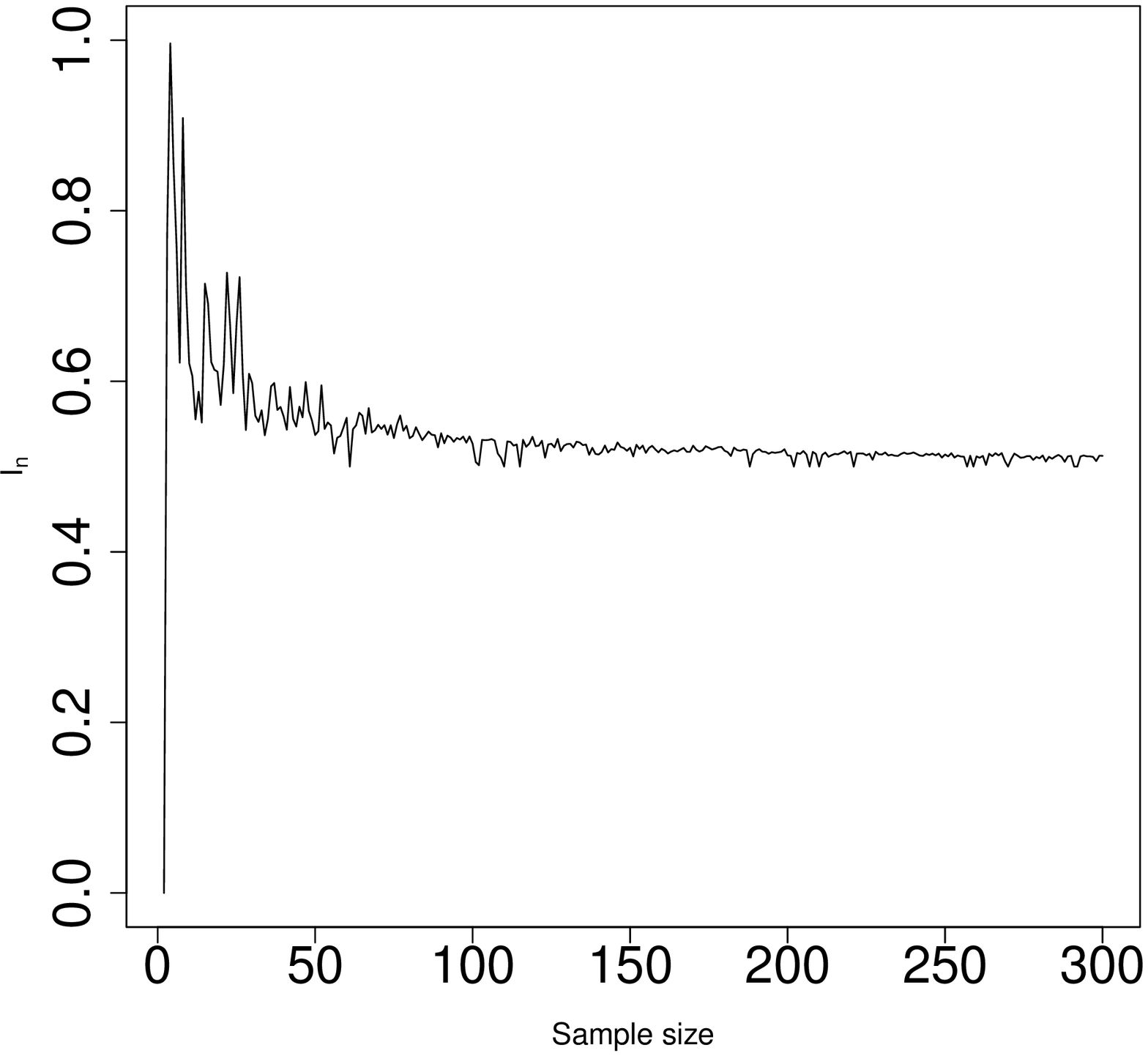}
        \caption{$I_n$ for $h(X_t, \delta_t,0)$.}
    \end{subfigure}
\hspace{10mm}
    %(or a blank line to force the subfigure onto a new line)
    \begin{subfigure}[b]{0.36\textwidth}
        \includegraphics[width=\textwidth]{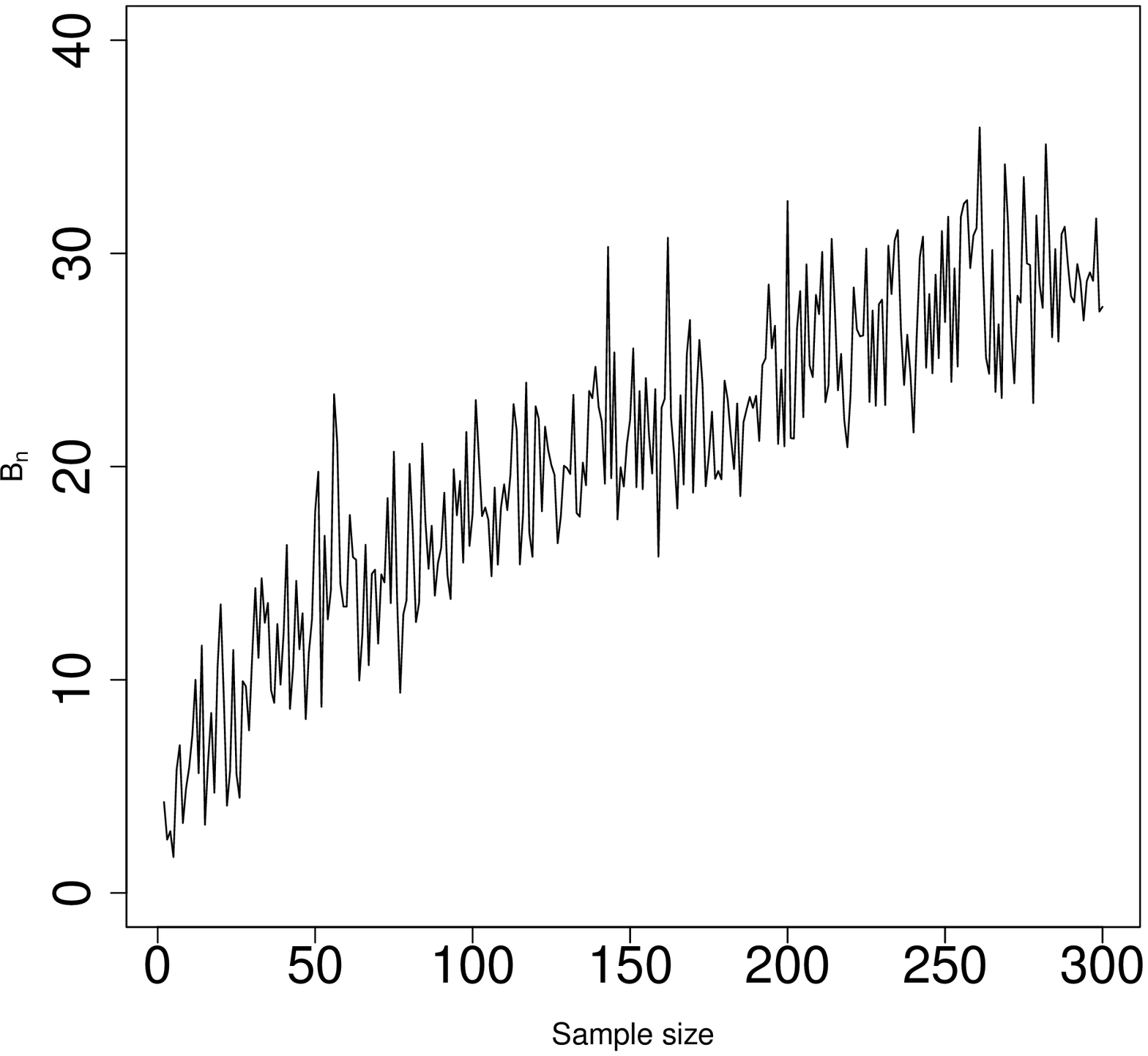}
        \caption{$B_{n,2}$ for $h(X_t, \delta_t,0)$.}
    \end{subfigure}
\hspace{10mm}
      %(or a blank line to force the subfigure onto a new line)
    \begin{subfigure}[b]{0.36\textwidth}
        \includegraphics[width=\textwidth]{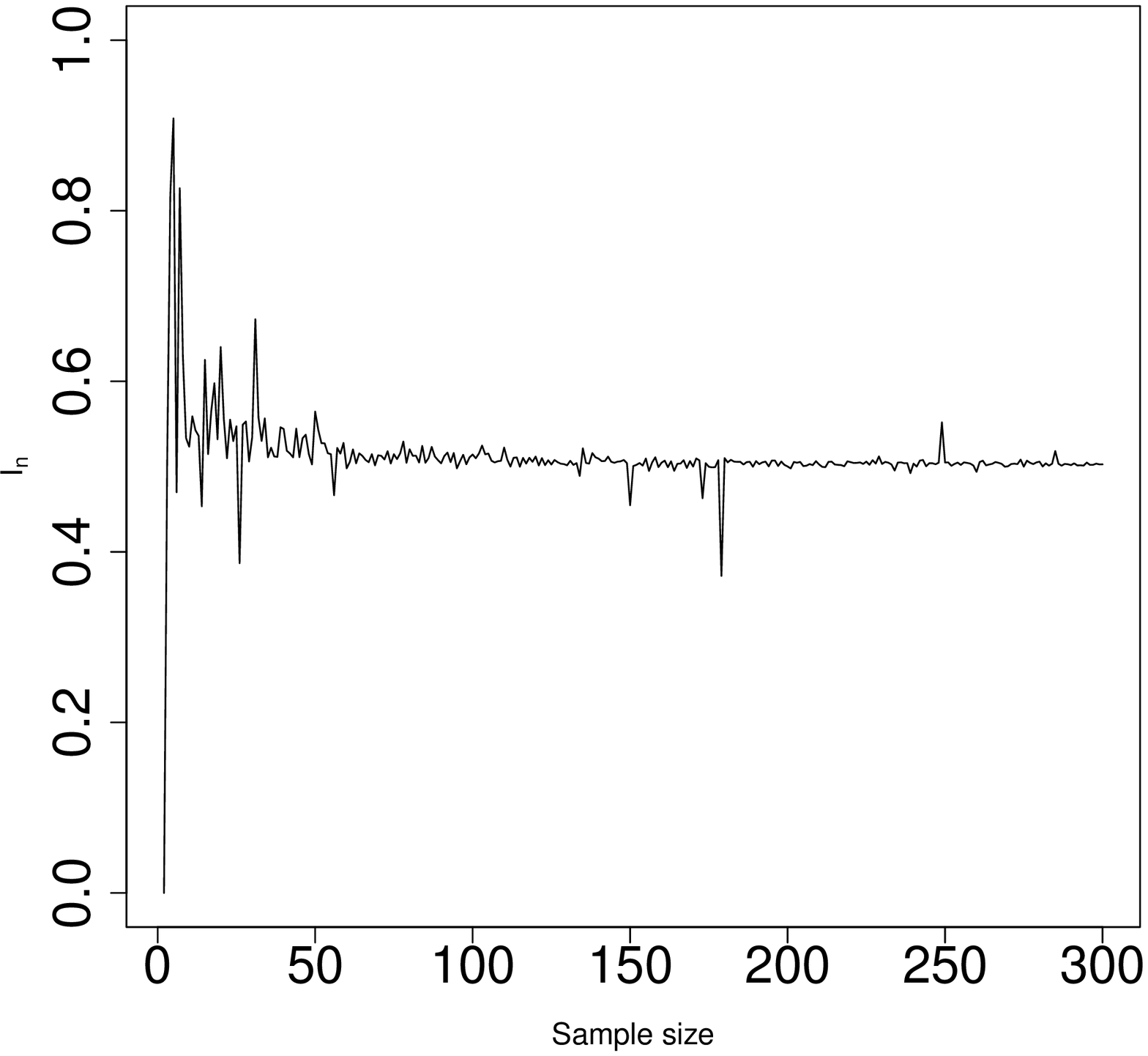}
        \caption{$I_n$ for $h(X_t, \delta_t, \epsilon_t)$.}
    \end{subfigure}
\hspace{10mm}
    %(or a blank line to force the subfigure onto a new line)
    \begin{subfigure}[b]{0.36\textwidth}
        \includegraphics[width=\textwidth]{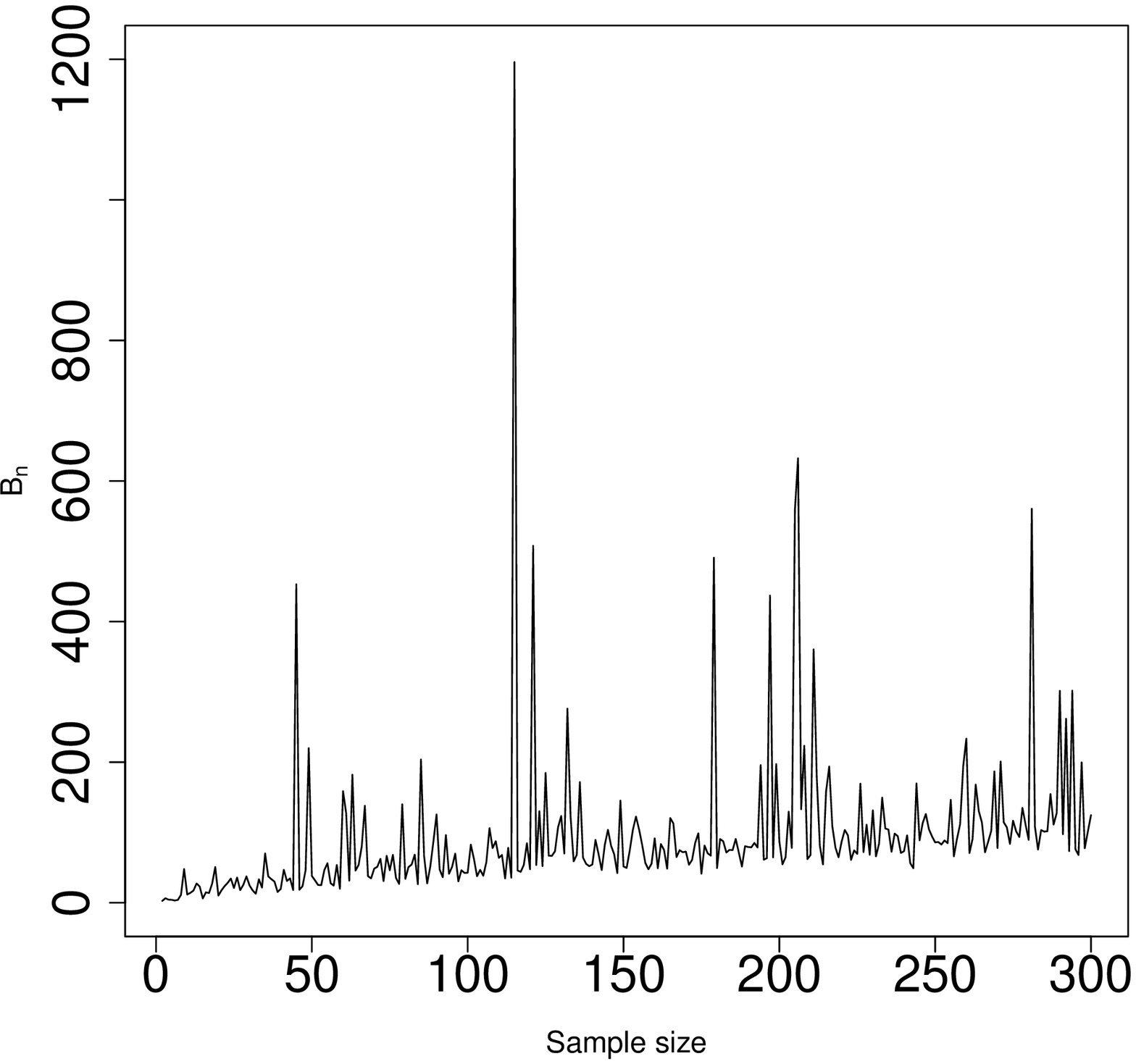}
        \caption{$B_{n,2}$ for $h(X_t, \delta_t, \epsilon_t)$.}
    \end{subfigure}
    \caption{The anomaly-affected indices $I_n$ and $B_{n,2}$ for the satisfactory service range with respect to  $2\le n\le 300$ for $\text{ARMA}(1,1)$ inputs and iid  $\text{Lomax}(1.2,1)$ anomalies.}
    \label{standard-3}
\end{figure}

\begin{figure}[h!]
    \centering
    \begin{subfigure}[b]{0.36\textwidth}
        \includegraphics[width=\textwidth]{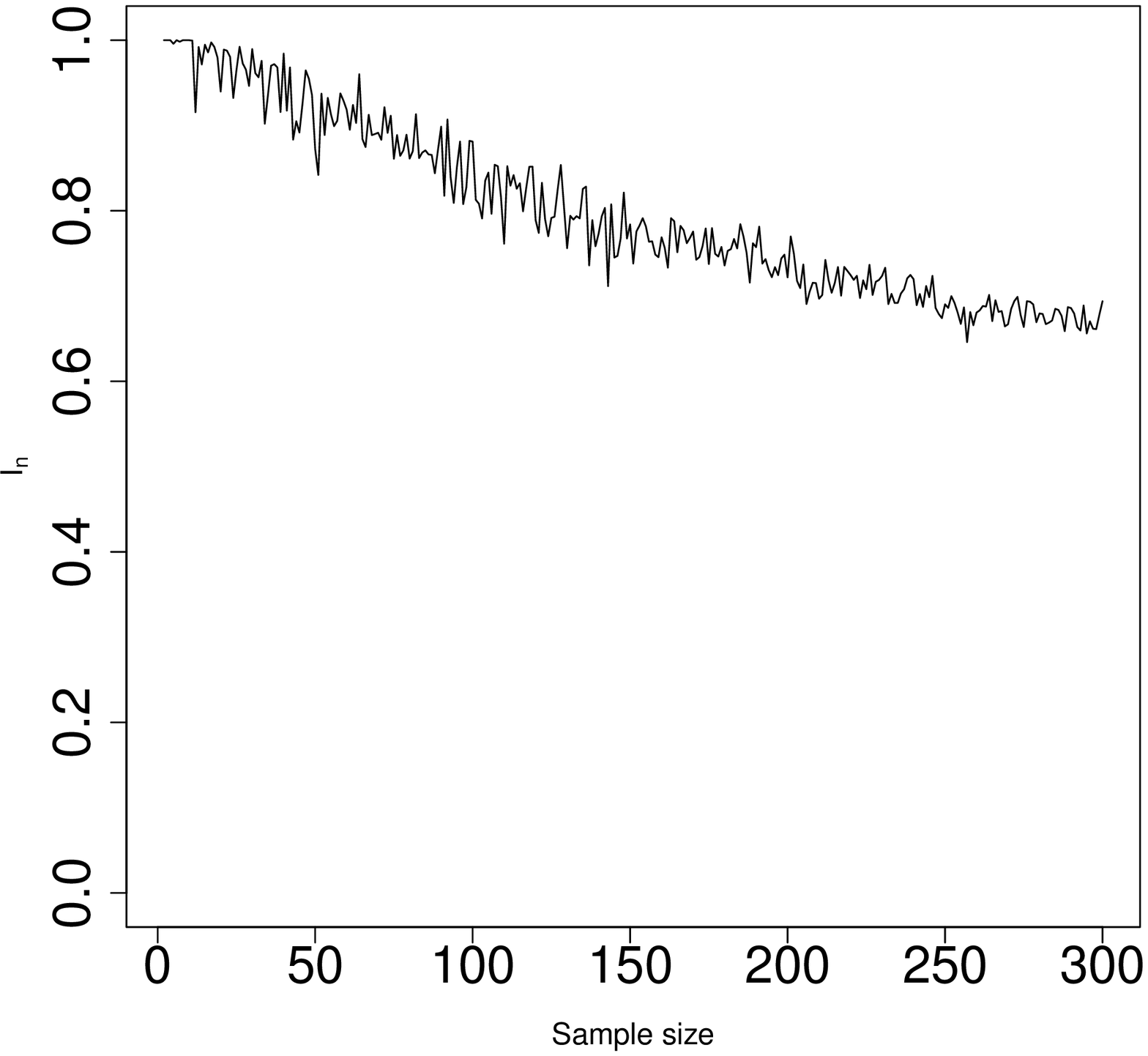}
        \caption{$I_n$ for $h(X_t,0, \epsilon_t)$.}
    \end{subfigure}
\hspace{10mm}
    \begin{subfigure}[b]{0.36\textwidth}
        \includegraphics[width=\textwidth]{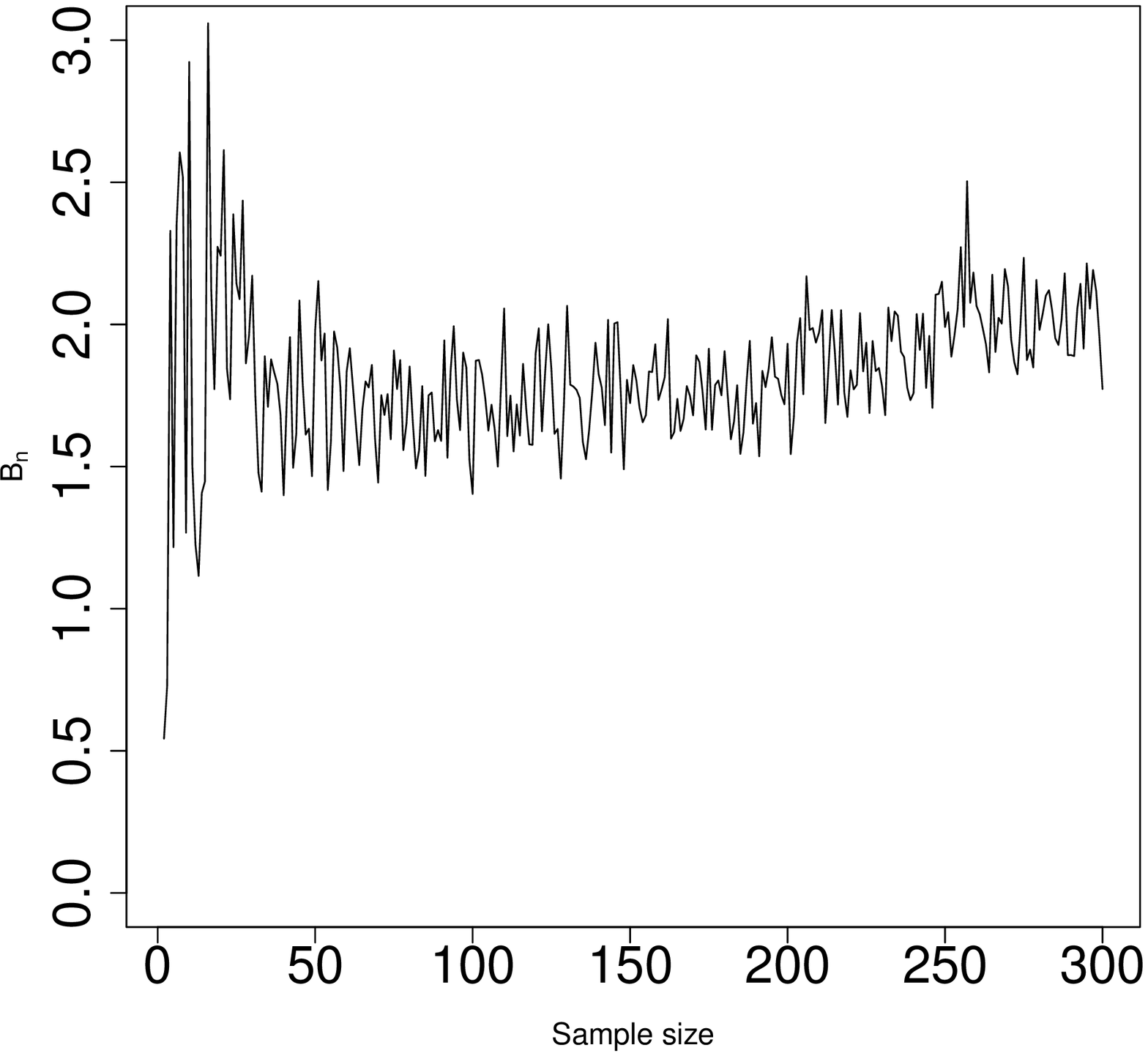}
        \caption{$B_{n,2}$ for $h(X_t,0, \epsilon_t)$.}
    \end{subfigure}
\hspace{10mm}
      %(or a blank line to force the subfigure onto a new line)
    \begin{subfigure}[b]{0.36\textwidth}
        \includegraphics[width=\textwidth]{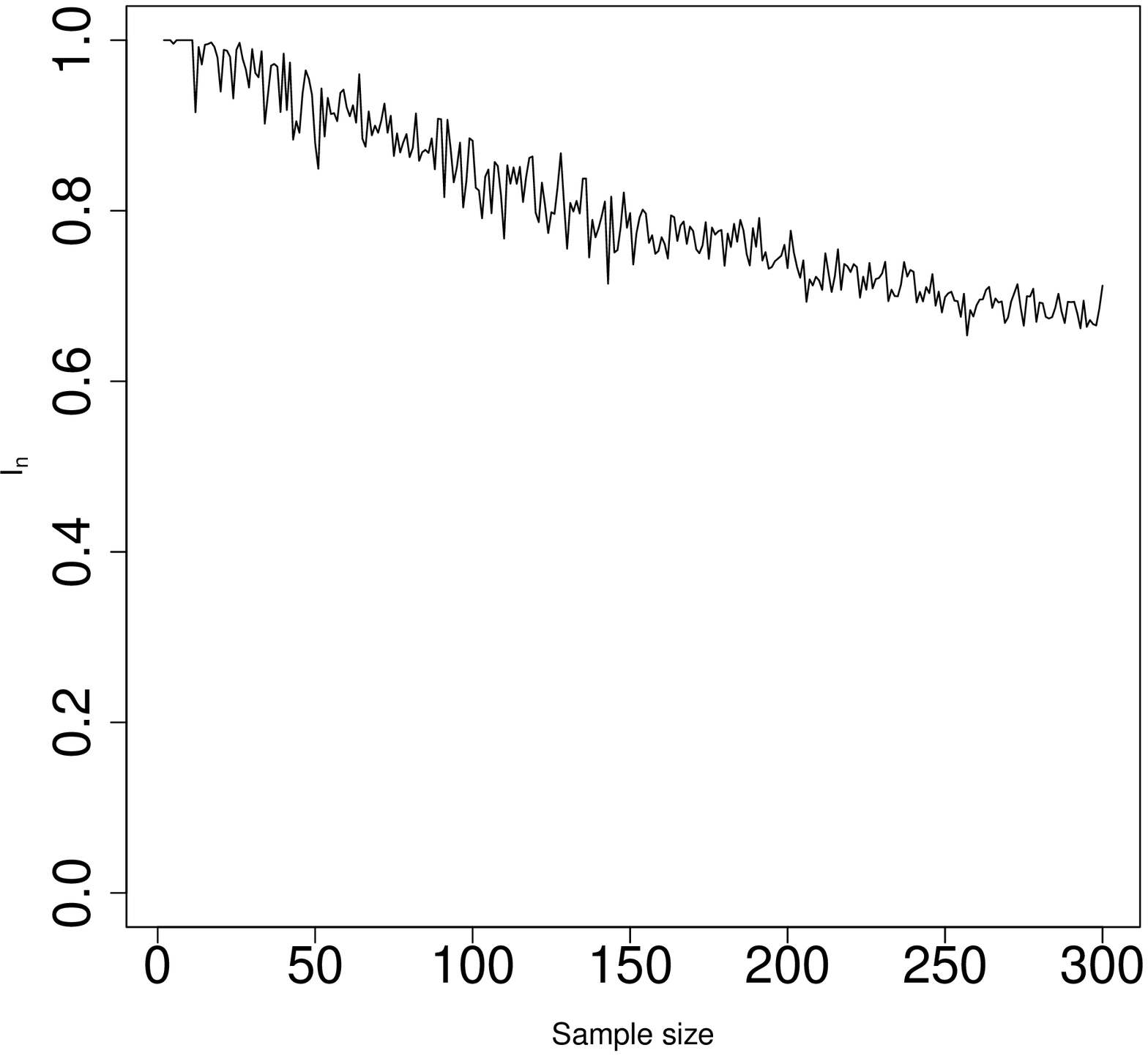}
        \caption{$I_n$ for $h(X_t, \delta_t,0)$.}
    \end{subfigure}
\hspace{10mm}
    %(or a blank line to force the subfigure onto a new line)
    \begin{subfigure}[b]{0.36\textwidth}
        \includegraphics[width=\textwidth]{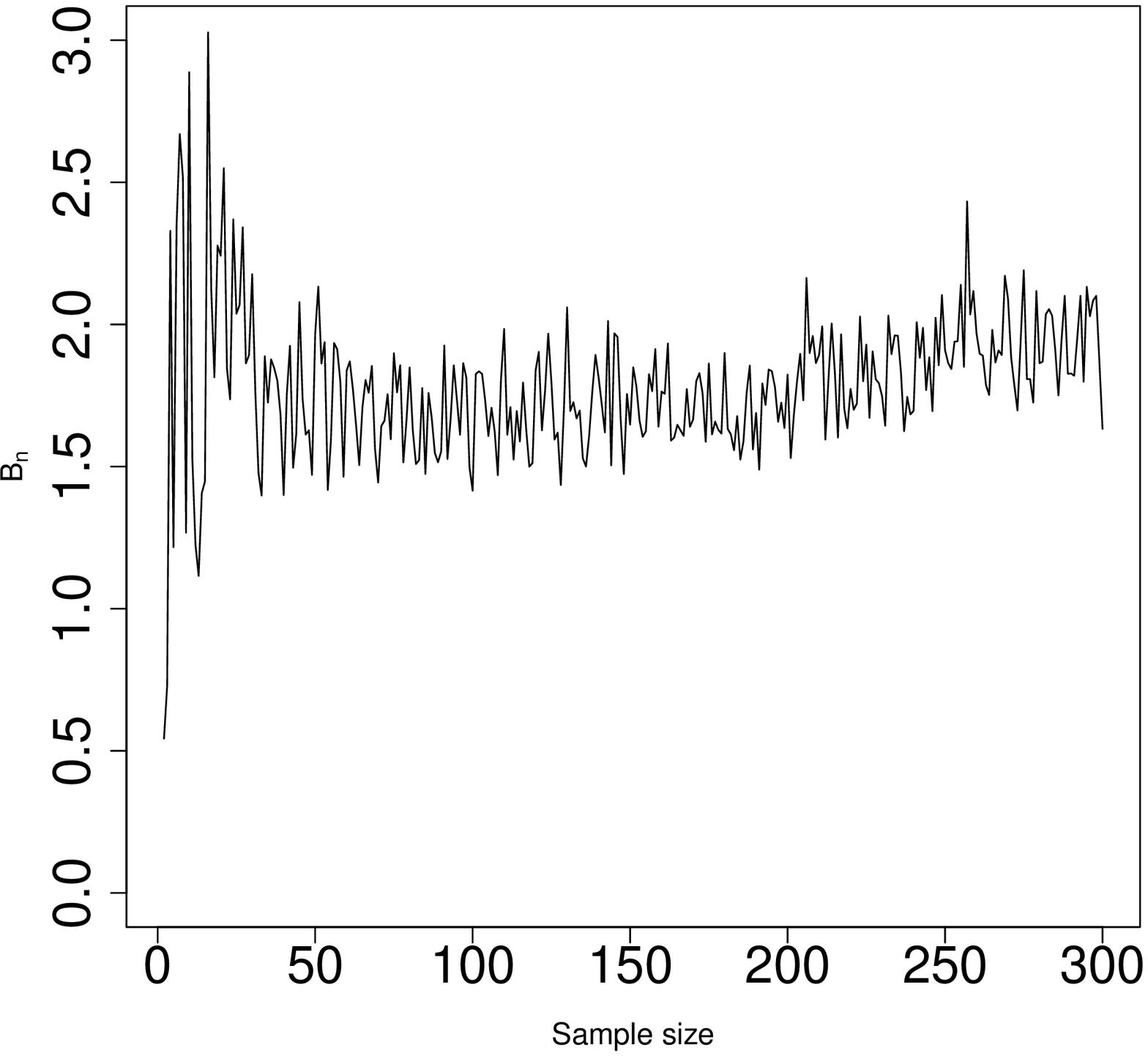}
        \caption{$B_{n,2}$ for $h(X_t, \delta_t,0)$.}
    \end{subfigure}
\hspace{10mm}
      %(or a blank line to force the subfigure onto a new line)
    \begin{subfigure}[b]{0.36\textwidth}
        \includegraphics[width=\textwidth]{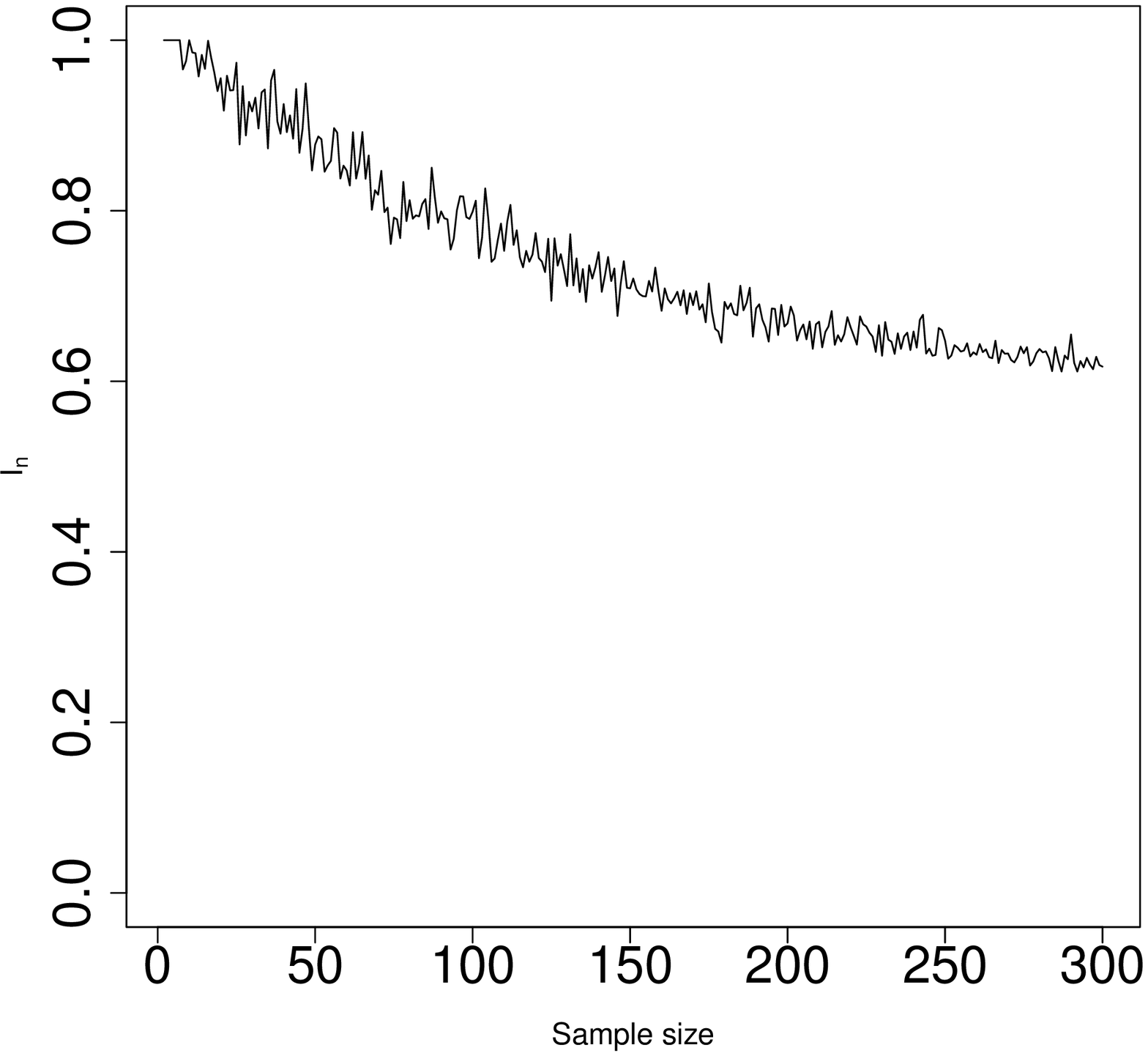}
        \caption{$I_n$ for $h(X_t, \delta_t, \epsilon_t)$.}
    \end{subfigure}
\hspace{10mm}
    %(or a blank line to force the subfigure onto a new line)
    \begin{subfigure}[b]{0.36\textwidth}
        \includegraphics[width=\textwidth]{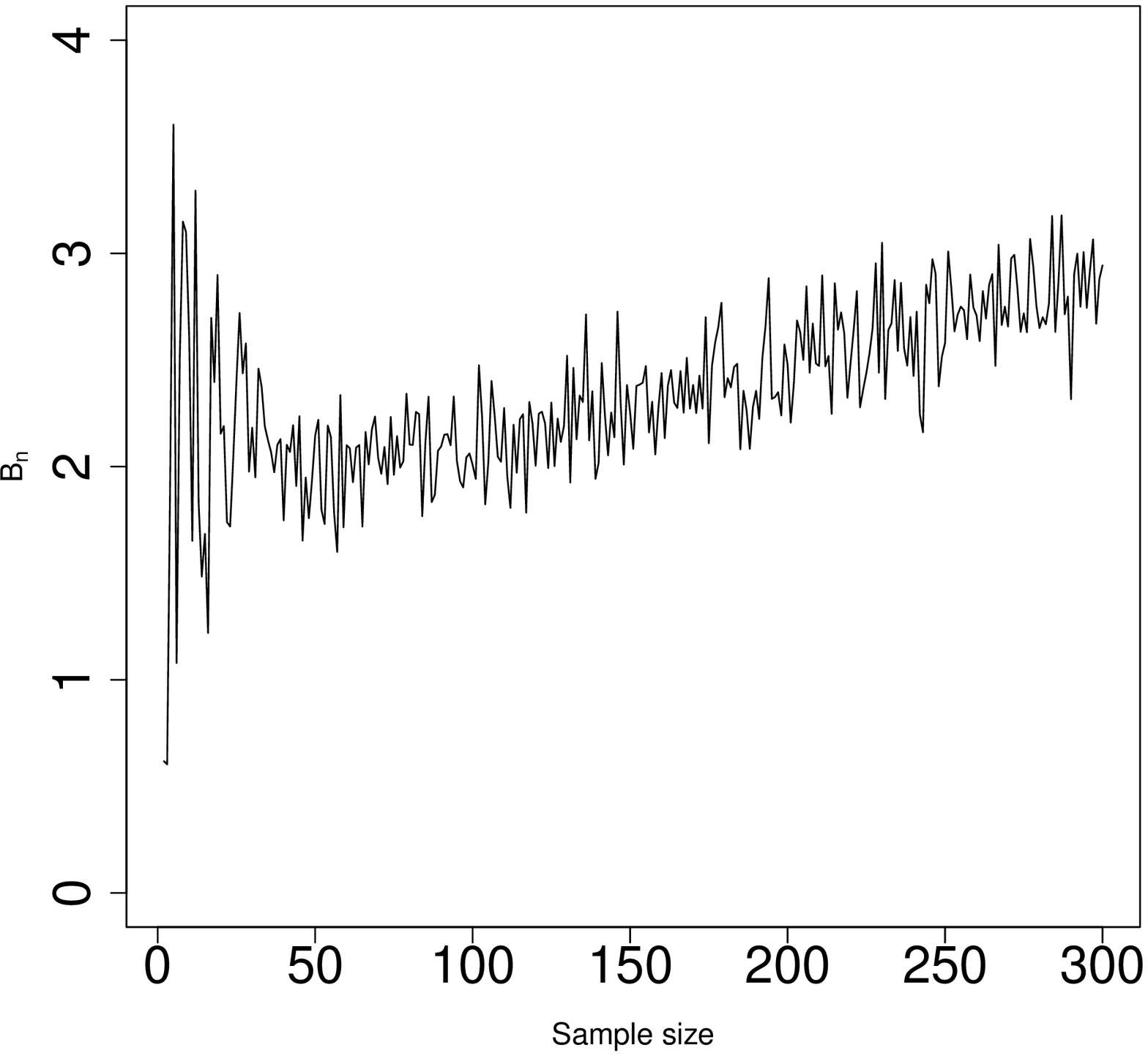}
        \caption{$B_{n,2}$ for $h(X_t, \delta_t, \epsilon_t)$.}
    \end{subfigure}
    \caption{The anomaly-affected indices $I_n$ and $B_{n,2}$ for the satisfactory service range with respect to  $2\le n\le 300$  for $\text{ARMA}(1,1)$ inputs and iid  $\text{Lomax}(11,1)$ anomalies.}
    \label{standard-3-11}
\end{figure}

\newpage

\section{Technical details}
\label{proofs}

To prove Theorem~\ref{th-1aa}, we need a lemma, which we shall also use when proving other results.

\begin{lemma}\label{lemma-GC}
Let $\xi_t$, $t\in \mathbb{Z}$, be identically distributed random variables such that  $\mathbb{E}(|\xi_t|^p)<\infty $ for some $p\ge 1$. Then
\begin{equation}\label{gc-result}
n^{-1/p}\big (\mathbb{E}(\xi_{n:n})-\mathbb{E}(\xi_{1:n})\big )=O(1)
\end{equation}
when $n\to \infty $, and thus $n^{-1/p}(\xi_{n:n}-\xi_{1:n})=O_{\mathbb{P}}(1)$, where $\xi_{1:n}\le \xi_{2:n}\le \cdots \le \xi_{n:n}$ are the order statistics of $\xi_{1}, \xi_{2}, \dots , \xi_{n}$.
\end{lemma}

\begin{proof}
Using bounds (7) of~\citet{GC1992} and then applying H\"{o}lder's inequality, we have
\begin{align*}
\mathbb{E}(\xi_{n:n})-\mathbb{E}(\xi_{1:n})
&\le n\int_{1-1/n}^{1}F_{\xi}^{-1}(y) \dd y - n\int_{0}^{1/n}F_{\xi}^{-1}(y) \dd y
\\
&\le n\int_{1-1/n}^{1}|F_{\xi}^{-1}(y)| \dd y + n\int_{0}^{1/n}|F_{\xi}^{-1}(y)| \dd y
\\
&\le c n \bigg( \int_0^1 |F_{\xi}^{-1}(y)|^p \dd y \bigg )^{1/p}n^{-1/q}
\\
& \le c n^{1/p} \big( \mathbb{E}(|\xi_1|^p) \big )^{1/p},
\end{align*}
where $F_{\xi}$ denotes the cdf of $\xi_1$, and $q\in [1,\infty ]$ is such that $p^{-1}+q^{-1}=1$. (When $p=1$, we set $q=\infty $.) This proves statement~\eqref{gc-result}. To prove the concluding part of the lemma, we choose any constant $\lambda >0$ and write the bounds
\[
\mathbb{P}(n^{-1/p}(\xi_{n:n}-\xi_{1:n})>\lambda )
\le {n^{-1/p}\over \lambda }\big( \mathbb{E}(\xi_{n:n})-\mathbb{E}(\xi_{1:n})\big)
\le {c\over \lambda }
\]
with a finite constant $c<\infty $ that does not depend on $n$ and $\lambda $.
This finishes the entire proof of Lemma~\ref{lemma-GC}.
\end{proof}

\citet[bounds~(8)]{GC1992} allow different distributions of $\xi_t$'s, and only simple though space consuming modifications of Theorem~\ref{th-1} and its proof (to be later given in this appendix) are required to accommodate this case. This is significant because it implies high robustness of convergence of $I_n$ to $1/2$ with respect to possibly varied (i.e., non-stationary)  marginal  distributions of the outputs $Y_t$.

\begin{proof}[Proof of Theorem~\ref{th-1aa}]
We first prove part~\eqref{part-1}. Using Lipschitz continuity of $h_0$, we have
\begin{align}
{1\over n^{1/p}}\sum_{t=2}^n |Y^0_{t,n}-Y^0_{t-1,n}|
&\le {K\over n^{1/p}}\sum_{t=2}^n (X_{t:n}-X_{t-1:n})
\notag
\\
&= {K\over n^{1/p}}(X_{n:n}-X_{1:n}).
\label{bound-w1}
\end{align}
Lemma~\ref{lemma-GC} with $X_t$'s instead of $\xi_t$'s gives us the statement  $n^{-1/p}(X_{n:n}-X_{1:n})=O_{\mathbb{P}}(1)$ and completes the proof of  part~\eqref{part-1}.

To prove part~\eqref{part-2}, we start with the bound
\begin{align}
{1\over n^{1/p}}\sum_{t=2}^n |Y^0_{t,n}-Y^0_{t-1,n}|
&={1\over n^{1/p}}\sum_{t=2}^n |h_0(X_{t:n})-h_0(X_{t-1:n}) |
\notag
\\
&={1\over n^{1/p}}\sum_{t=2}^n \bigg| \int_{X_{t-1:n}}^{X_{t:n}}h^*_0(x)\dd x \bigg|
\notag
\\
&\le {1\over n^{1/p}}\int_{X_{1:n}}^{X_{n:n}}|h^*_0(x)|\dd x .
\label{bound-w2}
\end{align}
Applying H\"{o}lder's inequality on the right-hand side of bound~\eqref{bound-w2} with  $\beta$ such that $\alpha^{-1}+\beta^{-1}=1$, we have
\begin{align}
{1\over n^{1/p}}\sum_{t=2}^n |Y^0_{t,n}-Y^0_{t-1,n}|
&\le {1\over n^{1/p}}(X_{n:n}-X_{1:n})^{1/ \alpha} \bigg(\int_{X_{1:n}}^{X_{n:n}}|h^*_0(x)|^{\beta}\dd x \bigg)^{1/\beta }
\notag
\\
&\le {c\over n^{1/p}}(X_{n:n}-X_{1:n})^{1/ \alpha}.
\label{bound-w3}
\end{align}
To show that the right-hand side of bound~\eqref{bound-w3} is of order $O_{\mathbb{P}}(1)$, we fix any $\lambda >0$ and write the bound
\begin{equation}
\mathbb{P}\bigg( {1\over n^{1/p}}(X_{n:n}-X_{1:n})^{1/ \alpha} > \lambda  \bigg )
\le
{1\over \lambda ^{\alpha} }
n^{-\alpha/p} \big(\mathbb{E}(\xi_{n:n})-\mathbb{E}(\xi_{1:n})\big) .
\label{bound-w4}
\end{equation}
Lemma~\ref{lemma-GC} with $X_t$'s instead of $\xi_t$'s and with $p/\alpha$ instead of $p$ shows that the right-hand side of bound~\eqref{bound-w4} can be made as small as desired by choosing a sufficiently large $\lambda $ and for all sufficiently large $n$.  This completes the proof of part~\eqref{part-2}, and that of Theorem~\ref{th-1aa} as well.
\end{proof}

\begin{proof}[Proof of Theorem~\ref{th-md}]
Since the baseline function $h_0$ is absolutely continuous on $[a_{X},b_{X}]$, there is an integrable on $[a_{X},b_{X}]$ function $h_0^*$ such that $h_0(v)- h_0(u)=\int_{u}^{v}h_0^*(x)\dd x$ for all $u,v\in [a_{X},b_{X}]$ such that $u\le v$.  Hence,
\begin{align*}
Y^0_{n,n}-Y^0_{1,n} - \big (h_0(b_{X})- h_0(a_{X})\big )
&= h_0(X_{n:n})-h_0(X_{1:n}) - \big (h_0(b_{X})- h_0(a_{X}) \big )
\\
&= -\int_{X_{n:n}}^{b_{X}}h_0^*(x)\dd x
- \int_{a_{X}}^{X_{1:n}}h_0^*(x)\dd x
\\
&= -\int_{a_{X}}^{b_{X}}\mathds{1}\{x\ge X_{n:n}\}h_0^*(x)\dd x
- \int_{a_{X}}^{b_{X}}\mathds{1}\{x<X_{1:n}\}h_0^*(x)\dd x .
\end{align*}
Consequently, for every $\lambda >0$, using Markov's inequality we have
\begin{align}
\mathbb{P}\Big( \big| Y^0_{n,n}-Y^0_{1,n} - \big (h_0(b_{X})- h_0(a_{X}) \big ) \big |>\lambda \Big )
&\le {1\over \lambda} \mathbb{E}\Big(\big| Y^0_{n,n}-Y^0_{1,n} - \big (h_0(b_{X})- h_0(a_{X}) \big ) \big | \Big )
\notag
\\
&\le {1\over \lambda}\int_{a_{X}}^{b_{X}}\mathbb{P}(X_{n:n}\le x)|h_0^*(x)|\dd x
\notag
\\
& \qquad +  {1\over \lambda}\int_{a_{X}}^{b_{X}}\mathbb{P}(X_{1:n}> x)|h_0^*(x)|\dd x .
\label{md-3}
\end{align}
Since the inputs $X_t$ are temperately dependent and $\int_{a_{X}}^{b_{X}}|h_0^*(x)|\dd x <\infty $, the Lebesgue dominated convergence theorem implies that the two integrals on the right-hand side of bound~\eqref{md-3} converge to $0$ when $n\to \infty $. This completes the proof of Theorem~\ref{th-md}.
\end{proof}

\begin{proof}[Proof of Theorem~\ref{lm-1}]
Fix any $t\in\mathbb{N}$ and let $n\ge t$. We have
\begin{align*}
F_{X_{n:n}}(x) &= \mathbb{P}(X_1\le x,\ldots,X_n\le x)
\\
&\le \mathbb{P}(X_{t}\le x,X_{2t}\le x,\ldots,X_{\lfloor n/t\rfloor t}\le x) \\
&\le \mathbb{P}(X_{t}\le x)\mathbb{P}(X_{2t}\le x,\ldots,X_{\lfloor n/t\rfloor t}\le x) + \alpha_X(t) \\
&= F(x)\mathbb{P}(X_{2t}\le x,\ldots,X_{\lfloor n/t\rfloor t}\le x) + \alpha_X(t) \\
&\le F(x)^2\mathbb{P}(X_{3t}\le x,\ldots,X_{\lfloor n/t\rfloor t}\le x) + \alpha_X(t)\big (1 + F(x)\big) \\
&\le \cdots \\
&\le F(x)^{\lfloor n/t\rfloor} + \alpha_X(t)\big(1 + F(x) + \cdots + F(x)^{\lfloor n/t\rfloor - 1}\big) \\
&= F(x)^{\lfloor n/t\rfloor} + \dfrac{\alpha_X(t)\big(1 - F(x)^{\lfloor n/t\rfloor}\big)}{1 - F(x)}.
\end{align*}
When $x<b_{X}$, we have $F(x) < 1$ and so
$$
\limsup_{n\to\infty}F_{X_{n:n}}(x) \le \dfrac{\alpha_X(t)}{1 - F(x)}.
$$
Letting $t\to\infty$, we have
\[
\limsup_{n\to\infty}F_{X_{n:n}}(x)=0
\]
and so  $F_{X_{n:n}}(x)\to 0$ when $n\to\infty$. This establishes the second part of property~\eqref{md-1}.

When $x>a_{X}$, we set $\xi_t := -X_t$ for all $t\in\mathbb{Z}$.  By the previous case, we know that $F_{\xi_{n:n}}(z) \to 0$
for all $z < b_{\xi}$. Since $X_{1:n}=-\xi_{n:n}$ and $a_{X}=-b_{\xi}$, we have  $\mathbb{P}(X_{1:n}\ge -z) \to 0$
for all $-z > a_{X}$.  This establishes the first part of property~\eqref{md-1}   and concludes the entire proof of Theorem~\ref{lm-1}.
\end{proof}

To prove Theorem~\ref{th-md-2}, we need a lemma.

\begin{lemma}\label{th-md-3new}
Let the inputs $X_t$ be strictly stationary, temperately dependent,  and satisfy the Glivenko-Cantelli property.  If the cdf $F$ and the corresponding quantile function $F^{-1}$ are continuous, then for any finite subinterval $[a,b]$ of $[a_{X},b_{X}]$, we have
\begin{equation}
\label{md-10new}
\max_{1\le t \le n+1}\big( Z_{t,n}-Z_{t-1,n} \big ) \to_{\mathbb{P}} 0
\end{equation}
when $n\to \infty $, where $Z_{0,n}:=a$, $Z_{n+1,n}:=b$, and, for all $t=1,\dots , n$,
\[
Z_{t,n}:=h_c(X_{t:n})=
\left\{
  \begin{array}{ll}
    a & \hbox{ when} \quad X_{t:n}<a, \\
    X_{t:n} & \hbox{ when} \quad a\le X_{t:n} \le b , \\
    b & \hbox{ when} \quad X_{t:n}>b .
  \end{array}
\right.
\]
\end{lemma}

\begin{proof}
Since $a_X\le a$ and the inputs $X_t$ are  temperately dependent, we have $X_{1:n} \to_{\mathbb{P}} a_X$ and so $Z_{1,n}-a \to_{\mathbb{P}} 0$. Likewise, since $b\le b_X$, we have $X_{n:n} \to_{\mathbb{P}} b_X$ and so $b-Z_{n,n} \to_{\mathbb{P}} 0$. Consequently, statement~\eqref{md-10new} holds provided that
\begin{equation}
\label{md-10new-0}
\max_{2\le t \le n}\big( Z_{t,n}-Z_{t-1,n} \big ) \to_{\mathbb{P}} 0.
\end{equation}
Note that we only need to consider those $t$'s for which $X_{t:n}> a$ and $X_{t-1:n}\le b$. These two restrictions are equivalent to $F_n^{-1}(t/n)> a$ and $F_n^{-1}((t-1)/n)\le b$, respectively. Note that $F_n^{-1}(t/n)> a$ is equivalent to $t/n> F_n(a)$, and $F_n^{-1}((t-1)/n)\le b$ is equivalent to $(t-1)/n\le F_n(b)$. Due to the Glivenko-Cantelli property, we therefore conclude that for any (small) $\delta>0$ and for all sufficiently large $n$,  all those $t$'s for which the bounds $X_{t:n}> a$ and $X_{t-1:n}\le b$ hold  are such that $(1-\delta)F(a)n\le t \le (1+\delta)F(a)n$. For typographical simplicity, we rewrite the latter bounds as $\alpha n \le t \le \beta n$, where $\alpha:=(1-\delta)F(a)$ and $\beta:=(1+\delta)F(a)$. Hence, statement~\eqref{md-10new-0} follows if
\begin{equation}
\label{md-10new1}
\max_{\alpha n \le t \le \beta n}\big( Z_{t,n}-Z_{t-1,n} \big ) \to_{\mathbb{P}} 0.
\end{equation}
Since $ Z_{t,n}-Z_{t-1,n} \le X_{t:n}-X_{t-1:n}$ for all $t=2,\dots , n$, statement~\eqref{md-10new1} follows if
\begin{equation}
\label{md-10new2}
\max_{\alpha n \le t \le \beta n}\big( X_{t:n}-X_{t-1:n} \big ) \to_{\mathbb{P}} 0.
\end{equation}
To prove the latter statement, we write
\begin{align}
\max_{\alpha n \le t \le \beta n}\big( X_{t:n}-X_{t-1:n} \big )
&=\max_{\alpha n \le t \le \beta n}\big( F_n^{-1}(t/n)-F_n^{-1}((t-1)/n) \big )
\notag
\\
&\le \max_{\alpha n \le t \le \beta n}\big( F^{-1}(t/n)-F^{-1}((t-1)/n) \big )
\notag
\\
&\qquad +\max_{\alpha n \le t \le \beta n}\big| F_n^{-1}(t/n)-F^{-1}(t/n) \big |
\notag
\\
&\qquad\qquad  +\max_{\alpha n \le t \le \beta n}\big| F_n^{-1}((t-1)/n)-F^{-1}((t-1)/n) \big |.
\label{md-10new3}
\end{align}
The first maximum on the right-hand side of bound~\eqref{md-10new3} converges to $0$ because $F^{-1}$ is continuous on $(0,1)$ and thus uniformly continuous on every closed subinterval of $(0,1)$. As to the second and third maxima on the right-hand side of bound~\eqref{md-10new3}, they converge to $0$ in probability because
\begin{equation}
\label{md-10new2u}
\Gamma_n:=\sup_{t\in [t_0,t_1]}\big| F_n^{-1}(t)-F^{-1}(t) \big | \to_{\mathbb{P}} 0
\end{equation}
for every closed interval $[t_0,t_1] \subset (0,1)$, because the Glivenko-Cantelli property holds. To show that the just noted implication is true, we proceed as follows.

Statement~\eqref{md-10new2u} means that, for any fixed $\gamma >0$, the probability of the event $\Gamma_n\le \gamma $ converges to $1$ when $n\to \infty$. This event
has at least the same, if not larger, probability as the event
\begin{equation}
\label{stat-9}
F^{-1}(t)- \gamma < F_n^{-1}(t)\le F^{-1}(t)+ \gamma
\quad \textrm{for all} \quad  t\in [t_0,t_1] ,
\end{equation}
which is equivalent to
\[
F_n(F^{-1}(t)- \gamma) < t \le F_n(F^{-1}(t)+ \gamma)
\quad \textrm{for all} \quad  t\in [t_0,t_1].
\]
The latter event has at least the same, if not larger, probability as the event
\begin{equation}\label{stat-10v}
F(F^{-1}(t)- \gamma) +\Vert F_n-F\Vert < t \le F(F^{-1}(t)+ \gamma) -\Vert F_n-F\Vert
\quad \textrm{for all} \quad  t\in [t_0,t_1].
\end{equation}
Since $t=F(F^{-1}(t))$, event~\eqref{stat-10v} has at least the same, if not larger, probability as the event
\begin{equation}\label{stat-10}
-\Delta_1( \gamma) +\Vert F_n-F\Vert < 0 \le \Delta_2( \gamma)-\Vert F_n-F\Vert ,
\end{equation}
where
\[
\Delta_1( \gamma):= \inf_{t\in [t_0,t_1]} \big( F(F^{-1}(t)) - F(F^{-1}(t)- \gamma)\big)
\]
and
\[
\Delta_2( \gamma):= \inf_{t\in [t_0,t_1]} \big( F(F^{-1}(t)+ \gamma) - F^{-1}(F^{-1}(t))\big).
\]
Since the cdf $F$ is strictly increasing (because we have assumed that $F^{-1}$ is continuous), the quantities $\Delta_1( \gamma)$ and $\Delta_2( \gamma)$ are (strictly) positive for every $\gamma>0$. We therefore conclude that statement~\eqref{stat-10} holds with as large a probability as desired, provided that $n$ is sufficiently large. This, in turn, implies that event~\eqref{stat-9} can be made as close to $1$ as desired, provided that $n$ is sufficiently large. The proof of Lemma~\ref{th-md-3new} is finished.
\end{proof}

\begin{proof}[Proof of Theorem~\ref{th-md-2}]
Since the baseline function $h_0$ is absolutely continuous on $[a_{X},b_{X}]$ and its Radon-Nikodym derivative $h^*_0$ vanishes outside the interval $[a,b]$, we have
\begin{align*}
|Y^0_{t,n}-Y^0_{t-1,n} |
&=\bigg| \int_{X_{t:n}}^{X_{t-1:n}}h^*_0(x)\dd x \bigg|
\\
&=\bigg| \int_{[X_{t-1:n},X_{t:n}]\cap [a,b]}h^*_0(x)\dd x \bigg|
= |h^*_0(\xi_{t,n})|(Z_{t,n}-Z_{t-1,n}),
\end{align*}
where, due to the mean-value theorem, the right-most equation holds for some $\xi_{t,n}\in [Z_{t-1,n},Z_{t,n}]$ with $Z_{t,n}$'s defined in Lemma~\ref{th-md-3new}. Consequently,
\begin{align*}
\Theta_n
:=&\sum_{t=2}^n |Y^0_{t,n}-Y^0_{t-1,n}|- \int_{a}^{b}|h^*_0(x)|\dd x
\\
=& \sum_{t=2}^n |h^*_0(\xi_{t,n})|(Z_{t,n}-Z_{t-1,n}) - \int_{a}^{b}|h^*_0(x)|\dd x.
\end{align*}
Obviously, $Z_{t,n}\in [a,b]$ for all $t=2,\dots , n$. We also have $Z_{0,n}=a$ and $Z_{n+1,n}=b$. By Lemma~\ref{th-md-3new},
\begin{equation}
\label{md-10d}
\max_{1\le t \le n+1}(Z_{t,n}-Z_{t-1,n}) \to_{\mathbb{P}} 0.
\end{equation}
Furthermore,
\begin{align}
\Theta_n
&= \sum_{t=1}^{n+1} |h^*_0(\xi_{t,n})|(Z_{t,n}-Z_{t-1,n}) - \int_{a}^{b}|h^*_0(x)|\dd x
- \sum_{t\in\{1,n+1\}} |h^*_0(\xi_{t,n})|(Z_{t,n}-Z_{t-1,n})
\notag
\\
&= \sum_{t=1}^{n+1} |h^*_0(\xi_{t,n})|(Z_{t,n}-Z_{t-1,n}) - \int_{a}^{b}|h^*_0(x)|\dd x
+ o_{\mathbb{P}}(1),
\label{md-10dd}
\end{align}
where the last equation holds when $n\to \infty $ because the function $h^*_0$ is bounded, and $ Z_{1,n}-a \to_{\mathbb{P}} 0$ and $b-Z_{n,n} \to_{\mathbb{P}} 0$
when $n\to \infty $, which we verified at the beginning of the proof of Lemma~\ref{th-md-3new}.  Hence, equation~\eqref{md-10dd} holds, and in order to prove $\Theta_n \to_{\mathbb{P}} 0$, we need to show
\begin{equation}
\label{md-30}
\Theta_n^*
:=\sum_{t=1}^{n+1} |h^*_0(\xi_{t,n})|(Z_{t,n}-Z_{t-1,n}) - \int_{a}^{b}|h^*_0(x)|\dd x \to_{\mathbb{P}} 0.
\end{equation}
In other words, we need to show that, for every $\gamma>0$,
\begin{equation}
\label{md-31}
\mathbb{P}\big(|\Theta_n^*|\ge \gamma\big ) \to 0
\end{equation}
when $n\to \infty $. For this, we first rewrite statement~\eqref{md-10d} explicitly: for every $\lambda>0$,
\begin{equation}
\label{md-32}
\mathbb{P}\Big(\max_{1\le t \le n+1}(Z_{t,n}-Z_{t-1,n})\ge \lambda\Big ) \to 0
\end{equation}
when $n\to \infty $. Hence, statement~\eqref{md-31} follows if, for any $\gamma>0$, we can find $\lambda>0$ such that
\begin{equation}
\label{md-33}
\mathbb{P}\Big(|\Theta_n^*|\ge \gamma, \max_{1\le t \le n+1}(Z_{t,n}-Z_{t-1,n})< \lambda\Big ) \to 0
\end{equation}
when $n\to \infty $. We now recall the very basic definition of Riemann integral, according to which, for any $\gamma>0$, we can find $\lambda >0$ such that
\[
\bigg| \sum_{t=1}^{n+1} |h^*_0(\zeta_{t,n})|(z_{t,n}-z_{t-1,n}) - \int_{a}^{b}|h^*_0(x)|\dd x \bigg| < \gamma
\]
whenever
\begin{equation}
\label{md-34}
\max_{1\le i \le n+1}(z_{t,n}-z_{t-1,n})<\lambda ,
\end{equation}
where $z_{0,n}:=a$, $z_{n+1,n}:=b$, $z_{t-1,n}\le z_{t,n}$ for $t=1,\dots , n+1$, and $\zeta_{t,n}\in [z_{t-1,n},z_{t,n}]$. Hence, with the same $\lambda>0$ as in statement~\eqref{md-34}, probability~\eqref{md-33} is equal to $0$.  This establishes statement~\eqref{md-31} and finishes the proof of Theorem~\ref{th-md-2}.
\end{proof}

\begin{proof}[Proof of Corollary~\ref{cor-md-2}]
By Theorems~\ref{th-md} and \ref{th-md-2}, we have
\[
I^0_n\to_{\mathbb{P}} {1\over 2}\bigg( 1+
{ h_0(b_{X})- h_0(a_{X})
\over \int_{a}^{b}|h^*_0(x)|\dd x} \bigg )
\]
when $n\to \infty $. Furthermore, we have
\[
h_0(b_{X})- h_0(a_{X})= h_0(b)- h_0(a)= \int_{a}^{b}h^*_0(x)\dd x
\]
because the Radon-Nikodym derivative $h^*_0$ of $h_0$ vanishes outside the interval $[a,b]$
and the positive part $z_{+}$ of every real number $z\in \mathbb{R}$ can be written as $(|z|+z)/2$. This concludes the proof of Corollary~\ref{cor-md-2}.
\end{proof}

\begin{proof}[Proof of Theorem~\ref{th-1}]
Since $z_{+}=(|z|+z)/2$ for every real number $z\in \mathbb{R}$, we have
\[
I_n={1\over 2}\bigg( 1+
{Y_{n,n}-Y_{1,n}\over \sum_{t=2}^n |Y_{t,n}-Y_{t-1,n}|} \bigg ).
\]
Lemma~\ref{lemma-GC} with $Y_t$'s instead of $\xi_t$'s  says that $n^{-1/p}(Y_{n:n}-Y_{1:n})=O_{\mathbb{P}}(1)$. Since the system is out of $p$-reasonable order, we have $n^{-1/p}\sum_{t=2}^n |Y_{t,n}-Y_{t-1,n}|\to_{\mathbb{P}} \infty $ when $n\to \infty $ and thus $I_n\to_{\mathbb{P}} 1/2$. This concludes the proof of Theorem~\ref{th-1}.
\end{proof}

To prove Theorem~\ref{tf-1}, we need a formula for $B_{n,p}$ analogous to equation~\eqref{bnp0-a}.

\begin{lemma}\label{conco-e}
The concomitants $Y_{1,n}, \dots ,  Y_{n,n}$ of the outputs $Y_t=h(X_t,\boldsymbol{\varepsilon}_t)$, $t=1,\dots, n$, with respect to the inputs $X_1,\dots , X_n$ are given by
\[
Y_{t,n}= h(X_{t:n},\boldsymbol{\varepsilon}_{t,n}),
\]
where $\boldsymbol{\varepsilon}_{1,n},\dots , \boldsymbol{\varepsilon}_{n,n}$ are the concomitants of the anomalies $\boldsymbol{\varepsilon}_1,\dots , \boldsymbol{\varepsilon}_n$  with respect to $X_1,\dots , X_n$, that is,
\[
\boldsymbol{\varepsilon}_{t,n}=\sum_{s=1}^n \boldsymbol{\varepsilon}_{s}\mathds{1}\{X_{s}=X_{t:n}\} .
\]
Consequently,
\begin{equation}\label{form-aa}
B_{n,p}={1\over n^{1/p}}\sum_{t=2}^n \big |h(X_{t:n},\boldsymbol{\varepsilon}_{t,n})-h(X_{t-1:n},\boldsymbol{\varepsilon}_{t-1,n})\big |.
\end{equation}
\end{lemma}

\begin{proof}
Since the cdf $F$ of each input $X_t$ is continuous, we can assume without loss of generality that all the inputs $X_1,\dots , X_n$ are unequal. Hence, we can write the equation
\[
\boldsymbol{\varepsilon}_{t}=\sum_{s=1}^n \boldsymbol{\varepsilon}_{s}\mathds{1}\{X_{s}=X_{t}\} .
\]
This implies that the concomitants of the outputs $Y_1,\dots , Y_n$ with respect to the inputs  $X_1,\dots , X_n$ can be expressed as follows:
\begin{align*}
Y_{t,n}
&=\sum_{s=1}^n Y_s\mathds{1}\{X_{s}=X_{t:n}\}
\\
&=\sum_{s=1}^n h(X_s,\boldsymbol{\varepsilon}_{s})\mathds{1}\{X_{s}=X_{t:n}\}
\\
&=\sum_{s=1}^n h(X_{s:n},\boldsymbol{\varepsilon}_{s,n})\mathds{1}\{X_{s:n}=X_{t:n}\}
\\
&= h(X_{t:n},\boldsymbol{\varepsilon}_{t,n}).
\end{align*}
This establishes equation~\eqref{form-aa} and concludes the proof of Lemma~\ref{conco-e}.
\end{proof}

\begin{note}
\citet[Section~4]{KD1990} use the notation $\boldsymbol{\varepsilon}_{[t]}$ instead of $\boldsymbol{\varepsilon}_{t,n}$, in which case the equation $Y_{t,n}= h(X_{t:n},\boldsymbol{\varepsilon}_{t,n})$ turns into
$Y_{t,n}= h(X_{t:n},\boldsymbol{\varepsilon}_{[t]})$. We prefer the notation $\boldsymbol{\varepsilon}_{t,n}$ as it reminds us that the anomaly concomitants depend on the sample size $n$.
\end{note}

\begin{proof}[Proof of Theorem~\ref{tf-1}]
Since the anomaly-free outputs $Y_t^0$ are in $p$-reasonable order with respect to the inputs $X_t$ for some $p>0$, we have  $B^0_{n,p}=O_{\mathbb{P}}(1)$. By Lemma~\ref{conco-e}, we have
\begin{align*}
B_{n,p}
& ={1\over n^{1/p}}\sum_{t=2}^n |h(X_{t:n},0,\epsilon_{t,n})-h(X_{t-1:n},0,\epsilon_{t-1,n})|
\\
&={1\over n^{1/p}}\sum_{t=2}^n |h_0(X_{t:n})+\epsilon_{t,n}-h_0(X_{t-1:n})-\epsilon_{t-1,n}|
\\
&\ge {1\over n^{1/p}}\sum_{t=2}^n |\epsilon_{t,n}-\epsilon_{t-1,n}|
-{1\over n^{1/p}}\sum_{t=2}^n |h_0(X_{t:n})-h_0(X_{t-1:n})|
\\
&= {1\over n^{1/p}}\sum_{t=2}^n |\epsilon_{t,n}-\epsilon_{t-1,n}|
+O_{\mathbb{P}}(1).
\end{align*}
Hence, if the output anomalies $\epsilon_t$ are out of $p$-reasonable order with respect to the inputs,  meaning that
\[
{1\over n^{1/p}}\sum_{t=2}^n |\epsilon_{t,n}-\epsilon_{t-1,n}|\to _{\mathbb{P}} \infty ,
\]
then $B_{n,p}\to _{\mathbb{P}} \infty $ when $n\to \infty $.

Conversely, if $B_{n,p}\to _{\mathbb{P}} \infty $ when $n\to \infty $, then the bound
\begin{align*}
B_{n,p}
&\le {1\over n^{1/p}}\sum_{t=2}^n |\epsilon_{t,n}-\epsilon_{t-1,n}|
+{1\over n^{1/p}}\sum_{t=2}^n |h_0(X_{t:n})-h_0(X_{t-1:n})|
\\
&= {1\over n^{1/p}}\sum_{t=2}^n |\epsilon_{t,n}-\epsilon_{t-1,n}|
+O_{\mathbb{P}}(1)
\end{align*}
implies that the output anomalies $\epsilon_t$ are out of $p$-reasonable order with respect to the inputs. This concludes the proof of Theorem~\ref{tf-1}.
\end{proof}

\begin{lemma}\label{positive-moment}
If random variables $\xi_t$ are iid and have finite first moments, then
\begin{equation}\label{xi-0}
{1\over n^{1/p}}\sum_{t=2}^n |\xi_{t}-\xi_{t-1}|\to _{\mathbb{P}} \infty
\end{equation}
whenever the distribution of $\xi_1$ is non-degenerate.
\end{lemma}

\begin{proof}
Since the summands $|\xi_{t}-\xi_{t-1}|$, $t=2,3,\dots $, are $1$-dependent, splitting the sum into the sums with respect to even and odd $t$'s yields statement~\eqref{xi-0} for every $p>1$ if the moment  $\mathbb{E}(|\xi_{2}-\xi_{1}|)$ is (strictly) positive. Since $\xi_{2}$ and $\xi_{1}$ are iid, the aforementioned moment is positive whenever the distribution of $\xi_1$ is non-degenerate.
\end{proof}

\begin{proof}[Proof of Theorem~\ref{tf-2}]
Since the baseline function $h_0$ is Lipshitz continuous, Theorem~\ref{th-1aa} implies that the anomaly-free outputs $Y_t^0$ are in $p$-reasonable order with respect to the inputs $X_t$. Consequently,  $B^0_{n,p}=O_{\mathbb{P}}(1)$. By Lemma~\ref{conco-e}, we have
\[
B_{n,p}={1\over n^{1/p}}\sum_{t=2}^n |h(X_{t:n},\delta_{t,n},0)-h(X_{t-1:n},\delta_{t-1,n},0)| .
\]
Using Lipschitz continuity of $h_0$ and also Lemma~\ref{lemma-GC} with $X_t$ instead of $\xi_t$, we have
\begin{align*}
B_{n,p}
&={1\over n^{1/p}}\sum_{t=2}^n |h_0(X_{t:n}+\delta_{t,n})-h_0(X_{t:n}+\delta_{t-1,n})
+h_0(X_{t:n}+\delta_{t-1,n})-h_0(X_{t-1:n}+\delta_{t-1,n})|
\\
&\ge {1\over n^{1/p}}\sum_{t=2}^n |h_0(X_{t:n}+\delta_{t,n})-h_0(X_{t:n}+\delta_{t-1,n})|
-{K\over n^{1/p}}\sum_{t=2}^n (X_{t:n}-X_{t-1:n})
\\
&= {1\over n^{1/p}}\sum_{t=2}^n |h_0(X_{t:n}+\delta_{t,n})-h_0(X_{t:n}+\delta_{t-1,n})|
-{K\over n^{1/p}}(X_{n:n}-X_{1:n})
\\
&= {1\over n^{1/p}}\sum_{t=2}^n |h_0(X_{t:n}+\delta_{t,n})-h_0(X_{t:n}+\delta_{t-1,n})|
+O_{\mathbb{P}}(1).
\end{align*}
Hence, with the notation
\[
\xi_t= |h_0(X_{t}+\delta_{t})-h_0(X_{t}+\delta_{t-1})|,
\]
we are left to check the statement
\begin{equation}\label{delta}
{1\over n^{1/p}}\sum_{t=2}^n \xi_t
\to _{\mathbb{P}} \infty.
\end{equation}
That is, we need to show that for every $\lambda <\infty $, we have
\[
\mathbb{P}\bigg({1\over n^{1/p}}\sum_{t=2}^n \xi_t \le \lambda  \bigg) \to 0
\]
when $n\to \infty $. Since the moment
$\mu:=\mathbb{E}(\xi_1)=\mathbb{E}(\xi_t) $ is strictly positive by assumption~\eqref{part-12}, for all sufficiently large $n$ we have
\begin{align}
\mathbb{P}\bigg({1\over n^{1/p}}\sum_{t=2}^n \xi_t \le \lambda  \bigg)
&\le \mathbb{P}\bigg(-\bigg|{1\over n^{1/p}}\sum_{t=2}^n (\xi_t-\mu)\bigg|+n^{1-1/p}\mu \le \lambda  \bigg)
\notag
\\
&= \mathbb{P}\bigg(\bigg|{1\over n^{1/p}}\sum_{t=2}^n (\xi_t-\mu)\bigg|\ge n^{1-1/p}\mu - \lambda  \bigg)
\notag
\\
&\le {1\over  (n^{1-1/p}\mu - \lambda )^2}\mathbb{E}\bigg(\bigg|{1\over n^{1/p}}\sum_{t=2}^n (\xi_t-\mu)\bigg|^2\bigg)
\notag
\\
&\le {c\over  n^2 }\mathbb{E}\bigg(\bigg|\sum_{t=2}^n (\xi_t-\mu)\bigg|^2\bigg).
\label{delta-2}
\end{align}
The use of the second moment for bounding the probability was prudent because $\xi_t$'s are bounded by a constant, which follows because the baseline function $h_0$ is bounded. Hence, our task becomes to prove
\begin{equation}\label{delta-3}
 {1\over  n^2 }\mathbb{E}\bigg(\bigg|\sum_{t=2}^n (\xi_t-\mu)\bigg|^2\bigg) \to 0.
\end{equation}
Since the sequence $\xi_t$ is strictly stationary, by \citet[Corollary~1.2, p.~10]{R2017} we have
\begin{align}
\mathbb{E}\bigg(\bigg|\sum_{t=2}^n (\xi_t-\mu)\bigg|^2\bigg)
&= \sum_{t=2}^n \mathrm{Var}(\xi_t) + 2 \sum_{2\le s < t \le n} \mathrm{Cov}(\xi_s,\xi_t)
\notag
\\
&=(n-1) \mathrm{Var}(\xi_0) + 2 \sum_{t=1}^{n-2} (n-1-t)\mathrm{Cov}(\xi_0,\xi_t)
\notag
\\
&\le c(n-1) \alpha_{\xi}(0) + c(n-2)\alpha_{\xi}(1) + c \sum_{t=2}^{n-2} (n-1-t)\alpha_{\xi}(t),
\label{bound-alpha}
\end{align}
where $c$ is a finite constant that depends on $h_0$.
Next, for every $t\ge 2$, we have
\begin{align}
\alpha_{\xi}(t)
&= \sup \Big\{
\big |\mathbb{P}(A\cap B) - \mathbb{P}(A)\mathbb{P}(B)\big |
: A\in\sigma(\xi_u, u \le 0), ~ B\in\sigma(\xi_v, v\ge t ) \Big \}
\notag
\\
&\le \sup \Big\{
\big |\mathbb{P}(A\cap B) - \mathbb{P}(A)\mathbb{P}(B)\big |
: A\in\sigma(X_u, \delta_u,\delta_{u-1}, u \le 0), ~ B\in\sigma(X_v, \delta_v,\delta_{v-1},  v\ge t ) \Big \}
\notag
\\
&\le \sup \Big\{
\big |\mathbb{P}(A\cap B) - \mathbb{P}(A)\mathbb{P}(B)\big |
: A\in\sigma(X_u, u \le 0), ~ B\in\sigma(X_v,  v\ge t ) \Big \},
\label{rhs-0}
\end{align}
where the last inequality holds because
\begin{align*}
\big |\mathbb{P}(A\cap C \cap B\cap D) - \mathbb{P}(A\cap C )\mathbb{P}(B\cap D )\big |
&=\big |\mathbb{P}(A\cap B) - \mathbb{P}(A)\mathbb{P}(B)\big | \mathbb{P}(C\cap D )
\\
&\le \big |\mathbb{P}(A\cap B) - \mathbb{P}(A)\mathbb{P}(B)\big |
\end{align*}
for all $A\in\sigma(X_u, u \le 0)$, $B\in\sigma(X_v,  v\ge t )$, and all $C\in\sigma(\delta_u, u \le 0)$,  $D\in\sigma(\delta_v,  v\ge t-1 )$, upon recalling that $A$ and $B$ are independent of $C$ and $D$, and also $C$ and $D$ are independent of each other as long as $t\ge 2$. Note that the right-hand side of bound~\eqref{rhs-0} is equal to $\alpha_{X}(t)$, and so we have the bound $\alpha_{\xi}(t) \le \alpha_{X}(t)$ for all $t\ge 2$. This result together with bound~\eqref{bound-alpha} imply
\begin{align*}
 {1\over  n^2 }\mathbb{E}\bigg(\bigg|\sum_{t=2}^n (\xi_t-\mu)\bigg|^2\bigg)
&\le {c\over n} \alpha_{X}(0) +  {c\over  n }\alpha_{X}(1) + {c\over n^2} \sum_{t=2}^{n-2} (n-1-t)\alpha_{X}(t)
\end{align*}
with a finite constant $c$. The right-hand side of the latter bound converges to $0$ when $n\to \infty $, provided that $\alpha_{X}(t)$ converges to $0$ when $t\to \infty $, which is true because the inputs $X_t$ are $\alpha $-mixing. This concludes the proof of Theorem~\ref{tf-2}.
\end{proof}

\begin{proof}[Proof of Theorem~\ref{tf-3}]
Since the baseline function $h_0$ is Lipshitz continuous, Theorem~\ref{th-1aa} implies that the anomaly-free outputs $Y_t^0$ are in $p$-reasonable order with respect to the inputs $X_t$. Consequently,  $B^0_{n,p}=O_{\mathbb{P}}(1)$. By Lemma~\ref{conco-e}, we have
\[
B_{n,p}={1\over n^{1/p}}\sum_{t=2}^n |h(X_{t:n},\delta_{t,n},\epsilon_{t,n})-h(X_{t-1:n},\delta_{t-1,n},\epsilon_{t-1,n})|.
\]
Proceeding analogously as in the proof of Theorem~\ref{tf-2}, we have
\[
B_{n,p} \ge {1\over n^{1/p}}\sum_{t=2}^n \zeta_t +O_{\mathbb{P}}(1)
\]
when $n\to \infty $, where
\[
\zeta_t= |h_0(X_{t}+\delta_{t})+\epsilon_{t}-h_0(X_{t}+\delta_{t-1})-\epsilon_{t-1}|.
\]
The rest is analogous to the proof of Theorem~\ref{tf-2} starting with statement~\eqref{delta}, and we thus skip the details. This finishes the proof of Theorem~\ref{tf-3}.
\end{proof}

\end{document}